\documentclass[12pt,reqno]{amsart}
\usepackage{amssymb}
\usepackage{cases}
\usepackage{bbm}
\usepackage{mathrsfs}
\usepackage{graphicx}
\usepackage{amsfonts}
\usepackage{enumerate}
\usepackage{hyperref}
\usepackage{bm}
 \usepackage[dvipsnames]{xcolor}

\usepackage{tikz}
\usetikzlibrary{calc}
\usepackage{caption}

\usepackage{amssymb, amsmath, amsfonts, latexsym}
\usepackage{enumerate}
\usepackage{color}

\usetikzlibrary{shapes.geometric, calc, decorations.pathmorphing}

\usepackage{bbm}

\newtheorem{theorem}{Theorem}

\newtheorem{lemma}{Lemma}
\newtheorem{proposition}{Proposition}

\newtheorem{remark}{Remark}

\newtheorem{definition}{Definition}

\numberwithin{equation}{section}

\date{}

\def\<{\langle}
\def\>{\rangle}

\def\d"{^{\prime\prime}}

\usepackage{geometry}
\begin{document}

 \title[{\tiny Kinetic Langevin equations driven by Lévy processes}]
 {{\small On some topological and spectral properties of  kinetic Langevin processes driven by  Lévy noises}}

\author[T. Batisse]{\textbf{\quad {T.} Batisse$^{\dag}$    }}
\address{{\bf T. Batisse}. Universit\'e Clermont Auvergne, CNRS, LMBP, F-63000 CLERMONT-FERRAND, FRANCE}
 \email{thomas.batisse@uca.fr}

\author[A. Guillin]{\textbf{\quad {A.} Guillin$^{\dag}$    }}
\address{{\bf A. Guillin}. Universit\'e Clermont Auvergne, CNRS, LMBP, F-63000 CLERMONT-FERRAND, FRANCE}
 \email{arnaud.guillin@uca.fr}

\author[B. Nectoux]{\textbf{\quad B. Nectoux$^{\dag}$  }}
\address{{\bf B. Nectoux}.Universit\'e Clermont Auvergne, CNRS, LMBP, F-63000 CLERMONT-FERRAND, FRANCE}
 \email{boris.nectoux@uca.fr}

\author[L. Wu]{\textbf{\quad L. Wu$^{\dag}$ \, \, }}
\address{{\bf L. Wu}. Universit\'e Clermont Auvergne, CNRS, LMBP, F-63000  CLERMONT-FERRAND, FRANCE, and,
Institute for Advanced Study in Mathematics, Harbin Institute of Technology, Harbin 150001, China}
\email{Li-Ming.Wu@uca.fr}

\begin{abstract}
We investigate several fundamental properties of  kinetic Langevin processes in $\mathbb{R}^{2d}$, defined as solutions to the following system:
$$dx_t = v_t \, dt, \qquad dv_t = \mathbf{B}(x_t, v_t) \, dt + dL_t$$
where $(L_t, t \ge 0)$ is a pure-jump Lévy process. Our analysis covers both the original process and its killed counterpart, where killing occurs upon exiting domains of the form $\mathscr{D} = \mathscr{O} \times \mathbb{R}^d$ for an arbitrary open  set $\mathscr{O} \subset \mathbb{R}^d$. Operating within a low-regularity framework - where the drift $\mathbf{B}$ is not assumed to be continuous - we establish key structural and spectral properties for both the associated non-killed and killed semigroups. These include: the strong Feller property, weak continuity of trajectories with respect to initial conditions, topological irreducibility and the existence of a spectral gap. Furthermore, we prove, in this low-regularity framework,  the existence and uniqueness of a weak solution when the driving noise is a rotationally invariant $\alpha$-stable process, when $\alpha \in (1,2)$. For this specific case, we show  that the aforementioned properties hold and further establish the existence of densities within certain $L^m$ spaces as well as the Feller $C_0(\mathbb R^{2d})$-semigroup  property.
Finally, we address the existence and uniqueness of stationary and quasi-stationary distributions, proving exponential ergodicity for the non-killed process and exponential convergence to the quasi-stationary limit for the conditioned process. We show that these results extend to every   $\alpha \in (0,1]$ when the drift is smooth. 
\end{abstract}
\maketitle
{\small
\tableofcontents
}



\section{Introduction}

\subsection{Setting}

\subsubsection{Kinetic Langevin process driven by jump processes} 
Let  $(\Omega, \mathcal{F}, (\mathcal{F}_t)_{t\ge 0}, \mathbb{P})$ be a  filtered probability space  satisfying the usual conditions. 
In this work, we consider the position-velocity process  $(X_t=(x_t,v_t),t\ge 0)$   defined as the  solution to     the kinetic Langevin equation in $\mathbb R^d \times \mathbb R^d $:
\begin{equation}\label{eq.Lan}
 dx_t=v_tdt, \ dv_t=\mathbf B(x_t,v_t)dt   +   dL_t,
 \end{equation} 
 driven by a pure jump Lévy process $L=(L_t,t\ge 0)$ (i.e. a   Lévy process with Lévy triplet $(0,0,\nu)$ with  $\int_{\mathbb{R}^d} (|z|^2 \wedge 1) \nu(dz) < \infty$),  and  where $\mathbf B:\mathbb R^{2d} \to \mathbb R^d$ is a  vector field.   Some   results in this work require $(L_t,t\ge 0)$ to be a rotationally invariant $\alpha$-stable process where we recall that for $\alpha \in (0,2)$ (the stability index), a 
  rotationally invariant $\alpha$-stable process (RI$\alpha$S for short)  is a   Lévy process,  denoted by $(L^\alpha_t,t\ge 0)$,  with characteristic exponent
 $$\Psi_\alpha(\xi)=    |\xi|^\alpha,$$
  see e.g.~\cite[Examples 1.3.14 and 3.3.8,  and Section 4.3.4]{applebaum2009levy}. Its  Lévy measure  is   $\nu_\alpha (dz)=c_{\alpha,d}/|z|^{d+\alpha}$ for some constant $c_{\alpha,d}>0$.

 From a physical perspective, $(X_t=(x_t,v_t),t\ge 0)$ describes a particle subject to external   forces, linear dissipation, and impulsive, discontinuous random forces generated by the environment. Each jump of the Lévy process corresponds to a sudden velocity change induced by  violent collisions or interactions, providing a natural interpretation as a Langevin equation with discontinuous forcing. 
Kinetic Langevin-type equations play a central role in mathematical physics (e.g. in statistical physics and molecular dynamics) and the natural sciences, as they describe the evolution of matter in motion. In its classical formulation, the interaction with the heat bath is modeled by linear friction and Gaussian noise.
The motivation comes from the physical picture in which the forcing results from adding up many independent random pulses with the same law. Provided the variance is finite, the central limit theorem shows that this sum becomes asymptotically Gaussian as the number of pulses increases without bound. Brownian motion is, however,   only a special case within the broader class of stochastic processes that serve as natural models of random noise.

The central limit theorem also extends to i.i.d.\ random variables with infinite variance. 
In this setting, normalized sums may converge to a stable law in the sense of 
Lévy--Khintchine. Such stable distributions are characterized by a parameter  
$\alpha\in(0,2)$, called the L\'evy stability index, which determines moment finiteness: 
moments of order $p$ are finite only when $p<\alpha$, while higher-order moments diverge. 
The Gaussian distribution is the special case $\alpha=2$, for which all moments are finite, 
see~\cite{yanovsky2000levy} and references therein. 
In many applications, substituting Brownian forcing with Lévy-stable noise in kinetic Langevin equations better reflects the empirically observed dynamics.
Several kinetic models featuring fractional derivatives in space and time have emerged in recent years and are now widely studied as a way to capture anomalous diffusion and nonstandard relaxation behavior~\cite{saichev1997fractional,metzler2000random,chechkin2002stationary,eliazar}.   
 Further examples of stable distributions in physical applications can be found in the review articles \cite{coleman1992fractal,bouchaud,brockmann2002levy,chechkin2006fundamentals} and the references therein.

\subsubsection{Non-killed and killed semigroup}  
Throughout this work, we write $(X_t(\mathsf x),t\ge 0)$ for the solution to~\eqref{eq.Lan} when $X_0=\mathsf x\in \mathbb R^{2d}$ (or in short, on some occasions, $X(\mathsf x)$).   
 Let us denote by  $(P_t,t\ge 0)$ the non-killed semigroup of~\eqref{eq.Lan}, i.e.
  $$ P_tf(\mathsf x)= \mathbb E_{\mathsf x}[f(X_t)], \ f \in bB(\mathbb R^{2d}), \mathsf x\in \mathbb R^{2d},
 $$   
 where $bB(\mathbb R^{2d})$ is the space of   all bounded and Borel measurable functions $f: \mathbb R^{2d} \to  \mathbb R$. 
 We also recall that    $(P_t,t\ge 0)$ is a Feller $C_0(\mathbb R^{2d})$-semigroup  if by definition, for all $t\ge 0$, $P_t(C_0(\mathbb{R}^{2d})) \subset C_0(\mathbb{R}^{2d})$ and $\lim_{t \to 0} \|P_t f - f\|_\infty = 0$, $ \forall f \in C_0(\mathbb{R}^{2d})$, where $C_0(\mathbb R^{2d})= \{f: \mathbb R^{2d}\to \mathbb R   \text{ continuous and } \lim_{|\mathsf x|\to +\infty} f(\mathsf x)=0\}$.
In all this work, $\mathscr O$ is  a nonempty  open subset of     $\mathbb R^d$, $d\ge 1$.   
We set 
\begin{equation}\label{eq.DD}
\mathscr D=\mathscr O\times \mathbb R^d \subset \mathbb R^{2d}.
\end{equation}
In all this work, we make no assumption on the regularity of the boundary of $\mathscr O$.  
 Let $\sigma_{\mathscr D}$ be the first time the process $(X_t,t\ge 0)$ exits $\mathscr D$, namely
 $$\sigma_{\mathscr D}:=\inf\{t\ge 0: X_t\notin \mathscr D\}.$$
Let us also denote by   $(  P^{\mathscr D}_t,t\ge 0)$ the   killed (outside of $\mathscr D$) semigroup of~\eqref{eq.Lan}:
\begin{equation}\label{eq.PTD}
 P_t^{\mathscr D}f(\mathsf x)= \mathbb E_{\mathsf x}[f(X_t)\mathbf 1_{t<\sigma_{\mathscr D}}], \ f \in b B(\mathscr D), \mathsf x \in \mathscr D.
\end{equation}   
The choice of such domains $\mathscr{D}$ is   motivated by the study of the metastable behavior of kinetic processes. In the small-noise regime, the presence of energetic barriers causes the position process $(x_t, t \ge 0)$ to remain confined within specific regions $\mathscr{O}$ of the state space $\mathbb{R}^d$ for long  periods of time before eventually escaping (see, e.g., \cite{di-gesu-lelievre-le-peutrec-nectoux-17, guillinqsd, ramilarxiv2}). Consequently, the metastable sets associated with kinetic processes are naturally represented by domains of the form $\mathscr{D}$ as defined in \eqref{eq.DD}. Furthermore, within a metastable set the process reaches a local equilibrium (a local ergodic measure) long before it exits~\cite{DLLN,IHPLLN}. Quasi-stationary distributions (see Definition~\ref{de.QSD}) that we study below in Section~\ref{sec.QSD-Sd} for the process \eqref{eq.Lan}  play a central role in statistical physics, since they describe these local equilibria attained by the process inside metastable sets $\mathscr{D}$ prior to exiting $\mathscr{D}$~\cite[Section 6.3]{lelievre2016partial}.

 We finally recall     some standard  definitions we will need for    a non-negative sub-Markovian kernel $Q_t(\mathsf x,d\mathsf y)$ over a nonempty open subset $\mathscr M$ of $\mathbb R^{2d}$:
 \begin{enumerate}
 \item[-]   $(Q_t,t> 0)$ has the strong Feller  property  over $\mathscr M$ if for  every $t>0$, $Q_t(b B(\mathscr M))\subset C_b(\mathscr M)$ (the space of bounded continuous functions $f: \mathscr M \to  \mathbb R$). 
  \vspace{0.1cm}
  \item[-]  $(Q_t,t> 0)$ is topologically irreducible over $\mathscr M$ if for all $t>0$, all $\mathsf x\in \mathscr M$ and all nonempty open subset $\mathscr V$  of $\mathscr M$, $Q_t (\mathsf x,\mathscr V)>0$.
    \vspace{0.1cm}
  \item[-]  The operator $Q_t$ has a spectral gap   $b B(\mathscr M)$ if 
 $   \mathsf r_{ess}(Q_t|_{b B(\mathscr M)})< \mathsf r_{sp}(Q_t|_{b B(\mathscr M)})$, 
   where     $\mathsf r_{ess}(Q_t|_{b B(\mathscr M)})$ is the Wolf essential spectrum of $Q_t$ over $b B(\mathscr M)$ (see Section~\ref{sec.ress} below for the definition of $\mathsf r_{ess}$) and $\mathsf r_{sp}(Q_t|_{b B(\mathscr M)})$ is its spectral radius over $b B(\mathscr M)$. 
 \end{enumerate}

\subsubsection{Purpose of this work} 
\label{sec.pu}
In this work, in addition to drifts $\mathbf B$ which are $C^\infty$, we consider  drifts $\mathbf B$ belonging to two classes of locally bounded but not continuous   vector  fields:
\begin{enumerate}
 \item[-] Drifts $\mathbf B$ that are merely measurable and have at most linear growth, see Assumption {\rm \textbf{[{\footnotesize B$_{\text{LG}}$}]}} below.
 \vspace{0.2cm}
 
 \item[-] Drifts of the form $\mathbf B(\mathsf x) =- \nabla \mathbf U(x) + \boldsymbol \Theta( \mathsf x)$ (perturbed gradient field), which may exhibit superlinear growth in the position variable $x$ (where $\mathsf x=(x,v)\in \mathbb R^d\times \mathbb R^d$),  see Assumption {\rm \textbf{[{\footnotesize B$_{\text{P-Grad}}$}]}} below.
 \end{enumerate}
In this low-regularity setting, we study both non-killed and killed (outside of $\mathscr D$)  kinetic Langevin processes~\eqref{eq.Lan}  driven by pure jump Lévy noises, focusing on:
\begin{enumerate}
 
\item[-]  The topological irreducibility  of the non-killed and killed  semigroups (Theorem~\ref{th.C5-infty} and Theorem~\ref{th.sta}).
 \vspace{0.2cm}
 
\item[-]  The existence of a spectral gap for the associated killed   semigroup (Theorem~\ref{th.bw} and Proposition~\ref{pr.Rsp}) as well as the  compactness of this semigroup. 
 \vspace{0.2cm}
 
\item[-]  When the noise is a rotationally invariant $\alpha$-stable Lévy process with $\alpha\in(1,2)$:
\begin{enumerate}
 \vspace{0.1cm}
\item[-] Weak existence and uniqueness of solutions to \eqref{eq.Lan}. 
 \vspace{0.1cm}
 
 \item[-] Existence of a density in some $L^m$ spaces, $m>1$. 
  \vspace{0.1cm}
  
\item[-]   
The strong Feller property of the non-killed and killed  semigroups.
 \vspace{0.1cm}
 
\item[-] The weak continuity of the trajectories w.r.t. initial conditions.
 \vspace{0.1cm}
 
 \item[-] 
  The Feller $C_0(\mathbb R^{2d})$-semigroup property of the non-killed process. 
    \vspace{0.1cm} 
    
\item [-] The existence of a spectral gap for the non-killed semigroup.
 \end{enumerate}
  \vspace{0.1cm}
These results are the contents of Theorems \ref{th.FF} and \ref{th.SFD} and their notes,   and Theorems \ref{th.ex}, \ref{th.uni},  \ref{th.ex2}, and~\ref{th.sta}.   All these results are extended to stability index $\alpha\in (0,1]$ when the drift is $C^\infty$. 
\end{enumerate}  
  \vspace{0.1cm} 
  
Furthermore, still considering     rotationally invariant $\alpha$-stable  process as  driving noise  with $\alpha\in(1,2)$ we also establish the following results. 
\begin{enumerate}
\item[-]   When $\mathbf B$ is  measurable with  at most linear growth, and  when $\mathscr O$ is bounded, we prove the existence and uniqueness of a quasi-stationary distribution for the corresponding killed process in the space of probability measures on  $\mathscr D$, together with exponential convergence of the conditioned process toward this distribution (see Theorem~\ref{th.main}). These results are extended to the perturbed gradient field setting in Theorem~\ref{co.main}. 
\vspace{0.2cm}
\item[-] In the   perturbed gradient field setting and under additional dissipativity assumptions on $\mathbf U$ and $\Theta$ (see Assumption {\rm \textbf{[{\footnotesize B$_{\text{P-Grad}}$}]}} below), we prove the existence and uniqueness of a stationary distribution for the corresponding non-killed process in suitable weighted spaces, together with exponential convergence of the non-killed process toward this distribution  (Theorem~\ref{th.sta}). 
  \end{enumerate} 
Finally, we show that these results regarding stationary and quasi-stationary distributions remain valid for the stability index $\alpha \in (0,1]$, provided the drift coefficient is smooth (i.e. $C^\infty$). The distinction between the stability indices $\alpha \in (0,1]$ and $\alpha \in (1,2)$ is explained in Section \ref{sec.As}. Before defining our  assumptions in Section \ref{sec.As} and addressing the primary challenges in Section \ref{sec.MC}, we first introduce the necessary notation and some definitions.

\subsection{Notation}
\label{sec.Nota}
Let $\mathscr{M}$ be an open subset of $\mathbb{R}^{2d}$. We write $K\Subset \mathscr{M}$ to indicate that $K$ is a compact subset of $\mathscr{M}$. The Borel $\sigma$-algebra on $\mathscr{M}$ is denoted by $B(\mathscr{M})$, and $bB(\mathscr{M})$ stands for the space of all bounded Borel measurable functions $f:\mathscr{M}\to\mathbb{R}$, endowed with the supremum norm
$$
\|f\|_\infty=\sup_{x\in\mathscr{M}} |f(x)|.
$$
 Given an initial distribution $\nu$ on $ {\mathscr M}$, we write $\mathbb P_\nu[\cdot]=\int_{{\mathscr M}} \mathbb P_{x}[\cdot]\nu(d{x}) $.   The indicator function of a measurable set $\mathscr A$ is  denoted by  $\mathbf 1_{\mathscr A}$. We denote by  $\mathbf d_{\mathbb R^d}$   the standard distance in $\mathbb R^d$.  The space of probability measures over $\mathscr M$ is denoted by $\mathcal P(\mathscr M)$. The open ball of $\mathbb R^k$ centered at $m\in \mathbb R^{k}$ of radius $r>0$ is denoted by $ \mathscr B_{\mathbb R^{k}}(m,r)$. 
For $T\in [0,+\infty]$, the space of càdlàg functions from $[0,T]$ in $\mathbb R^{2d}$    is denoted by    $ D([0,T], \mathbb R^{2d})$ (the so-called Skorokhod  space) and its metric by  $\mathbf d_T$,  see~\cite[Section 12]{billingsley2013}.  
 For $p\in [1,+\infty]$, we denote by $L^p(\mathbb R^{2d})$ the standard $L^p$-space over $\mathbb R^{2d}$ of Borel functions $f:\mathbb R^{2d}\to \mathbb R$ satisfying  
 $$\Vert f\Vert_{L^p}:=  \Big [\int_{\mathbb R^{2d}}|f(x,v)|^p\,dx\,dv \Big ]^{1/p}, $$
 with obvious change if $p=+\infty$.
 Finally, we  recall that for $\beta \in (0,1)$, the anisotropic Hölder-Zygmund space $C_{\mathbf a}^\beta(\mathbb R^{2d})$ is the set of functions $f:\mathbb R^{2d}\to \mathbb R$ such that (see \cite[Proposition 2.5]{RocknerZhang}):
$$ \Vert f \Vert_{C_{\mathbf a}^\beta}:=\Vert f \Vert_\infty+ \sup_{\mathsf x\neq \mathsf x'\in \mathbb R^{2d}}\frac{|f(\mathsf x)-f(\mathsf x')|}{|\mathsf x-\mathsf x'|_{\mathbf a}^\beta}<+\infty,$$
 where $|\mathsf x-\mathsf x'|_{\mathbf a}= |v-v'|+ |x-x'|^{\frac{1}{1+\alpha}}$.

\subsection{Weak solution and martingale problem}
\label{sec.We}
In this section,    we recall the  definition of a weak solution to \eqref{eq.Lan} and give a proposition which links this notion with solutions to the associated martingale problem for SDEs driven by Lévy processes. 
Throughout this work, $ \mathbf B$ is measurable and locally bounded, and the Lévy measure $\nu$ is non-atomic. Such basic assumptions allow us to use the results of~\cite{Lepeltier} in what follows.  The following definition is standard (see~\cite[Définition \textbf{II}$_6$]{Lepeltier}).

\begin{definition}\label{de.w}  
A  weak solution to~\eqref{eq.Lan}  driven by a pure jump Lévy process with Lévy   measure $\nu$   and with initial   law $\mu \in \mathcal P (\mathbb R^{2d})$ is a $2d$-dimensional adapted process $(X_t=(x_t,v_t),t\ge 0)$ on some filtered probability space $(\Omega, \mathcal{F}, (\mathcal{F}_t)_{t\ge 0}, \mathbb{P})$ satisfying the usual conditions, with càdlàg  sample paths,  such that:
\begin{itemize}
    \item[-] The law of $X_0$ is $\mu$.
         \item[-] A.s for all $t\ge 0$, $\int_0^t( |v_s|+ |\mathbf B(X_s)|)ds<+\infty$. 
    \item[-] There exists  a pure jump  Lévy process  $(L_t,t\ge 0)$ with Lévy   measure $\nu$ on this filtered probability space  such that $\mathbb P$-a.s.
$$
      dx  _t=v _tdt, \ dv _t=\mathbf B(x _t,v_t)dt +   dL_t, \ t\ge 0.   
$$  
\end{itemize}
\end{definition}
\noindent
Let us denote by  $\mathcal L$     the infinitesimal generator associated with~\eqref{eq.Lan} whose expression is:
\begin{equation}\label{eq.gene}
\mathcal L= \mathcal S_v^{\nu} +v\cdot \nabla_x+ \mathbf B\cdot \nabla_v,\end{equation}
where   for $\mathsf y\in \mathbb R^{2d}$ and a function $\varphi:\mathbb R^{2d}\to \mathbb R$, 
 $$\mathcal S_v^{\nu}\varphi(\mathsf y)= \int_{\mathbb R^d}\big[ \varphi(\mathsf y+ (0,z))-\varphi(\mathsf y)-  \mathbf 1_{|z|\le 1} \nabla_v \varphi(\mathsf y)\cdot z\big] \nu (dz)  .$$
 We recall the following definition (see~\cite[Sections 3 and 4 in Chapter 4]{EK}).


\begin{definition}
A probability measure $\mathbf{P}$ on   $D([0,+\infty), \mathbb{R}^{2d})$  is a solution to the martingale problem for $\mathcal L$ starting from $\mu\in \mathcal P(\mathbb R^{2d})$      if the   coordinate process $(Y_t,t\ge 0)$, defined by $Y_t(\ell )=\ell _t$ for $\ell \in D([0,+\infty), \mathbb{R}^{2d})$, satisfies: 
\begin{itemize}
    \item[-]  The law of  $Y_0$ under $\mathbf{P}$ is $\mu$, i.e.  $\mathbf{P} \circ Y_0^{-1}=\mu$. 
    \item[-] For any $\varphi\in C_c^2(\mathbb R^{2d})$, the process $(M_{t}(\varphi),t\ge 0)$ defined by:
    \begin{equation*}
        M_t(\varphi)= \varphi(Y_t) - \varphi(Y_0) - \int_0^t \mathcal L \varphi(Y_s) \, ds  
    \end{equation*}
    is a $\mathbf{P}$-martingale with respect to  the natural filtration $(\mathcal F^Y_t)_{t\ge 0}$ of  the process $(Y_t,t\ge 0)$. 
\end{itemize}
\end{definition} 
\noindent
\textbf{Note}. Let  $(\mathcal{F}^a_t)_{t\ge 0}$ be   the $\sigma$-algebra $(\mathcal F_{t+}^Y)_{t\ge 0}$ completed by all $\mathbf P$-null sets of $\mathcal F^Y_\infty$. Then,  $(\mathcal{F}^a_t)_{t\ge 0}$ satisfies the usual conditions and $(M_t(\varphi),t\ge 0)$ is also a $\mathbf{P}$-martingale with respect to  $(\mathcal{F}^a_t)_{t\ge 0}$ (see e.g.~\cite[Remarque 4]{Lepeltier}). 
\medskip

\noindent
We denote by $(\mathcal L ,\mu)$ the  martingale problem  for $\mathcal L $ with initial condition $ \mu$. By extension, a process $(X_t,t\ge 0)\in D([0,+\infty), \mathbb{R}^{2d})$ satisfies the  martingale problem $(\mathcal L ,\mu)$ if its law satisfies it.   We also recall that uniqueness of the weak solution for~\eqref{eq.Lan} with given initial distribution  $\mu$  holds if, by definition,  the laws of the trajectories of two weak solutions are equal (see e.g.~\cite[Definition 129]{situ2005theory} or \cite[Définition \textbf{II}$_3$]{Lepeltier}). In the same way, uniqueness of the martingale problem $(\mathcal L ,\mu)$ holds if, by definition,  
$\mathbf P = \mathbf Q $ for any  two solutions $\mathbf P,\mathbf Q\in \mathcal P(D([0,+\infty), \mathbb{R}^{2d}))$ of this problem. 
By Itô's formula~\cite[Theorem 4.4.7]{applebaum2009levy}, 
the law of any weak solution provides a solution to the martingale problem. The reverse is also true, see e.g.~\cite[Théorème II$_{10}$]{Lepeltier} or~\cite[Theorem 2.3]{kurtz2010equivalence}. 
\begin{proposition}
\label{pr.Pre}
Let  $ \mu\in \mathcal P(\mathbb R^{2d})$. Then, if  $(X_t,t\ge 0)$ is a weak solution to \eqref{eq.Lan} with initial law $\mu$, the law of $(X_t,t\ge 0)$ is  a solution to the martingale
problem  $(\mathcal L ,\mu)$. Conversely, if $\mathbf P$    is  a solution to the martingale
problem  $(\mathcal L ,\mu)$, then there is a weak solution with law $\mathbf P$  to~\eqref{eq.Lan}  with initial condition~$\mu$. Hence, weak uniqueness holds for~\eqref{eq.Lan}  with initial condition $\mu$ if and only if uniqueness holds for the martingale problem $(\mathcal L ,\mu)$.
\end{proposition}

We finally recall for later purposes the following definition of the stopped martingale problem~\cite[Chapter 4, Section 6]{EK}.
\begin{definition}\label{de.SMP}
Given  an open subset $\mathscr U$ of $\mathbb R^{2d}$ and $\mathsf x\in  \mathbb R^{2d}$,   $(Y_t(\mathsf x),t\ge 0)$  is a solution to the \textit{stopped martingale problem} for $(\mathcal L,\mathsf x,\mathscr U)$  if    
$Y_t(\mathsf x)= Y_{t \wedge   \vartheta}(\mathsf x)$ for $t\ge 0$, $Y_0=\mathsf x$,  and if for all $\varphi \in C_c^2(\mathbb R^{2d})$, 
$$\varphi (Y_t(\mathsf x))-\varphi(\mathsf x)-\int_0^{t \wedge   \vartheta(\mathsf x)}\mathcal L \varphi (Y_s(\mathsf x))ds$$
is a $(\mathcal F_t^{Y(\mathsf x)})_{t\ge 0}$ martingale,   where  $\vartheta(\mathsf x)  := \inf \{t \ge 0 : Y_t(\mathsf x)\notin \mathscr U  \text{ or } Y_{t_-}(\mathsf x)\notin \mathscr U\}$. 
\end{definition}

Let us emphasize that, due to the presence of jumps, the localization time in the previous definition may be significantly shorter than the one   might naively be inclined to use, namely   $\inf \{t \ge 0 : Y_t(\mathsf x)\notin \mathscr U  \}$.

 
\subsection{Main assumptions and examples}
\label{sec.As}

\subsubsection{Assumptions}
We introduce the following three assumptions on the vector field $\mathbf{B}:\mathbb R^{2d}\to \mathbb R^d$:
 \begin{enumerate}
     \item[] {\rm\textbf{[{\footnotesize B$_{\text{smooth}}$}]}}  $\mathbf B$ is  $C^\infty$ and all its partial  derivatives of   order   $1$ are bounded over $\mathbb R^{2d}$.
 \vspace{0.1cm}
 
 \item[]  {\rm \textbf{[{\footnotesize B$_{\text{LG}}$}]}}      $\mathbf B$  is  measurable and has at most linear growth, i.e. there exists $C>0$ such that for all $\mathsf x\in \mathbb R^{2d}$, 
 $$|\mathbf B(\mathsf x)|\le C(1+|\mathsf x|).$$ 
 
 \item[]  {\rm \textbf{[{\footnotesize B$_{\text{P-Grad}}$}]}}  $\mathbf B(\mathsf x)=-\nabla \mathbf U(x)+ \boldsymbol \Theta( \mathsf x)$ for all  $\mathsf x=(x,v)$, where $\mathbf U:\mathbb R^d\to [1,+\infty)$ is $C^1$    and $ \boldsymbol \Theta:\mathbb R^{2d}\to \mathbb R^d$  is measurable and locally bounded.  Moreover,
 there  are constants $q \ge 2$  and $m_1, C^1_{\mathbf U},m_2, C^2_{\mathbf U}>0$  such that for all $x\in \mathbb R^{d}$:
\begin{equation}\label{eq.cond-U}
-\nabla  \mathbf U(x)\cdot  x \le -m_1\mathbf U(x)+C^1_{\mathbf U}  \ \text{ and } \      \mathbf U(x)\ge  m_2 |x|^q - C^2_{\mathbf U}.
\end{equation} 
In addition,  there are constants $C^1_\Theta , \Gamma>0$ and  $\ell_1\in [1,q)$ such that for all $\mathsf x=(x,v)\in \mathbb R^{2d}$:
\begin{equation}\label{eq.cond-gammaA}
\boldsymbol \Theta(\mathsf x) \cdot v \le -\Gamma |v|^2 + C^1_\Theta (1+ |x|^{\ell_1}). 
 \end{equation}    
 Finally,  $\Theta$ satisfies  one of the following conditions  \eqref{eq.cond-gammaB1} or \eqref{eq.cond-gammaB2}:
 \begin{equation}\label{eq.cond-gammaB1}
 \left\{
    \begin{array}{ll}
        \vspace{0.2cm}
        & \exists C^2_\Theta>0, \ell_2\in [1,q), \forall \mathsf x=(x,v)\in \mathbb R^{2d},  \\
        \vspace{0.2cm}
         & $\qquad \qquad  \qquad $\boldsymbol \Theta (\mathsf x)\cdot x  \le  C^2_\Theta \big[ 1+|v|^2  +   |x|^{\ell_2}\big],     \\
             \vspace{0.2cm}
         &\forall r>0, \exists C_r>0,  \forall \mathsf x=(x,v)\in \mathbb R^{2d},\\
           &$\qquad \qquad  \qquad $ \mathbf 1_{|x|\le r} | \Theta(\mathsf x)|\le C_r(1+|v|),
    \end{array} 
    \right.
  \end{equation}  
  or 
   \begin{equation}\label{eq.cond-gammaB2}
 \left\{
    \begin{array}{ll}
        \vspace{0.2cm}
        & \exists C^2_\Theta>0, \ell_2\in [1,q), \forall \mathsf x=(x,v)\in \mathbb R^{2d},  \\
        \vspace{0.2cm}
         & $\qquad \qquad  \qquad $ |\boldsymbol \Theta (\mathsf x) | \le  C^2_\Theta \big[ 1+|v|  +   |x|^{\ell_2-1}\big].  
    \end{array} 
    \right.
  \end{equation} 
   \end{enumerate} 
   The nomenclature for these assumptions reflects their structural properties: {\rm \textbf{[{\footnotesize B$_{\text{LG}}$}]}} refers to vector fields with at most linear growth, while {\rm \textbf{[{\footnotesize B$_{\text{P-Grad}}$}]}}  refers to  perturbed gradient fields. 
We incorporate condition \eqref{eq.cond-gammaB2} to account for standard friction models when $q=2$, exemplified by the damping force $\Theta(x,v) = -\Gamma v$.  The second condition in \eqref{eq.cond-gammaB1} is required to establish the topological irreducibility of both the killed and non-killed semigroups. Specifically, this condition facilitates control over the trajectories of the Langevin system \eqref{eq.Lan} via Grönwall’s inequality. The precise role of this requirement is demonstrated within the proof of Proposition \ref{pr.C5-infty}.
  
    We also define the following three assumptions: 
  \begin{enumerate}
    \item[]{\rm\textbf{[{\footnotesize H$^{\alpha\in (0,2)}_{\text{smooth}}$}]}}   Assumption  {\rm\textbf{[{\footnotesize B$_{\text{smooth}}$}]}} holds,  the driving  noise $L$ in \eqref{eq.Lan}   is a RI$\alpha$S  process, and   $\alpha \in (0,2)$. 
 \vspace{0.1cm}
 \item[]{\rm\textbf{[{\footnotesize H$^{\alpha\in (1,2)}_{\text{LG}}$}]}}   Assumption  {\rm\textbf{[{\footnotesize B$_{\text{LG}}$}]}} holds, the driving noise  $L$ in \eqref{eq.Lan} is a RI$\alpha$S  process,    and   $\alpha\in (1,2)$.
  \vspace{0.1cm}
 \item[]{\rm\textbf{[{\footnotesize{H$^{\alpha\in (1,2)}_{\text{P-Grad}}$}}]}}  Assumption   {\rm\textbf{[{\footnotesize B$_{\text{P-Grad}}$}]}} holds, the  driving noise  $L$  in \eqref{eq.Lan} is a RI$\alpha$S  process    and   $\alpha\in (1,2)$. 
\end{enumerate} 
Before we proceed, we make some comments on these assumptions.   
      Note that {\rm\textbf{[{\footnotesize B$_{\text{smooth}}$}]}} implies {\rm\textbf{[{\footnotesize B$_{\text{LG}}$}]}}.  We maintain the distinction   between assumptions \textbf{[{\footnotesize B$_{\text{smooth}}$}]} and \textbf{[{\footnotesize B$_{\text{LG}}$}]} to highlight the specific range of the parameter $\alpha$ covered in each case:
      \begin{itemize}
    \item[-] Under {\rm\textbf{[{\footnotesize B$_{\text{smooth}}$}]}}, our analysis accommodates the full range $\alpha \in (0,2)$ as defined in the assumption  {\rm\textbf{[{\footnotesize H$^{\alpha\in (0,2)}_{\text{smooth}}$}]}}.
    \vspace{0.1cm}
    
    \item[-] Under \textbf{[{\footnotesize B$_{\text{LG}}$}]}, our analysis specifically applies to $\alpha \in (1,2)$.
\end{itemize} 
 The restriction $\alpha\in(1,2)$ in  {\rm\textbf{[{\footnotesize H$^{\alpha\in (1,2)}_{\text{LG}}$}]}} and {\rm\textbf{[{\footnotesize{H$^{\alpha\in (1,2)}_{\text{P-Grad}}$}}]}}  is purely due to the perturbative method we use  to prove both the weak well-posedness  and  the strong Feller property of the non-killed process, see the discussion at the beginning of Section~\ref{sec.duha} and   the note following Theorem~\ref{th.FF}. This constraint on $\alpha$ is the technical trade-off required to accommodate coefficients with low regularity. 
On the other hand, the assumption that $\mathbf{U}: \mathbb{R}^d \to [1, +\infty)$ is $C^1$ in {\rm \textbf{[{\footnotesize B$_{\text{P-Grad}}$}]}}  (and not just differentiable with locally bounded gradient field)   is only imposed to ensure that the Lyapunov functional $W_p$ in \eqref{eq.cond-Wp} is $C^1$ in the spatial variable $x$, a prerequisite for the application of It\^o's formula. 
   We emphasize that, for results concerning the killed semigroup on $\mathscr D$ (see~\eqref{eq.DD}), only the restriction of $\mathbf B$ to $\bar{\mathscr D}$ is relevant.  
  \medskip

As already detailed in Section \ref{sec.pu}, we prove the existence of a spectral gap and the compactness of  the killed semigroup as well as   the topological irreducibility property of both the killed and non-killed semigroups   (see Theorems~\ref{th.bw}, Theorem  \ref{th.C5-infty},  Theorem~\ref{th.SFD}, and Proposition~\ref{pr.Rsp})     for rather general pure jump Lévy processes (which are not necessarily RI$\alpha$S processes)   assuming  a weak well-posedness framework and the weak continuity property of the trajectories w.r.t initial conditions. For that reason,  given  a pure jump Lévy process $L=(L_t,t\ge 0)$, we finally define the following assumptions, focusing only on \textbf{[{\footnotesize B$_{\text{LG}}$}]}.
\begin{enumerate}
     \item[]  \textbf{[{\footnotesize W$_{\text{LG}}$}]}    For every initial distribution $\mu \in \mathcal P(\mathbb R^{2d})$ and  $\mathbf B$ satisfying \textbf{[{\footnotesize B$_{\text{LG}}$}]}, there is a unique weak solution to~\eqref{eq.Lan}  driven by $L$ with $X_0=\mu$ in law. 
 \vspace{0.1cm}
 \item[] \textbf{[{\footnotesize C$_{\text{LG}}$}]}    Assumption   \textbf{[{\footnotesize W$_{\text{LG}}$}]} is satisfied and the following mapping   is  weakly continuous\footnote{I.e. is continuous with respect to weak convergence of probability measures in $\mathcal P(D([0,+\infty), \mathbb{R}^{2d}))$.}:   
 $$\mathsf x\in \mathbb R^{2d}  \to  \mathbb P_{\mathsf x}[X   \in \cdot ] \in \mathcal P(D([0,+\infty), \mathbb{R}^{2d})).$$
    \end{enumerate}
  We recall that we prove that \textbf{[{\footnotesize W$_{\text{LG}}$}]}  and \textbf{[{\footnotesize C$_{\text{LG}}$}]}  are satisfied  when $L$ is a RI$\alpha$S process and $\alpha\in (1,2)$, see Theorem~\ref{th.ex2}. We also prove the weak well-posedness of \eqref{eq.Lan} and the weak continuity property of  its trajectories  w.r.t. initial conditions  when      \textbf{[{\footnotesize H$^{\alpha\in (1,2)}_{\text{P-Grad}}$}]} holds, see Theorem~\ref{th.sta} where we also extend these results to $\alpha \in (0,1]$ when $\mathbf U$ and $\Theta$ are $C^\infty$. 
    
    \subsubsection{Some examples}
    
  Assumption {\rm\textbf{[{\footnotesize B$_{\text{LG}}$}]} encompasses non-conservative Hamiltonian vector fields   given by:
  $$\mathbf{B}(x,v) = \mathbf{M}(x) - \mathbf C(x,v) v, \qquad (x,v) \in \mathbb{R}^{2d},$$
   where  $\mathbf{M} : \mathbb{R}^d \to \mathbb{R}^d$ is a measurable mapping exhibiting at most linear growth and where,  for instance,  the  friction matrix   $\mathbf C$ is  measurable and bounded.
Discontinuous friction matrices $\mathbf{C}(x,v)$ naturally arise in various physical contexts where the dissipation regime undergoes an abrupt change, as illustrated by the following examples:
 \begin{enumerate}
 \item[-] Interface in an heterogeneous media (a particle moving through two different viscosities separated by an interface, e.g. at $x_1 = 0$):
 $$\mathbf{C}(x) = \gamma_1 \mathbb{I}_d \cdot \mathbf{1}_{ x_1 < 0 } + \gamma_2 \mathbb{I}_d\cdot \mathbf{1}_{ x_1 \geq 0 }, \ \  \gamma_1 \neq \gamma_2 \in \mathbb R.$$
 \item[-] Velocity-threshold (dissipation kicks in above a critical speed $v_c\in \mathbb R_+$):
 $$\mathbf{C}(v) = \gamma \cdot \mathbf{1}_{ |v| > v_c}\mathbb{I}_d,   \ \  \gamma \in \mathbb R.$$
 \item[-]   Anisotropic dissipation (the friction depends on the orientation of the particle's velocity, leading to a rank-one projection matrix):
 $$\mathbf{C}(v) = \gamma \frac{v \otimes v}{|v|^2} \cdot \mathbf{1}_{|v| > 0 },  \ \  \gamma \in \mathbb R.$$
 This matrix is bounded but discontinuous at $v = 0$.
  \end{enumerate}
    
More generally, the class {\rm\textbf{[{\footnotesize B$_{\text{LG}}$}]}  accommodates the piecewise-constant force fields  with at most linear growth  frequently encountered in the study of heterogeneous media. In such systems, 
 $$
  \mathbf B(x,v)
  =  \sum_{j=1}^n \mathbf  B_j(x,v)\,\mathbf 1_{ (x,v)\in  \Omega_j }, \qquad  (x,v)\in \mathbb R^{2d},
  $$
  where  $\mathbf B_j:\mathbb R^{2d}\to \mathbb R^d$ are  measurable   (here  $\Omega_j\in B( \mathbb R^{2d})$). 
Conversely,   the class {\rm \textbf{[{\footnotesize B$_{\text{P-Grad}}$}]}}  encompasses  conservative Hamiltonian vector fields of the form 
    $$\mathbf{B}(x,v) = -\nabla \mathbf U -\mathbf C(x,v)  v,$$
   where $\mathbf U$ satisfies \eqref{eq.cond-U} and where, to give an example,    the  friction  matrix   $\mathbf C\ge c>0$ is  measurable and bounded.  Such Hamiltonian vector fields are fundamental to statistical physics, as they naturally describe the dynamics of particles subject to external potentials and dissipative forces. 
      However,  the scope of {\rm \textbf{[{\footnotesize B$_{\text{P-Grad}}$}]}}   is significantly broader.  


 \subsection{Main challenges} 
 \label{sec.MC}
 The analysis of several properties of both the non-killed and the killed (outside of $\mathscr D$) process~\eqref{eq.Lan} presents several significant theoretical challenges when the driving noise is a pure-jump Lévy  process. 
In this section, we discuss some of them. 
    
 \begin{enumerate}
 \item[-] \textbf{Spectral gap}.  
 The first main difficulty we faced when studying    the killed process ~\eqref{eq.Lan} concerns the spectral gap of  its killed semigroup $(P^{\mathscr D}_t,t>0)$  (see~\eqref{eq.PTD} below).  Because the essential spectral radius of the killed semigroup~\eqref{eq.PTD} is unknown (and decreases exponentially fast when $t\to +\infty$) the spectral gap is usually obtained   using a Lyapunov functional $W$ such that $\mathcal{L}W/W \to -\infty$ at infinity (see indeed the discussion about assumption \textbf{(C3)} in ~\cite{guillinqsd}), which is a powerful criteria  to establish  spectral gap of    killed semigroup~\cite{guillinqsd}.  
 Here $\mathcal{L}$ is the generator of the considered non-killed process.
 However, without extra  assumptions on the moments of the Lévy measure (which would exclude RI$\alpha$S processes), we are not able to prove the existence of such a Lyapunov function. 
 We overcome this challenge by employing direct arguments to calculate the measure of non-compactness $\beta_w$ (see Section \ref{sec.ress}), while taking full advantage of the process's kinetic properties.  This   measure of non-compactness  allows us to quantify the essential spectral radius of the semigroup, see Theorem~\ref{th.bw}, but also to derive the compactness of the killed semigroup.
\vspace{0.2cm}

 \item[-] \textbf{Topological irreducibility}.  Establishing the irreducibility of the process is non-trivial. When the drift $\mathbf B$ is merely locally  bounded, we cannot rely on standard support theorems (in any case, without extra assumptions, the controlled curves are not well-defined in this case). Furthermore, even in the   Lipschitz case, constructing admissible trajectories  within $\mathscr{D}$ that connect arbitrary points in $\mathscr{D}$ using only velocity jumps within a fixed time $t>0$ is a delicate geometric problem. We prove the topological irreducibility of  both the killed and the non-killed process; see Theorem~\ref{th.C5-infty} for a rather general result which is then applied to the case when the driven noise if a RI$\alpha$S  in   Theorem~\ref{th.sta} and in (the proof of) Theorem~\ref{co.main}. To establish  the topological irreducibility, we first  construct suitable events through an adapted decomposition of the driving noise. On these events, we carefully analyze the trajectories of \eqref{eq.Lan} over sufficiently small time intervals, ensuring uniformity with respect to the initial conditions (see Proposition \ref{pr.C5-infty}). By leveraging this uniformity alongside the Markov property and the lower semicontinuity of the killed semigroup (see Proposition  \ref{pr.Co}), we extend these attainability results to any arbitrary time $t > 0$.   
\vspace{0.2cm}

  \item[-] \textbf{Global weak well-posedness of \eqref{eq.Lan}}. The technical challenge in proving Theorem~\ref{th.ex} arises from the combination of the  possible discontinuity of the drift $\mathbf{B}$ and the hypoelliptic (degenerate) nature of~\eqref{eq.Lan}, where the noise acts exclusively on the velocity variable. Furthermore, establishing weak well-posedness via a Girsanov transform - the standard approach for handling non-smooth coefficients in Brownian-driven SDEs - presents significant technical complications in this context. 
We refer to the paragraphs following Theorem~\ref{th.ex} for a discussion of the related literature. To establish the weak well-posedness of~\eqref{eq.Lan} under assumptions \textbf{[{\footnotesize H$^{\alpha\in (1,2)}_{\text{LG}}$}]} or \textbf{[{\footnotesize{H$^{\alpha\in (1,2)}_{\text{P-Grad}}$}}]}, we first address weak well-posedness for  a bounded drift $\mathbf{B}$ (Theorems~\ref{th.ex} and~\ref{th.uni}). In this globally bounded case,  our approach relies on an integrated Duhamel-type perturbative formula \eqref{eq.DI} for the non-killed semigroup, expressed in terms of the driftless process~\eqref{eq.L00}; this requires a careful  analysis of the semigroup gradient's singularity as $t \to 0$. 

This perturbative framework enables us to derive the integrated Krylov-type estimates\footnote{This analytical tool has already been successfully employed in various different contexts, demonstrating its broad applicability to the well-posedness of SDEs.} in \eqref{eq.Kt}.    These estimates are subsequently employed to establish the invertibility of a broad class of associated operators, which serves as a central argument for proving weak uniqueness. For weak existence, these Krylov-type estimates also enable us to pass to the limit in the associated martingale problem, thereby circumventing the discontinuity of $\mathbf{B}$. Finally, we extend these well-posedness results to the unbounded drift case (Theorems~\ref{th.ex2} and~\ref{th.sta}) by using Grönwall's inequality or by constructing suitable Lyapunov functions,    and by using the stopped martingale problem as developed in~\cite[Chapter 4, Section 6]{EK} (see Definition \ref{de.SMP}).


  \vspace{0.2cm}
 \item[-] \textbf{The strong Feller property}.  
In order to prove that $P_t^{\mathscr{D}}$ is strongly Feller (see Theorem~\ref{th.SFD}), we begin by showing that $ P_t$ itself enjoys the strong Feller property (Theorem~\ref{th.FF}). 
    Proving that  the non-killed semigroup $P_t$ is strongly Feller when $\mathbf B$   is a measurable and locally bounded is challenging because classical tools are unavailable. Specifically, the jump nature of the   noise and the lack of regularity of $\mathbf B$ preclude   the use of Girsanov transformation, Bismut-Elworthy-Li formulas, or Malliavin calculus, which are typically employed in the Gaussian case or when the drift is smooth enough. 
 We start by considering the case when   $\mathbf B$ is bounded and  we show the strong Feller property of $P_t$ in this case  in several steps: 

\vspace{0.2cm}
\begin{enumerate}
\item[$\bullet$] We prove another (non integrated) perturbative formula~\eqref{eq.Sfd}  for the non-killed semigroup involving the \textit{non drifted} associated process~\eqref{eq.L00}, handling  carefully  again the semigroup gradient's singularity as $t \to 0$.  

 \vspace{0.2cm}

\item[$\bullet$] Using this perturbative formula and bounds on the   derivatives of the semigroup of $(\hat X_t,t\ge 0)$  recently derived in  \cite{RocknerZhang} in Besov  spaces, we prove existence of a density of the process~\eqref{eq.Lan} as well as  upper bounds in $L^p$ spaces  which are uniform  in the starting point $\mathsf x\in \mathbb R^{2d}$. These upper bounds are used to prove the uniform integrability of the density, which is the main ingredient in the proof of the strong Feller property (see Lemma \ref{pr.Schilling} and Theorem~\ref{th.FF}). 
\end{enumerate}
\vspace{0.2cm}

\noindent
 The strong Feller property of the non-killed semigroup is then extended  when  $\mathbf B$ is not bounded (and more precisely under \textbf{[{\footnotesize H$^{\alpha\in (0,2)}_{\text{smooth}}$}]}, \textbf{[{\footnotesize H$^{\alpha\in (1,2)}_{\text{LG}}$}]} or \textbf{[{\footnotesize{H$^{\alpha\in (1,2)}_{\text{P-Grad}}$}}]}) using suitable Lyapunov functions and localization arguments (Theorems \ref{th.FF} and~\ref{th.sta}). The strong feller property of $P_t^{\mathscr D}$ is then derived in Theorem~\ref{th.SFD} using the strong Feller property of $P_t$, $t>0$ and by showing that the distribution function of the exit times from balls converges uniformly with respect to the initial condition on compact sets toward zero~\eqref{eq.Lpop}. It is important to emphasize that the jump characteristics of the driving noise (here, RI$\alpha$S processes) play a crucial role in establishing the strong Feller property.

   \end{enumerate}

\section{On the essential spectral radius of the   killed process} 
\label{sec.ress} 
 We recall that for a bounded linear operator  $T$ on  $bB(\mathbb R^{2d})$,  the essential spectral radius
$$
\mathsf r_{ess}(T|_{bB(\mathbb R^{2d})})
:= \sup\{|\lambda| \,;\, \lambda \in \sigma_{\mathrm{ess}}(T|_{bB(\mathbb R^{2d})})\}
\quad
(\text{convention: } \sup \varnothing := 0),
$$
where $\sigma_{\mathrm{ess}}(T|_{bB(\mathbb R^{2d})})$ is the Wolf essential spectrum of $T|_{bB(\mathbb R^{2d})}$, i.e. the set of $\lambda \in \mathbb C$ such that $\lambda I - T$ is not a Fredholm operator on $bB(\mathbb R^{2d})$.
The goal of this section is to prove the following result.

\begin{theorem}\label{th.bw} 
Let $\mathscr O$ be  a nonempty bounded open subset of $\mathbb R^d$ and $\mathscr D=\mathscr O\times \mathbb R^d$.   Consider~\eqref{eq.Lan} driven by a Lévy process $(L_t,t\ge 0)$. 
 Assume   that either {\rm\textbf{[{\footnotesize B$_{\text{smooth}}$}]}} holds, or that the process $(L_t,t\ge0)$ satisfies {\rm\textbf{[{\footnotesize W$_{\text{LG}}$}]}}. Assume finally    that  $(P_t,t> 0)$ is strongly Feller.   Then, for any $t>0$,  $\mathsf r_{ess}(P_t^{\mathscr D}|_{b B(\mathscr D)})=0$ (see~\eqref{eq.PTD}) and  $P_t^{\mathscr D}: bB(\mathscr D)\to bB(\mathscr D)$ is compact. 
 \end{theorem} 

We now provide some comments on Theorem~\ref{th.bw}. Theorem~\ref{th.bw} is the only result in this work where  we require  $\mathscr O$ to be bounded.  
 Notice that even if $\mathscr O$ is bounded, the set $\mathscr D$ is unbounded. 
We   emphasize that, given a   Lévy process $(L_t,t\ge 0)$,  under the assumption of {\rm\textbf{[{\footnotesize W$_{\text{LG}}$}]}}, we assume that weak well-posedness holds for any drift $\mathbf{B}$ satisfying \textbf{[{\footnotesize B$_{\text{LG}}$}]}.    Let us  recall that, under \textbf{[{\footnotesize B$_{\text{smooth}}$}]},   \eqref{eq.Lan} is strongly well-posed, for any Lévy process $L$. 
Moreover, we assume a weak well-posedness framework in Theorem~\ref{th.bw} so that  the process~\eqref{eq.Lan} is Markovian, and therefore both the  non-killed and the killed semigroups are well defined.  Theorem~\ref{th.bw} can also be extended to other types of noise, provided one works within a weak well-posedness framework and assumes that the non-killed semigroup is strongly Feller (or more generally satisfies Assumption \textbf{(A1)} below).

 When   \textbf{[{\footnotesize H$^{\alpha\in (1,2)}_{\text{LG}}$}]}   is satisfied, we  prove below  in Theorem~\ref{th.ex2}  that {\rm\textbf{[{\footnotesize W$_{\text{LG}}$}]}} holds. In    Theorem~\ref{th.FF},  we show that   that $P_t$ is strongly Feller for $t>0$ when   \textbf{[{\footnotesize H$^{\alpha\in (1,2)}_{\text{LG}}$}]} or \textbf{[{\footnotesize H$^{\alpha\in (0,2)}_{\text{smooth}}$}]} holds. 
 Hence, according toTheorem \ref{th.bw}, under \textbf{[{\footnotesize H$^{\alpha\in (1,2)}_{\text{LG}}$}]} or \textbf{[{\footnotesize H$^{\alpha\in (0,2)}_{\text{smooth}}$}]}, $\mathsf r_{ess}(P_t^{\mathscr D}|_{b B(\mathscr D)})=0$  and  $P_t^{\mathscr D}: bB(\mathscr D)\to bB(\mathscr D)$ is compact.
 
  In addition, when   \textbf{[{\footnotesize H$^{\alpha\in (1,2)}_{\text{P-Grad}}$}]} is satified, we prove that \eqref{eq.Lan} is weakly well-posed and that $P_t$ is strongly Feller for $t>0$, and extend   Theorem~\ref{th.bw} to this setting, see Theorem~\ref{th.sta} and the proof of Theorem~\ref{co.main}.

 \medskip
 
 To prove Theorem~\ref{th.bw}, we rely on the measure of non-compactness $\beta_w$ for non-negative transition kernels introduced in \cite{Wu2004}, which we   recall below.
 For a bounded non-negative transition kernel $T$ on  $\mathbb R^{2d}$, set
$$
\beta_w(T  ):=\inf_{K\Subset \mathbb R^{2d}}\, \sup_{ \mathsf x\in \mathbb R^{2d}}\,  T(  \mathsf x, \mathbb R^{2d}\setminus K),$$
and 
$$
\beta_\tau(T ):=\sup_{(A_n)_n}\, \lim_{n\to +\infty} \,  \sup_{  \mathsf x\in \mathbb R^{2d}}\, T(  \mathsf x, A_n) ,
$$
where the  supremum is taken over all the sequences  $(A_n)_n$ of elements of  $   B(\mathbb R^{2d})$ such that  $A_n\searrow \varnothing$. The kernel $T$ satisfies  \cite[\textbf{(A1)}]{Wu2004} if by definition,  
$$
\text{\textbf{(A1)} }\  \, \exists N\ge 1, \forall  K\Subset \mathbb R^{2d}, \  \,  \beta_w(\mathbf 1_K T)=0 \text{ and  } \beta_\tau(\mathbf 1_KT^N)=0.
$$ 
Under this  assumption, the last author proves (see \cite[Theorem 3.5]{Wu2004}) the following Gelfand-Nussbaum type formula  for the essential spectral radius of $T$: 
 \begin{equation}\label{eq.Gelfand}
 \mathsf r_{ess}(T|_{bB(\mathbb R^{2d})})=\lim_{n\to +\infty}\big [ \beta_w((T|_{bB(\mathbb R^{2d})})^n)\big ]^{1/n}.
 \end{equation} 
 We also recall that, using Dini's theorem, the condition   \textbf{(A1)} is satisfied if $T$ is strongly Feller, with $N=1$ (see~\cite{guillinqsd}).


 \begin{proof} (Theorem~\ref{th.bw})    
Since by assumption, $P_t$ is strongly Feller for $t>0$, it satisfies  \textbf{(A1)}  with $N=1$. 
We consider $P_t^{\mathscr D}$ as an operator  on $bB(\mathbb R^{2d})$ ($t>0$). 
In addition,  $P_t^{\mathscr D}$ satisfies  \textbf{(A1)} with $N=1$ for every $t>0$ since $P_t^{\mathscr D}\le P_t$.  
Moreover, we have that 
\begin{align*}
 \beta_w(P_t^{\mathscr D})&= \inf_{K\Subset \mathbb R^{2d}} \sup_{\mathsf x=(x,v)\in \mathbb R^{2d}} P_t^{\mathscr D} (\mathsf x, \mathbb R^{2d} \setminus K)\\
 &= \inf_{K\Subset \mathbb R^{2d}} \sup_{\mathsf x=(x,v)\in \mathscr D} P_t^{\mathscr D} (\mathsf x, \mathscr D\setminus K).
 \end{align*}
Choosing $K_R=\overline{\mathscr O} \times \{  |v'|\le R\}$ (for $R>0$), one deduces that:
\begin{align}
\label{eq.bR}
 \beta_w(P_t^{\mathscr D})&\le \inf_{R> 0} \ \sup_{\mathsf x=(x,v)\in \mathscr O\times \mathbb R^d}\  P_t^{\mathscr D}(\mathsf x,  \mathscr O\times \{v',\, |v'|> R\}).   
\end{align}
The goal is to prove that  $\beta_w(P_t^{\mathscr D})=0$ for any $t>0$. 
 Indeed, by~\cite[Proposition 3.2.(e)]{Wu2004}, if $\beta_w(P_t^{\mathscr D})=0$, we will then get that  
 $$\beta_w((P_t^{\mathscr D})^n)\le \beta_w(P_t^{\mathscr D})^n=0,$$
  and the result $\mathsf r_{ess}(P_t^{\mathscr D}|_{b B(\mathscr D)})=0$ will  follow from~\eqref{eq.Gelfand} and the fact that 
\begin{equation}\label{eq.res}
\mathsf r_{ess}(P_t^{\mathscr D}|_{b B(\mathscr D)})=\mathsf r_{ess}(P_t^{\mathscr D}|_{bB(\mathbb R^{2d})}).
\end{equation} 
The proof of~\eqref{eq.res} is rather standard, but we briefly sketch it for the sake of completeness.  Note that 
 $$b B(\mathbb{R}^{2d}) = b B(\mathscr{D}) \oplus b B(\mathbb{R}^{2d} \setminus \mathscr{D}).$$
Moreover, $P_t^{\mathscr{D}}f(\mathsf{x}) = 0$ for all $\mathsf{x} \in \mathbb{R}^{2d} \setminus \mathscr{D}$ and, for $f=f_1+f_2\in  b B(\mathscr{D}) \oplus b B(\mathbb{R}^{2d} \setminus \mathscr{D})$, $P_t^{\mathscr{D}}f=P_t^{\mathscr{D}}f_1$. The operator $P_t^{\mathscr{D}}$ then acts as follows  relative to this decomposition:
$$P_t^{\mathscr{D}} |_{bB(\mathbb R^{2d})}= \begin{pmatrix} P_t^{\mathscr{D}}|_{b B(\mathscr{D})} & 0 \\ 0 & 0 \end{pmatrix}
.$$
 Using the definition of a Fredholm operator, the Wolf essential spectrum of a block diagonal operator is the union of the essential spectra of its blocks, and hence in our case:
 $$\sigma_{ess}(P_t^{\mathscr{D}}|_{b B(\mathbb{R}^{2d})}) = \sigma_{ess}(P_t^{\mathscr{D}}|_{b B(\mathscr{D})}) \cup \sigma_{ess}(0).$$
%
 proving~\eqref{eq.res}. 
Let us come back to the proof of  $\beta_w(P_t^{\mathscr D})=0$. To show it, we aim to prove that  (see~\eqref{eq.bR}):
 \begin{equation}\label{eq.sy0}
 \lim_{R\to +\infty}  \ \sup_{\mathsf x=(x,v)\in \mathscr O\times \mathbb R^d}\  P_t^{\mathscr D}(\mathsf x,  \mathscr O\times \{  |v'|> R\})=0.
 \end{equation} 
Fix $t>0$ and define
\[
\mathscr A_t(\mathsf x) := \{x_s(\mathsf x)\in\mathscr O \text{ for all } s\in[0,t]\}.
\]
Then, proving \eqref{eq.sy0} is equivalent, by definition of $P_t^{\mathscr D}$,  to showing that  \begin{align}
\label{eq.Pgb}
\lim_{R\to\infty}
\sup_{\mathsf x=(x,v)\in\mathscr O\times\mathbb R^d}
\mathbb P_{\mathsf x}\bigl[\mathscr A_t\cap\{|v_t |>R\}\bigr]=0.
\end{align} 
 In what follows, $C>0$ is a   constant  depending only on $\mathbf B$ and  on the bounded set $\mathscr O$,  which may change from one occurrence to another. 
We now prove~\eqref{eq.Pgb}.  
We write  for $t>0$, $R>0$,  and $\mathsf x=(x,v)\in \mathscr O\times \mathbb R^d$, 
\begin{align}
\nonumber
 \mathbb{P}_{\mathsf{x}}[\mathscr A_t \cap  \{|v_t| > R\}] &=  \mathbf 1_{  |v|< \sqrt R} \  \mathbb{P}_{\mathsf x}\big [\mathscr A_t \cap  \{|v_t| > R\}\big ] \\
 \label{eq.Kh}
&\quad +\mathbf 1_{  |v|\ge \sqrt R}   \   \mathbb{P}_{\mathsf{x}}\big [\mathscr A_t \cap  \{|v_t| > R\}\big ].
\end{align}
Recall that  we have (by assumption) for all $\mathsf x=(x,v)\in \mathbb R^{2d}$, 
\begin{equation}\label{eq.LGr}
|\mathbf{B}(\mathsf x)| \le C[1 + |x| + |v|].
\end{equation}

 Let us first consider the regime of the sufficiently small   initial velocities   $|v| \le \sqrt{R}$.
Let  $\mathsf x=(x,v)\in\mathscr O\times\mathbb R^d$. On the event $\mathscr A_t(\mathsf x)$, we have $|x_s(\mathsf x)| \le C$ for all $s\in [0,t]$ (because $\mathscr O$ is bounded), and hence,   by \eqref{eq.LGr}, we deduce that 
\begin{equation}\label{eq.Lpo}
|\mathbf{B}(X_s(\mathsf x))| \le C[1+|v_s(\mathsf x)|], \,\, \forall x\in \mathscr O,\, \,v\in \mathbb R^d.
\end{equation}
  From the velocity equation, we have on $\mathscr A_t(\mathsf x)$, for every  $s\in [0,t]$  and  $\mathsf x=(x,v)\in\mathscr O\times\mathbb R^d$:
$$
|v_s(\mathsf x)| \le |v| + C \int_0^s \big[ 1+|v_u(\mathsf x)|  \big] du + |L_s|.$$
  Applying Grönwall's inequality~\cite[Section 5 in Appendices]{EK}, on the event $\mathscr A_t(\mathsf x)$, it holds that:
\begin{equation}\label{eq.Gr}
|v_s(\mathsf x)| \le \left[ |v| + Ct + S_t \right] e^{Cs}, \quad \text{where  } S_t = \sup_{s \in [0,t]} |L_s|.
\end{equation}
In particular,  on the event $\mathscr A_t(\mathsf x)$,
$$|v_t(\mathsf x)| \le  [ |v| + Ct + S_t  ] e^{Ct}.$$ 
  On the event $\mathscr A_t(\mathsf x)\cap \{|v_t(\mathsf x)| > R\}$ where  $\mathsf x=(x,v)\in\mathscr O\times\mathbb R^d$  is such that  $|v| \le \sqrt{R}$, we thus obtain that:
  $$R < [\sqrt{R} + Ct + S_t] e^{Ct} \implies S_t > R e^{-Ct} - \sqrt{R} - C t.$$
 
Let us now  consider    the regime of  large  initial velocities   $|v| \ge \sqrt{R}$, $\mathsf x\in \mathscr O\times \mathbb R^d$.
By  the triangular integration identity we have for any non-negative measurable  functions $f:[0,t]\to \mathbb R$:
$$\int_0^t \int_0^s f_r dr ds = \int_0^t (t-r)  f_r dr.$$
Then, integrating the position equation yields:
$$x_t(\mathsf x) = x + vt + \int_0^t (t-r) \mathbf{B}(x_r(\mathsf x), v_r(\mathsf x)) dr + \int_0^t L_r dr.$$ 
Thus, one has:
\begin{equation}\label{eq.jhl}
|v|t \le |x_t(\mathsf x) - x| + \int_0^t (t-r) |\mathbf{B}(x_r(\mathsf x), v_r(\mathsf x))| dr + \left| \int_0^t L_r dr \right|.
\end{equation}
On the other hand, recall that for $\mathsf x\in \mathscr O\times \mathbb R^d$, on $\mathscr A_t(\mathsf x)$, $|x_t(\mathsf x) - x|\le C$ and we have~\eqref{eq.Gr} that we recall:
$$|v_u(\mathsf x)| \le \left[ |v| + Ct + S_t \right] e^{Cu}, \ u\in [0,t].$$ 
We then deduce from~\eqref{eq.Lpo} and ~\eqref{eq.jhl} that on $\mathscr A_t(\mathsf x)$ ($\mathsf x\in \mathscr O\times \mathbb R^d$):
$$|v|t \le C+ t S_t + C\int_0^t (t-r) dr +C \int_0^t  (t-r) \big[|v| + C t + S_t\big] e^{Cr} dr.$$
Isolating terms with $|v|$, one gets that:
$$|v| \left[ t - \int_0^t C(t-r)e^{Cr} dr \right] \le C+ t S_t + \frac{C t^2}{2} + [C t + S_t] \int_0^t C(t-r)e^{Cr} dr.
$$
The coefficient of $|v|$ in the left-hand side of the previous bound  is 
$$g(t) = t - \int_0^t C(t-r)e^{Cr} dr = 2t - \frac{e^{Ct} - 1}{C}.$$
 For small $t>0$, 
 $$g(t) = t - \frac{Ct^2}{2} + O(t^3).$$
  Let $\delta_C>0$ such that  
  $$0<t \le  \delta_C  \Rightarrow  g(t) > 0.$$ 
Note that $\delta_C$ depends solely on the constant $C>0$, which itself is a function of the domain $\mathscr{O}$ and the drift $\mathbf{B}$ restricted to $\bar{\mathscr{D}}$.
  Let 
  $$I(t) = \int_0^t C(t-r)e^{Cr} \, dr = \frac{e^{Ct}-1}{C} - t .$$ 
  The inequality thus becomes:
  $$g(t) |v|   \le C+ t S_t + \frac{C t^2}{2} + [C t + S_t] I(t) .$$ 
Let  $\mathsf x\in \mathscr O\times \mathbb R^d$  with $|v| > \sqrt{R}$. Then,  on $\mathscr A_t(\mathsf x)$, it holds that:
   $$S_t \ge \frac{\sqrt{R}\,  g(t) - C - {C t^2}/{2}-CtI(t)  }{I(t)+t}.$$
Consequently, as $R \to \infty$, the magnitude of the noise required to remain within $\mathscr{O}$ tends to infinity. Let us emphasize that this phenomenon is characteristic of kinetic equations. We now summarize the previous estimates.  Recall~\eqref{eq.Kh}.  We have proved that for all $\mathsf{x} \in \mathscr O\times \mathbb R^d$,  $R>0$,  and $0<t \le  \delta_C$:
\begin{align*}
  \mathbb{P}_{\mathsf{x}}[\mathscr A_t \cap  \{|v_t| > R\}] &\le \mathbb{P}\left[ S_t > R e^{-Ct} - \sqrt{R} - C t \right] \\
&\quad + \mathbb{P}\left[ S_t \ge \frac{\sqrt{R}\,  g(t) - C - {C t^2}/{2}-CtI(t)  }{I(t)+t} \right].
\end{align*}
Both terms in the right-hand side are independent of  $\mathsf{x} \in \mathscr O\times \mathbb R^d$ and converge to  $0$ as $R \to \infty$.  This shows that~\eqref{eq.sy0} holds for all   $0<t \le \delta_C$. Thus for any $0<t \le \delta_C$:
$$\beta_w(P_t^{\mathscr D})=0.$$
We propagate this over all $t>0$ by induction as follows. 
Let $t=   \delta_C +s$ with $0<s \le  \delta_C$.  Recall that for all $K\Subset \mathbb R^{2d}$, $\beta_w(\mathbf 1_KP_{s}^{\mathscr D})=0$ (since $P_s^{\mathscr D}$ satisfies  \textbf{(A1)}, for all $s>0$). 
Then, by~\cite[Proposition 3.2.(e)]{Wu2004},
$$\beta_w(P_t^{\mathscr D})\le \beta_w(P_{\delta_C}^{\mathscr D})\beta_w(P_s^{\mathscr D})=0\times \beta_w(P_s^{\mathscr D})=0 .$$
 We extend this by induction to get that $\beta_w(P_t^{\mathscr D})=0$ for all $t>0$, the desired result. Hence, 
 $$\mathsf r_{ess}(P_t^{\mathscr D}|_{b B(\mathscr D)}) =0, \ t>0.$$
We now prove that  $P_t^{\mathscr D}: bB(\mathscr D)\to bB(\mathscr D)$  is compact  for $t>0$. Again, we consider  $P_t^{\mathscr D}$ as a non-negative bounded linear kernel over $\mathbb R^{2d}$ and proves that it   is compact  for $t>0$.  We recall that  for all $t>0$,    $\beta_w(P_t^{\mathscr D})=0$ and, for all $K\Subset \mathbb R^{2d}$,  $\beta_\tau (\mathbf 1_{K}(P_t^{\mathscr D})) \le \beta_\tau (\mathbf 1_{K}(P_t))=0$ (since $P_t$ satisfies \textbf{(A1)}). 
Let  $t>0$ and write $t=3u$,  with $u>0$. We can thus use~\cite[Proposition 3.2 (f)]{Wu2004}, to deduce that
$$\beta_\tau (P^{\mathscr D}_{2u})\le \beta_w (P^{\mathscr D}_{u})\beta_\tau (P^{\mathscr D}_{u})=0.$$
 Consequently, thanks to~\cite[Proposition 3.2 (g)]{Wu2004}, $P^{\mathscr D}_{t}=P^{\mathscr D}_{u} P^{\mathscr D}_{2u}$ is compact over $bB (\mathbb R^{2d})$ (and thus on $bB (\mathscr D)$). 
 The proof of the theorem is complete.
 \end{proof}

%
%
%

\section{Topological irreducibility for the killed and non-killed processes}
\noindent

In this section we prove that the killed semigroup of~\eqref{eq.Lan} is   topologically irreducible over $\mathscr D$.
%
 
%
%


%

 \begin{theorem}\label{th.C5-infty} 
 Consider~\eqref{eq.Lan} driven by a pure jump  Lévy process $(L_t,t\ge 0)$  with Lévy measure $\nu$ having full topological support\footnote{We recall that $z$ belongs to the topological support of a measure $\mu$ on a metric space $E$ if and only if for any open ball  $\mathscr B$  of $E$ centered at $z$, $\nu(\mathscr B)>0$.},namely,
 \begin{equation}\label{eq.supp}
   {\rm supp}(\nu)=\mathbb R^d\setminus \{0\} .
  \end{equation}   
  We also assume that 
 \begin{equation}\label{eq.supp2}
   \sup_{\delta\in (0,1]} \big| \int_{\delta\le |z|\le 1}z\nu(dz)\big|<+\infty.
  \end{equation} 
 Assume in addition that either {\rm\textbf{[{\footnotesize B$_{\text{smooth}}$}]}} holds, or that the process $(L_t,t\ge0)$ satisfies {\rm\textbf{[{\footnotesize C$_{\text{LG}}$}]}}.  Let $\mathscr D=\mathscr O \times \mathbb R^d$ where $\mathscr O$ is a subdomain  (i.e. a nonempty, open,  and connected subset)   of $\mathbb R^d$. 
  Then,  $(P_t^{\mathscr D},t>0)$ is topologically irreducible over $\mathscr D$ (see~\eqref{eq.PTD}). If, in addition, $\mathbb R^{d}\setminus \bar{\mathscr O}$ is nonempty, then for all $\mathsf x\in \mathscr D$, $\mathbb P_{\mathsf x}[\sigma_{\mathscr D}<+\infty]>0$. 
 \end{theorem}

 Condition \eqref{eq.supp2} is a technical requirement, analogous to the conditions for the support theorem in \cite{kulyk22} for globally Lipschitz drifts, which arises naturally from the decomposition of the Lévy process into small and large jump components in the proof of Proposition~\ref{pr.C5-infty}. Specifically, condition \eqref{eq.supp2} is satisfied under either of the following scenarios: 
\begin{itemize}
\item[-] Symmetry:  the Lévy measure $\nu$ is symmetric, in which case the truncated mean vanishes, i.e., $\int_{\delta < |z| \le 1} z \nu(dz) = 0$ for all $\delta \in (0,1]$.
\item[-]  Finite variation:   $\int_{|z| \le 1} |z| \nu(dz) < +\infty$.
\end{itemize}
Notably, both conditions \eqref{eq.supp} and \eqref{eq.supp2} are satisfied for RI$\alpha$S processes for the entire range $\alpha \in (0,2)$.

\medskip

\noindent
\textbf{Note}.  Notice that the Theorem  \ref{th.C5-infty} applies when $\mathscr O=\mathbb R^d$.  As a consequence, the semigroup $(P_t,t>0)$ is topologically irreducible over $\mathbb R^{2d}$.  In addition, we prove in Theorem~\ref{th.ex2} below that     {\rm\textbf{[{\footnotesize C$_{\text{LG}}$}]}} is  satisfied when \textbf{[{\footnotesize H$^{\alpha\in (1,2)}_{\text{LG}}$}]} holds. 
Thus, when    \textbf{[{\footnotesize H$^{\alpha\in (1,2)}_{\text{LG}}$}]} holds, $(P_t^{\mathscr D},t>0)$ is topologically irreducible over $\mathscr D$.  
Moreover, when {\rm\textbf{[{\footnotesize B$_{\text{smooth}}$}]}} holds,  \eqref{eq.Lan} is strongly well-posed for any Lévy process $L$ and  the (weak) continuity of the trajectories of \eqref{eq.Lan} w.r.t. the initial condition $\mathsf x\in \mathbb R^{2d}$  is straightforward to deduce. Indeed,  by Grönwall inequality, it holds a.s. 
\begin{equation}\label{eq.cbn}
\forall T\ge 0, \ \  \lim_{\mathsf y\to \mathsf x} \ \sup_{t\in [0,T]}\ |X_t(\mathsf x)-X_t(\mathsf y)|=0.
\end{equation}
Consequently, when {\rm\textbf{[{\footnotesize B$_{\text{smooth}}$}]}} holds and when the pure jump Lévy process $(L_t,t\ge 0)$ satisfies \eqref{eq.supp}-\eqref{eq.supp2},   $(P_t^{\mathscr D},t>0)$ is topologically irreducible over $\mathscr D$.  Theorem \ref{th.C5-infty}   will be extended below when \textbf{[{\footnotesize H$^{\alpha\in (1,2)}_{\text{P-Grad}}$}]} holds, see Theorem \ref{th.sta}. 
\medskip

\begin{remark}
We emphasize that assumption \eqref{eq.supp} is not exhaustive, and the results can be extended to broader classes of Lévy processes. For instance, consider the case where $d=kn$ and $L_t := (L_t^1, \dots, L_t^k)$ is a (cylindrical) L\'evy process whose components $L_t^i$ are independent $n$-dimensional L\'evy processes with L\'evy measures satisfying \eqref{eq.supp} and \eqref{eq.supp2}. 

In this setting, the L\'evy measure of the joint process $(L_t, t \ge 0)$ is supported on the union of the $k$ subspaces defined by 
\[
\bigoplus_{j=1}^k V_j, \quad \text{where } V_j = \{0\}^{(j-1)n} \times \mathbb{R}^n \times \{0\}^{(k-j)n}.
\]
The assertions of Theorem~\ref{th.C5-infty} remain valid in this framework. Indeed, the result follows by adapting the proof of Proposition~\ref{pr.C5-infty} via standard, albeit notationally dense, modifications to account for the component-wise jump structure.
\end{remark}   
 
Let  us also mention~\cite{zhangAIHP2020}, where topological irreducibility is
established for two classes of SDEs with low
regularity coefficients: elliptic jump-diffusions (using Girsanov transformation), and elliptic SDEs driven by a rotationally invariant $\alpha$-stable
process with $\alpha\in(1,2)$ (using density estimates for the Dirichlet heat
kernel). We also refer to~\cite[Section 4]{kulik2009} and references therein, for similar results.  When $\mathbf B$ is locally Lipschitz on $\{|\mathbf B|\neq +\infty\}$,  the topological irreducibility of \eqref{eq.Lan} driven by RI$\alpha$S processes was investigated in~\cite[Proposition 3.2]{bao2024exponential} using an elegant framework based on subordinated Brownian motion.   However, it seems to us that the final estimation involves a supremum limit that is technically inconsistent with the jump nature of the subordinator. Specifically,  the continuous functions $u_{\ell_s^\delta}$ there cannot converge (locally) uniformly in time as $\delta\to 0$ to the discontinuous function $u_{\ell_s}$. Nevertheless, this gap can probably be filled using different arguments.


 To prove Theorem~\ref{th.C5-infty}, we need  two  intermediate results.

 \begin{proposition}\label{pr.Co} 
Consider~\eqref{eq.Lan} driven by a pure jump  Lévy process $(L_t,t\ge 0)$.  
 Assume in addition that either {\rm\textbf{[{\footnotesize B$_{\text{smooth}}$}]}} holds, or that the process $(L_t,t\ge0)$ satisfies {\rm\textbf{[{\footnotesize C$_{\text{LG}}$}]}}.  
  Then, for every open ball $\mathscr B\Subset \mathscr D$ and $t\ge 0$, 
\begin{equation}\label{eq.Map}
\mathsf x \in \mathscr D\mapsto P_{t}^{\mathscr D}  (\mathsf x, \mathscr B) \text{  is lower semicontinuous}. 
\end{equation}
 \end{proposition}

\begin{proof}
Let $\mathsf{x}_n=(x_n,v_n) \to \mathsf{x}=(x,v)\in \mathscr D$. Let $t\ge 0$.  
To prove the lower semicontinuity, we have to show that:
\begin{equation}\label{eq.LMm}
\liminf_{n \to \infty} \mathbb{P}\big [t < \sigma_{\mathscr D}(\mathsf x_n), X_t(\mathsf x_n) \in \mathscr B\big ] \ge \mathbb{P}\big [t < \sigma_{\mathscr D}(\mathsf x), X_t(\mathsf x) \in \mathscr B\big ].
\end{equation}
When {\rm\textbf{[{\footnotesize B$_{\text{smooth}}$}]}} holds, 
then, by \eqref{eq.cbn}, $
X(\mathsf x_n) \to X(\mathsf x) \quad \text{a.s. in } D([0,+\infty), \mathbb{R}^{2d})$. 
On the other hand, when  $(L_t,t\ge0)$ satisfies {\rm\textbf{[{\footnotesize C$_{\text{LG}}$}]}},  since the   quantities  appearing in \eqref{eq.LMm} only involve    the laws of the processes separately, using   {\rm\textbf{[{\footnotesize C$_{\text{LG}}$}]}}  and the Skorokhod representation theorem, we may assume  that:
$$
X(\mathsf x_n) \to X(\mathsf x) \quad \text{a.s. in } D([0,+\infty), \mathbb{R}^{2d}).
$$
Moreover,  because  $(x_s(\mathsf x_n),s\ge 0) \to (x_s(\mathsf x),s\ge 0)$ a.s. in $ D([0,+\infty), \mathbb{R}^{d})$, 
and both $(x_s(\mathsf x_n),s\ge 0)$ and $(x_s(\mathsf x),s\ge 0)$ are continuous, it follows that 
\begin{equation}\label{eq.Jm}
x(\mathsf x_n) \to x(\mathsf x) \quad \text{a.s. in } C([0, \infty), \mathbb{R}^{d}).
\end{equation}
Write $\mathscr B=\mathscr B_{\mathbb R^{2d}}(\mathsf x',r)$, for some $r>0$ and  $\mathsf x'=(x',v')\in \mathscr D$. 
Let us prove that a.s. it holds:
\begin{equation}\label{eq.infL}
\liminf_{n \to \infty} \mathbf 1_{t < \sigma_{\mathscr D}(\mathsf x_n), |X_t(\mathsf x_n)-\mathsf x'|<r} \ge \mathbf 1_{t < \sigma_{\mathscr D}(\mathsf x),  |X_t(\mathsf x)-\mathsf x'|<r} 
\end{equation}
Indeed if~\eqref{eq.infL} holds, we conclude the proof of the proposition applying Fatou’s Lemma.  
Let us show~\eqref{eq.infL}. To this end, it suffices to prove that on the event   $\{t < \sigma_{\mathscr D}(\mathsf x), X_t(\mathsf x) \in \mathscr B\}$, there exists a random integer $m$ such that for all $n\ge m$, 
$$\mathbf 1_{t < \sigma_{\mathscr D}(\mathsf x_n), |X_t(\mathsf x_n)-\mathsf x'|<r}=1.$$  
Let us thus work on  the event $\{t < \sigma_{\mathscr D}(\mathsf x), X_t(\mathsf x) \in \mathscr B\}$. Then,  because $x_{[0,t]}(\mathsf x)\subset \mathscr O$, we have by~\eqref{eq.Jm}, that there exists $m\ge 0$ (random) such that  for every $n\ge m$, $x_{[0,t]}(\mathsf x_n)\subset \mathscr O$. Thus for $n\ge m$, 
$$t < \sigma_{\mathscr D}(\mathsf x_n).$$
On the other hand, by  \cite[Proposition 5.2 in Chapter 3]{EK},   for every  $u \in \mathscr C_v:=\{s> 0, \mathbb P_{\mathsf x}[v_s=v_{s_-}]=1\}$, we have a.s. that 
$$v_u(\mathsf x_n)\to v_u(\mathsf x).$$
Since $v_s(\mathsf x)$ is driven by a Lévy process, $\mathscr C_v=\mathbb R_+$. Consequently, we have a.s. $v_t(\mathsf x_n)\to v_t(\mathsf x)$. Because also 
 $x_t(\mathsf x_n)\to x_t(\mathsf x)$, up to increasing $m$, we have for all  $n\ge m$, 
$$|X_t(\mathsf x_n)- \mathsf x'|<r.$$
We have thus proved~\eqref{eq.infL}. The proof of the proposition is complete.  
\end{proof}

 \begin{proposition}\label{pr.C5-infty} 
 Let the assumptions of Theorem~\ref{th.C5-infty} hold. 
 Let $\mathsf x_0,\mathsf x_F\in \mathscr D$. 
Then, for  any  $\boldsymbol\epsilon>0$, there exists $t_{\boldsymbol\epsilon}>0$ such that for all $t\in (0,t_{\boldsymbol\epsilon}]$ and all $\mathsf x \in \{\mathsf x_0\}\cup \bar{\mathscr B}_{\mathbb R^{2d}}(\mathsf x_F,\boldsymbol\epsilon) $,
 $$ P_{t}^{\mathscr D} \big (\mathsf x, \mathscr B_{\mathbb R^{2d}}(\mathsf x_F,\boldsymbol\epsilon)  \big )>0.$$ 
 \end{proposition}
 \noindent
\textbf{Note.} From the proof of Proposition~\ref{pr.C5-infty}, one easily shows that if  $\mathbb R^{d}\setminus \bar{\mathscr O}$ is nonempty, then for all $\mathsf x\in \mathscr D$, $\mathbb P_{\mathsf x}[\sigma_{\mathscr D}<+\infty]>0$.  
\medskip

Let us emphasize, even if we   did not make it explicit it in the notation, that the   time $t_{\boldsymbol\epsilon}>0$ also depends on $\mathsf x_0,\mathsf x_F\in \mathscr D$.
 With this result in hand, we now prove Theorem~\ref{th.C5-infty}.

 \begin{proof}(Theorem~\ref{th.C5-infty})
 Let $\mathsf x_0,\mathsf x_F\in \mathscr D$ and $\boldsymbol\epsilon>0$. There is no loss of generality in assuming  that   $\boldsymbol\epsilon>0$ is such that $\mathscr B_{\mathbb R^{2d}}(\mathsf x_F,\boldsymbol\epsilon)\Subset  \mathscr D$.
 Let $K_{\boldsymbol\epsilon}:= \bar{\mathscr B}_{\mathbb R^{2d}}(\mathsf x_F,\boldsymbol\epsilon)\cup\{\mathsf x_0\}$. 
 Let $t=t_{\boldsymbol\epsilon}+s$, $s\in (0,t_{\boldsymbol\epsilon}]$, where $t_{\boldsymbol\epsilon}>0$ is as in Proposition~\ref{pr.C5-infty}.  By the Markov property of the process~\eqref{eq.Lan}, we have that 
\begin{align*}
 P_{t}^{\mathscr D} \big (\mathsf x_0, \mathscr B_{\mathbb R^{2d}}(\mathsf x_F,\boldsymbol\epsilon)  \big )&= \int_{\mathscr D}  P_s^{\mathscr D}(\mathsf y, \mathscr B_{\mathbb R^{2d}}(\mathsf x_F,\boldsymbol\epsilon) )P_{t_{\boldsymbol\epsilon}}^{\mathscr D}(\mathsf x_0, d\mathsf y) \\
 &\ge P_{t_{\boldsymbol\epsilon}}^{\mathscr D} \big (\mathsf x_0, \mathscr B_{\mathbb R^{2d}}(\mathsf x_F,\boldsymbol\epsilon  )\big )\\
 &\quad \times \inf_{\mathsf y_0\in \bar{\mathscr B}_{\mathbb R^{2d}}(\mathsf x_F,\boldsymbol\epsilon  )} P_{s}^{\mathscr D} \big (\mathsf y_0, \mathscr B_{\mathbb R^{2d}}(\mathsf x_F,\boldsymbol\epsilon)\big ).
 \end{align*}
Using  Proposition~\ref{pr.C5-infty} and because   $\mathsf x_0,\mathsf x_F\in K_{\boldsymbol\epsilon}$, 
$$P_{t_{\boldsymbol\epsilon}}^{\mathscr D} \big (\mathsf x_0, \mathscr B_{\mathbb R^{2d}}(\mathsf x_F,\boldsymbol\epsilon )\big )>0.$$
 On the other hand,   by  Proposition~\ref{pr.C5-infty} and since $s\in (0,t_{\boldsymbol\epsilon}]$, we have  for all $\mathsf y_0\in \bar{\mathscr B}_{\mathbb R^{2d}}(\mathsf x_F,  \boldsymbol\epsilon  )\subset K_{\boldsymbol\epsilon}$, 
 $$  P_{s}^{\mathscr D}  (\mathsf y_0, \mathscr B_{\mathbb R^{2d}}(\mathsf x_F,\boldsymbol\epsilon)  ) >0.$$ 
 Moreover, since   $\mathsf x\mapsto  P_{s}^{\mathscr D}  (\mathsf x, \mathscr B_{\mathbb R^{2d}}(\mathsf x_F,\boldsymbol\epsilon))$  is lower semicontinuous (see Proposition~\ref{pr.Co}), we obtain that
 $$\inf_{\mathsf y_0\in \bar{\mathscr B}_{\mathbb R^{2d}}(\mathsf x_F,\boldsymbol\epsilon  )} P_{s}^{\mathscr D} \big (\mathsf y_0, \mathscr B_{\mathbb R^{2d}}(\mathsf x_F,\boldsymbol\epsilon)\big )>0,$$ 
 and therefore $P_{t}^{\mathscr D}   (\mathsf x_0, \mathscr B_{\mathbb R^{2d}}(\mathsf x_F,\boldsymbol\epsilon)) >0$. We can iterate this argument to obtain that for every $t>0$,  
 $$P_{t}^{\mathscr D}   (\mathsf x_0, \mathscr B_{\mathbb R^{2d}}(\mathsf x_F,\boldsymbol\epsilon)) >0.$$ This concludes the proof of Theorem~\ref{th.C5-infty}. 
 \end{proof}

We now prove  Proposition  \ref{pr.C5-infty}.

 \begin{proof}(Proposition~\ref{pr.C5-infty})   
The proof is organized as follows. First,  we decompose the Lévy process into its small-jump and big-jump components by introducing a threshold parameter $\delta > 0$.  The second step is concerned with  some  preliminary analysis. Given $\boldsymbol\epsilon>0$, $\mathsf x_0,\mathsf x_F\in \mathscr D$, the subsequent objective is to construct specific events (having nonzero probabilities) on which the process - starting from any initial condition $\mathsf{x} \in \{\mathsf{x}_0\} \cup \bar{\mathscr{B}}_{\mathbb{R}^{2d}}(\mathsf{x}_F, \boldsymbol{\epsilon})$ - reaches the target ball $\mathscr{B}_{\mathbb{R}^{2d}}(\mathsf{x}_F, \boldsymbol{\epsilon})$ within a time $t \in (0, t_{\boldsymbol{\epsilon}}]$ without exiting the domain $\mathscr{D}$.  The construction is divided into two distinct cases based on whether  $\mathsf x \in \bar{\mathscr B}_{\mathbb R^{2d}}(\mathsf x_F,\boldsymbol\epsilon)$ or $\mathsf x=\mathsf x_0$. These constructions are detailed in the  final  two steps of the proof. 
\medskip

\noindent
\textbf{Step 1}. Decomposition of   the Lévy process. For every  threshold $\delta\in (0,1]$
and $t\ge 0$, we decompose $L_t$ as
$
L_t = L_t^{+}(\delta) + L_t^{-}(\delta) - \mathbf c_\delta \, t
 $
where  
\[
L_t^{-}(\delta) =  \int_0^t \int_{|z| \le \delta} z \tilde{N}(ds, dz), \ 
L_t^{+}(\delta) =\int_0^t \int_{|z| > \delta} z N(ds, dz) 
$$
and 
$$ \mathbf c_\delta  :=  \int_{\delta < |z| \le 1} z\,  \nu(dz).
\]
The parameter $\delta\in (0,1]$ will be chosen and fixed in the subsequent steps, solely to guarantee that the constructed events have nonzero probability.

The Lévy processes $L^{-}(\delta)$   and $L^{+}(\delta)$ are  independent, and have Lévy measures given respectively by  $\nu(\cdot \cap \{|z|\le\delta\})$ and $\nu(\cdot \cap \{|z|>\delta\})$. 
Here \(N \) is the Poisson random measure associated with $(L_t,t\ge 0)$, and \(\tilde N \) is its compensated version. The process \(L^{-}(\delta)\) is a càdlàg martingale with finite second moment for all $t\ge 0$. 
By Doob's   maximal inequality for martingales, we  have for $\beta >0$, $\delta>0$,  and $t\ge 0$, 
$$
\mathbb P\Big[ 
\sup_{s\in[0,t]} |L_s^{-}(\delta)| > \beta 
\Big ]
\le \frac{\mathbb E\bigl[|L_t^{-}(\delta)|^2\bigr]}{|\beta |^2}
=
\frac{t}{|\beta |^2}
\int_{|z|\le\delta} |z|^2\,\nu (dz),
$$
where the last equality follows from  the Itô isometry for Poisson measures. 
Hence, for every $\beta >0$,  there exists  $\delta_{\beta }(1)>0$  such that  for all $0< \delta \le \delta_{\beta }(1)$ and $t\in [0,1]$,  the event
\begin{equation}\label{eq.e-}
\mathscr E^{-}(\delta,t,\beta ):=\Big  \{ \sup_{0\le s\le t} |L_s^{-}(\delta)|\le \beta \Big \} \text{ has positive probability}.
\end{equation}

\noindent
\textbf{Step 2}. Preliminary analysis.  To  only use Grönwall's inequality~\cite[Section 5 in Appendices]{EK} on the velocity component $(v_s,s\ge 0)$,  it is worth noting that we may assume without loss of generality (up to a localization argument in the position variable which is detailed just below) that 
\begin{equation}\label{eq.Hgp}
\exists C_{\mathbf B}>0, \forall (x,v)\in \mathbb R^{2d}, \ |\mathbf B(x,v)|\le C_{\mathbf B}[1+|v|].
\end{equation}
Indeed, let $\mathsf x_0,\mathsf x_F\in \mathscr D$. Consider a bounded subdomain $\mathscr O^*$ of $\mathscr O$ containing 
$ x_0$ and $ x_F$, with $\bar{\mathscr O}^*\subset \mathscr O$.
Set for $\mathsf x=(x,v)\in  \mathbb R^{2d}$,
$$\mathbf B^*(\mathsf x)=\mathbf 1_{x\in \bar{\mathscr O}^*}\mathbf B(\mathsf x).$$
Note that $\mathbf B^*$ satisfies \eqref{eq.Hgp} since $\mathbf B$ satisfies \textbf{[{\footnotesize B$_{\text{LG}}$}]} and since $\mathscr O^*$ is bounded. 
Set $\mathscr D^*=\mathscr O^*\times \mathbb R^d$.
Then, consider the unique weak solution $(X^*_t=(x^*_t,v^*_t),t\ge 0)\in \mathbb R^d \times \mathbb R^d $     to the kinetic Langevin equation (recall we assumed \textbf{[{\footnotesize W$_{\text{LG}}$}]}): 
$$
 dx^*_t=v^*_tdt, \ dv^*_t=\mathbf B^*(x^*_t,v^*_t)dt   +   dL_t.
$$
Set $\sigma_{\mathscr D^*}^*:=\inf\{t\ge 0: x_t^*\notin \mathscr O^*\}$. 
Since $\mathbf B^*= \mathbf B$ on  $\bar{\mathscr D}^*$, by Itô's formula~\cite[Theorem 4.4.7]{applebaum2009levy},  for $\mathsf x\in \mathbb R^{2d}$ with $x\in \mathscr O^*$, 
  $X_{\cdot \wedge \sigma_{\mathscr D^*}}(\mathsf x)$ and $X^*_{\cdot\wedge \sigma_{\mathscr D^*}^*}(\mathsf x)$ solve the stopped martingale problem  $(\mathcal L, \mathsf x,  \mathscr D^*)$, where $\mathcal L$ is the infinitesimal generator of \eqref{eq.Lan}. This stopped martingale problem   is well-posed  by \textbf{[{\footnotesize W$_{\text{LG}}$}]}, Proposition~\ref{pr.Pre} and~\cite[Theorem 6.1 in Chapter 4]{EK}. Hence,  we have for all $\mathsf x=(x,v)$ with  $x\in \mathscr O^*$,
$$
(X_{t\wedge \sigma_{\mathscr D^*}}(\mathsf x),t\ge 0)\overset{\text{law}}{=}(X^*_{t\wedge \sigma_{\mathscr D^*}^*}(\mathsf x),t\ge 0).
$$
As a result
and because $\sigma_{\mathscr D^*}\le \sigma_{\mathscr D}$, 
 for every $\boldsymbol\epsilon>0$, we have, for all $\mathsf x=(x,v)$ with  $x\in \mathscr O^*$,
\begin{align*}
P_{t}^{\mathscr D} \big (\mathsf x, \mathscr B_{\mathbb R^{2d}}(\mathsf x_F,\boldsymbol\epsilon)  \big )&\ge P_{t}^{\mathscr D^*} \big (\mathsf x, \mathscr B_{\mathbb R^{2d}}(\mathsf x_F,\boldsymbol\epsilon)  \big ) =\mathbb P_{\mathsf x}\big [t<\sigma^*_{\mathscr D^*}, X_t^*\in \mathscr B_{\mathbb R^{2d}}(\mathsf x_F,\boldsymbol\epsilon)\big ]. 
\end{align*}
Thus, up to considering  $X^*$ instead of $X$ and working over $\mathscr D^*$ instead of $\mathscr D$, we may assume, without loss of generality, that \eqref{eq.Hgp} is satisfied, and we will do so throughout the rest of the proof. 

We further incorporate the deterministic drift correction $ \mathbf c_\delta $ into the term $\mathbf{B}$ by setting:
$$\mathbf{B}_0:= \mathbf{B} - \mathbf c_\delta.$$
By virtue of \eqref{eq.supp2},  the upper bound \eqref{eq.Hgp} is still satisfied with a constant, still denoted by  $C_{\mathbf B}>0$, which is   independent of the threshold $\delta \in (0,1]$.  With a slight abuse  of notation we still denote $\mathbf B_0$ by $\mathbf B$. 

 A direct consequence of  \eqref{eq.Hgp} is as follows. Let  $[a,b]\subset \mathbb R_+$ and $\mathsf x\in \mathscr D$.   
For $s\in [a,b]$
\begin{equation}\label{eq.Klp}
v_s(\mathsf x )= v_{a}(\mathsf x) + \int_{a}^s \mathbf B(X_u(\mathsf x))du+ 
 L^{-}_s(\delta)-L^{-}_{a}(\delta)+  L^{+}_s(\delta)-L^{+}_{a}(\delta).
 \end{equation}
Then, by Grönwall's inequality and \eqref{eq.Hgp}, we have for all  $s\in [a,b]\subset [0,t]$ and on $\mathscr E^{-}(\delta,t,\beta)$ (see \eqref{eq.e-}),  
$$|v_s(\mathsf x) | \le \Big[ |v_{a}(\mathsf x) | + C_{\mathbf B}(b-a) + 2\beta + \sup_{u \in [a, b]} |L^{+}_u(\delta)-L^{+}_{a}(\delta)|  \Big] e^{C_{\mathbf B}(b-a)},$$
where $C_{\mathbf B}>0$ is the constant from \eqref{eq.Hgp}.
Hence, coming back to \eqref{eq.Klp} and using again \eqref{eq.Hgp},  one gets  that for $s\in [a,b]$:
$$
 v_s(\mathsf x)=v_{a}(\mathsf x) +  L^{+}_s(\delta)-L^{+}_{a}(\delta) + \mathbf e (s),
 $$
 where 
 \begin{align}
 \nonumber
  |\mathbf e(s)| & \le C_{\mathbf B}(b-a)+ C_{\mathbf B}(b-a) \, e^{C_{\mathbf B}(b-a)} \Big[ |v_{a}(\mathsf x) | + C_{\mathbf B}(b-a) + 2\beta \\
  \label{eq.E}
  &\quad + \sup_{u \in [a, b]} |L^{+}_u(\delta)-L^{+}_{a}(\delta)|  \Big] +  2\beta. 
 \end{align} 
  Let $\mathsf x_0=(x_0,v_0),\mathsf x_F=(x_F,v_F)\in \mathscr D$.
 Consider a continuous curve $\gamma_0: [0,1]\to \mathscr O$ with $\gamma_0(0)=x_0$ and  $\gamma_0(1)=x_F$. 
 Let $\boldsymbol\epsilon_0>0$ such that  
  \begin{equation}\label{eq.h1}
 \forall \boldsymbol\epsilon\in (0,\boldsymbol\epsilon_0]: \    \bar{\mathscr B}_{\mathbb R^{d}}( x_0, 4\boldsymbol\epsilon) \subset {\mathscr O} \text{ and } \text{Ran}_{4\boldsymbol\epsilon}(\gamma_0) \subset \mathscr O, 
  \end{equation}
  where $\text{Ran}_{4\boldsymbol\epsilon}(\gamma_0)$ is the $4\boldsymbol\epsilon$-closed neighborhood of the range $\text{Ran}(\gamma_0)$ of $\gamma_0$ (see Figure~\ref{fig.R}). 
  Let $\boldsymbol\epsilon>0$.   In what follows we assume that  $\boldsymbol\epsilon\in (0,\boldsymbol\epsilon_0]$. We can always only consider such $\boldsymbol\epsilon>0$ since if $\boldsymbol\epsilon >\boldsymbol\epsilon_0$, we write
  $$ P_{t}^{\mathscr D} \big (\mathsf x_0, \mathscr B_{\mathbb R^{2d}}(\mathsf x_F,\boldsymbol\epsilon)  \big )\ge P_{t}^{\mathscr D} \big (\mathsf x_0, \mathscr B_{\mathbb R^{2d}}(\mathsf x_F,\boldsymbol\epsilon_0 ) \big)  .$$
To prove Proposition~\ref{pr.C5-infty}, the goal is to show that there   exists $t_{\boldsymbol\epsilon}>0$ such that for all $t\in (0,t_{\boldsymbol\epsilon}]$ and all   $\mathsf x \in \{\mathsf x_0\}\cup \bar{\mathscr B}_{\mathbb R^{2d}}(\mathsf x_F,\boldsymbol\epsilon) $, it holds:
\begin{equation}\label{eq.P>}
 \mathbb P\Big[ |X_t(\mathsf x)- \mathsf x_F|\le \boldsymbol\epsilon \text{ and } x_s(\mathsf x)\in \mathscr O \text{ for all $s\in [0,t]$}     \Big]>0.
\end{equation}
In  what follows  $c_{\boldsymbol\epsilon_0}>0$ is a constant such that 
\begin{equation}\label{eq.COF}
\text{$\forall \mathsf x  \in  \bar{\mathscr B}_{\mathbb R^{2d}}(\mathsf x_F, 4\boldsymbol\epsilon_0) $ and $\forall y\in  \text{Ran}_{4\boldsymbol\epsilon_0 }(\gamma_0) $, $|\mathsf x| + |y|\le c_{\boldsymbol\epsilon_0}$}.
\end{equation}
 We now turn to the proof of \eqref{eq.P>}. 
  \medskip 

\noindent
\textbf{Step 3 (Case A)}. Proof of \eqref{eq.P>}  when $\mathsf x=(x,v)\in \mathscr B_{\mathbb R^{2d}}(\mathsf x_F,\boldsymbol\epsilon) $.
\medskip

\noindent
When  $|\mathsf x-\mathsf x_F|\le \boldsymbol\epsilon$,  it holds $| x- x_F|\le \boldsymbol\epsilon$  and thus (see \eqref{eq.h1}):
\begin{equation}\label{eq.seg}
[x,x_F]\subset \mathscr B(x_F, \boldsymbol\epsilon) \Subset \mathscr B(x_F,2 \boldsymbol\epsilon) \Subset \mathscr O.
\end{equation}
Here $[x,x_F]=\{ux_F+ (1-u)x, u\in [0,1]\}$ is the  line segment joining $x$ to $x_F$.  
 Consider  $t_0:=0<t_1<t_2<t$. Set  $\pmb{\Delta}_0=t_1$, $\pmb{\Delta}_1=t_2-t_1$, and $\pmb{\Delta}_2=t-t_2$.  
 For  $x,y\in \mathbb R^d$, we define the velocity  to go from $x$ to $y$ during the interval  time $\pmb{\Delta}$ by:
  \begin{equation}\label{eq.lmp}
  v_{x\to y}(\pmb{\Delta})= \frac{y-x}{\pmb{\Delta}}.
 \end{equation} 
For ${\mathsf x}\in \mathbb R^{2d}$, define the following events:
   \begin{align*}
   \mathscr E^{+}_{\text{jump}}(\delta, t_0,t_1,\beta)&=  \Big \{\big (L^{+}_s(\delta),s\in (0, t_1)\big ) \text{ has exactly one jump}\\
&\quad \text{ and this jump lies in } \mathscr B_{\mathbb R^d}(v_{x\to x_F}(\pmb{\Delta}_1)-v,\beta) \Big \} ,
   \end{align*}
 $$
 \mathscr E^{+}_{\text{no jump}}(\delta,t_1,t_2) := \Big \{\big (L^{+}_s(\delta)-L^{+}_{t_1}(\delta),  s\in [t_1,t_2]\big )\text{ has no jumps}\Big \},$$
and 
\begin{align*}
\mathscr E^{+}_{\text{jump}}(\delta, t_2,t,\beta) &=\Big  \{\big (L^{+}_s(\delta)-L^{+}_{t_2}(\delta),  s\in (t_2,t)\big ) \text{ has exactly one jump}\\
&\quad \text{ and this jump lies in }  \mathscr B_{\mathbb R^d}(v_F-v_{t_2}(\mathsf x),\beta) \Big \}.
\end{align*} 
Note that the center of the ball in the last event   is random. 
Finally, set 
  \begin{align*}
 \mathscr E^{+}_{\textbf{A}} (\mathsf x,\delta,t_1,t_2,t,\beta) &:=    \mathscr E^{+}_{\text{jump}}(\delta, t_0,t_1,\beta)     \pmb{\cap}   \mathscr E^{+}_{\text{no jump}}(\delta,t_1,t_2)   \pmb{\cap}\mathscr E^{+}_{\text{jump}}(\delta, t_2,t,\beta).
\end{align*}  
We now claim that for  every $0<t_1<t_2<t\le 1$ and  all  $0<\delta \le \frac \beta2\wedge 1$, it holds:
$$\mathbb P\Big [\mathscr E^{+}_{\textbf{A}} (\mathsf x,\delta,t_1,t_2,t,\beta)\Big]>0.$$  
We now prove this claim.   Using the independence between the increments of the Lévy process $(L^{+}_s(\delta),s\ge 0)$, it holds:
$$\mathbb P\big [\mathscr E^{+}_{\textbf{A}} (\mathsf x,\delta,t_1,t_2,t,\beta)\big ] =    \mathbb{E} \left[ \mathbf{1}_{  \mathscr E^{+}_{\text{jump}}(\delta, t_0,t_1,\beta)    \pmb{\cap}   \mathscr E^{+}_{\text{no jump}}(\delta,t_1,t_2)} \,  \Psi (v_{t_2}(\mathsf x)) \right],$$
where $\Psi (y)=\mathbb P\big [\mathscr E^{+}_{\text{jump}}(\delta, t_2,t,\beta) \, |\,  v_{t_2} (\mathsf x)= y\big ]$, i.e. 
$$\Psi (y) = \underbrace{  \lambda_\delta \, \pmb{\Delta}_2 e^{-\lambda_\delta  \pmb{\Delta}_2} }_{\text{Prob. of exactly one jump}} \ \underbrace{\frac{\nu \left( \mathscr B_{\mathbb R^d}(v_F - y, \beta) \cap \{ |z| > \delta \} \right)}{\lambda_\delta }}_{\text{Prob. jump size is in the ball}}, $$
where $\lambda_\delta := \nu(\{|z| > \delta\})$. 
Because $0<\delta\le \frac \beta2\wedge 1$,    we have for every $y\in \mathbb R^d$, 
$$\nu(\mathscr B_{\mathbb R^d}(y,\beta)\cap \{|z|>\delta\})>0.$$ 
Hence, $\Psi (y)>0$ for every possible velocity state $y\in \mathbb R^d$.  Moreover, it holds because $0<\delta\le \frac \beta2\wedge 1$:
\begin{align*}
\mathbb P\Big [\mathscr E^{+}_{\text{jump}}(\delta, t_0,t_1,\beta)    \pmb{\cap}   \mathscr E^{+}_{\text{no jump}}(\delta,t_1,t_2)\Big]&=\mathbb P\Big [\mathscr E^{+}_{\text{jump}}(\delta, t_0,t_1,\beta)     \Big]\\
&\quad \times \mathbb P\Big [   \mathscr E^{+}_{\text{no jump}}(\delta,t_1,t_2)\Big]>0,
\end{align*}
where we have used that  
$$\mathbb P\Big [\mathscr E^{+}_{\text{jump}}(\delta, t_0,t_1,\beta)     \Big] 
 =   \pmb{\Delta}_0 e^{-\lambda_\delta\pmb{\Delta}_0 }\  \nu \big( \mathscr B_{\mathbb R^d}(v_{x\to x_F}(\pmb{\Delta}_1)-v, \beta)  {\cap} \{ |z| > \delta \} \big) 
$$
and
$$\mathbb P \Big [ \mathscr E^{+}_{\text{no jump}}(\delta,t_1,t_2)    \Big]  = e^{-\lambda_\delta \pmb{\Delta}_1}.$$
 As a result, 
$$\mathbb P[\mathscr E^{+}_{\textbf{A}} (\mathsf x,\delta,t_1,t_2,t,\beta)]  =    \int_{ \mathscr E^{+}_{\text{jump}}(\delta, t_0,t_1,\beta)    \pmb{\cap}   \mathscr E^{+}_{\text{no jump}}(\delta,t_1,t_2) } \Psi(v_{t_2}(\mathsf x)(\omega)) \,  \mathbb{P}(d\omega)>0,$$
 proving the claim. 

 Hence,   for every $0<t_1<t_2<t\le 1$ and if $\delta\in (0,\delta_{\beta}(1)\wedge \frac \beta2\wedge 1]$ (see \eqref{eq.e-}),  
we have since $(L^{-}_s(\delta),s\ge 0)$ and $(L^{+}_s(\delta),s\ge 0)$ are independent:
 \begin{equation}\label{eq.Proba}
 \mathbb P\Big [\mathscr E_{\textbf{A}} (\mathsf x,\delta,t_1,t_2,t,\beta) \Big ]>0,
   \end{equation}
   where  
   $$\mathscr E_{\textbf{A}} (\mathsf x,\delta,t_1,t_2,t,\beta) :=\mathscr E^{-} (\delta,t,\beta)\cap \mathscr E^{+}_{\textbf{A}} (\mathsf x,\delta,t_1,t_2,t,\beta).$$
We work on the event  $\mathscr E_{\textbf{A}} (\mathsf x,\delta,t_1,t_2,t,\beta)$ and consider the process~\eqref{eq.Lan} starting from $\mathsf x=(x,v)\in \mathscr B_{\mathbb R^{2d}}(\mathsf x_F,\boldsymbol\epsilon) $. 
   Recall we want to prove~\eqref{eq.P>}. To this end, we analyse the trajectories  of the process \eqref{eq.Lan}  on the event   $\mathscr E_{\textbf{A}} (\mathsf x,\delta,t_1,t_2,t,\beta)$ when $\mathsf x=(x,v)\in \mathscr B_{\mathbb R^{2d}}(\mathsf x_F,\boldsymbol\epsilon) $. 
   \medskip

  \noindent
\underline{On $[0,t_1]$}.
The process $(L^{+}_s(\delta),s\in (0,t_1))$ has only one jump, and this jump is in the $\beta$ neighborhood of $v_{x \to x_F}(\pmb{\Delta}_1) - v$. 
Hence, for all $s\in [0,t_1]$:   
 \begin{align}
  \label{eq.tL0}
  |L^{+}_s(\delta)|\le |v_{x\to x_F}(\pmb{\Delta}_1)-v|+\beta\le \frac{2c_{\boldsymbol\epsilon_0}}{\pmb{\Delta}_1}+ c_{\boldsymbol\epsilon_0}+ \beta.
   \end{align}
By \eqref{eq.E}, it then holds for $s\in [0,t_1]$:
 \begin{align}
 \nonumber
  \big |v_s(\mathsf x)- v  -  L^{+}_s(\delta) \big  |   &\le \mathbf r^{(0)}_\star:= C_{\mathbf B}\pmb{\Delta}_0+ C_{\mathbf B}  \, e^{C_{\mathbf B}\pmb{\Delta}_0} \Big[ \pmb{\Delta}_0c_{\boldsymbol\epsilon_0}  +  C_{\mathbf B}\pmb{\Delta}_0^2 + 2\beta \pmb{\Delta}_0 \\
  \label{eq.r0star}
  &\quad \quad \quad \quad +  2c_{\boldsymbol\epsilon_0}\frac{  \pmb{\Delta}_0}{\pmb{\Delta}_1}+ c_{\boldsymbol\epsilon_0} \pmb{\Delta}_0+  \pmb{\Delta}_0\beta \Big] +  2\beta, 
 \end{align}
 and then by the triangle inequality, \eqref{eq.r0star} and \eqref{eq.tL0}: 
  \begin{align*}
 \big |v_s(\mathsf x)|&\le   c_{\boldsymbol\epsilon_0}+\frac{2c_{\boldsymbol\epsilon_0}}{\pmb{\Delta}_1}+ c_{\boldsymbol\epsilon_0}+ \beta+ \mathbf r^{(0)}_\star.
 \end{align*}
This implies that for   $s\in [0,t_1]$:
 $$|x_s(\mathsf x)-x|\le\mathbf R^{(0)}_\star:=  2c_{\boldsymbol\epsilon_0}\frac{  \pmb{\Delta}_0}{\pmb{\Delta}_1}+ \pmb{\Delta}_0[2c_{\boldsymbol\epsilon_0}+ \beta+   \mathbf r^{(0)}_\star].$$
  Moreover, because  $L^{+}_{t_1}(\delta)= v_{x\to x_F}(\pmb{\Delta}_1)-v + z$ with $|z|\le \beta$, we get from  \eqref{eq.r0star}, 
 \begin{align}
  \label{eq.vt1}
 & \big |v_{t_1}(\mathsf x)-  v_{x\to x_F}(\pmb{\Delta}_1)|   \le \beta+ \mathbf r^{(0)}_\star.  
 \end{align}

 \noindent
\underline{On $[t_1,t_2]$}. Since  $(L^{+}_s(\delta)-L^{+}_{t_1}(\delta),s\in [t_1,t_2])$ has no jumps, we then deduce by \eqref{eq.E}, \eqref{eq.vt1}, and the triangle inequality,  that for $s\in [t_1,t_2]$:
 \begin{align}
  \nonumber
 & \big |v_s(\mathsf x)- v_{t_1}(\mathsf x) \big  | \\
  \nonumber
 &\le C_{\mathbf B}\pmb{\Delta}_1+ C_{\mathbf B}  \, e^{C_{\mathbf B}\pmb{\Delta}_1} \pmb{\Delta}_1 \Big[ |v_{t_1}(\mathsf x)|  +  C_{\mathbf B}\pmb{\Delta}_1 + 2\beta     \Big] +  2\beta \\
 \label{eq.rstar1}
 & \le \mathbf r^{(1)}_\star:= C_{\mathbf B}\pmb{\Delta}_1+  C_{\mathbf B}  \, e^{C_{\mathbf B}\pmb{\Delta}_1}  \Big[    2c_{\boldsymbol\epsilon_0} +  \pmb{\Delta}_1 (     \mathbf r^{(0)}_\star +  C_{\mathbf B}\pmb{\Delta}_1 + 3\beta     )\Big] +  2\beta.
 \end{align}
Note that, though bounded, the term  $\mathbf r^{(1)}_\star$ cannot be made arbitrary small 
because of the second term appearing in its definition (namely the term $2c_{\boldsymbol\epsilon_0}C_{\mathbf B}  \, e^{C_{\mathbf B}\pmb{\Delta}_1}$).   This explains the fact that we have to consider in the third event in the definition of $\mathscr E^{+}_{\textbf{A}} (\mathsf x,\delta,t_1,t_2,t,\beta)$ a jump around the random velocity $v_F-v_{t_2}(\mathsf x)$ (and not around the deterministic velocity $v_F-v_{x\to x_F}(\pmb{\Delta}_1)$). 
We then deduce using \eqref{eq.vt1} and the triangle inequality, that for $s\in [t_1,t_2]$:
$$|v_s(\mathsf x) -  v_{x\to x_F}(\pmb{\Delta}_1)|\le \mathbf r^{(2)}_\star:= \mathbf r^{(1)}_\star+  \beta+ \mathbf r^{(0)}_\star.$$
In addition, writing 
$$ x_s(\mathsf x) =x_{t_1}(\mathsf x)+  (s-t_1) v_{x\to x_F}(\pmb{\Delta}_1)  +   \int_{t_1}^s[v_u(\mathsf x) - v_{x\to x_F}(\pmb{\Delta}_1)] du,$$
and recalling that $|x_{t_1}(\mathsf x)-x|\le \mathbf R^{(0)}_\star $, we deduce that:  
\begin{align*}\Big |x_s(\mathsf x)- x -  \frac{s-t_1}{\pmb{\Delta}_1} (x_F-x)\Big |  &\le |x_{t_1}(\mathsf x)-x| + \int_{t_1}^{t_2} |v_s(\mathsf x) - v_{x\to x_F}( \pmb{\Delta}_1)|ds\\
 &\le \mathbf R^{(1)}_\star:= \mathbf R^{(0)}_\star + \pmb{\Delta}_1\,  \mathbf r^{(2)}_\star. 
 \end{align*}
 Recall that $[x,x_F]=\{x + \frac{s-t_1}{\pmb{\Delta}_1} (x_F-x), s\in [t_1,t_2]\}\subset \mathscr B(x_F, \boldsymbol\epsilon)$ (see \eqref{eq.seg}). The previous estimate thus implies that  $\{x_s(\mathsf x),s\in [t_1,t_2]\}$ can be made  as close as desired to the  segment $[x,x_F]$ (which lies in $\mathscr O$). 
Note that in particular, it holds setting $s=t_2$ in the previous upper bound: 
  $$ |x_{t_2}(\mathsf x)-x_F|\le  \mathbf R^{(1)}_\star,$$
  and 
\begin{equation}\label{eq.vtt2}
 |v_{t_2}(\mathsf x) |\le |v_{t_2}(\mathsf x) -  v_{x\to x_F}(\pmb{\Delta}_1)|+ |v_{x\to x_F}(\pmb{\Delta}_1)|\le  \mathbf r^{(2)}_\star+ \frac{2c_{\boldsymbol\epsilon_0}}{\pmb{\Delta}_1}.
 \end{equation}

 \noindent
\underline{On $[t_2,t]$}. The process $(L^{+}_s(\delta)-L^{+}_{t_2}(\delta),s\in (t_2,t))$ has only one jump, and this jump is in the $\beta$ neighborhood of $v_F-v_{t_2}(\mathsf x)$.   
Therefore, for $s\in  [t_2,t]$: 
$$|L^{+}_s(\delta)-L^{+}_{t_2}(\delta)|\le |v_F|+ |v_{t_2}(\mathsf x) |+\beta\le c_{\boldsymbol\epsilon_0}+ |v_{t_2}(\mathsf x) |+\beta.$$  
Then, using \eqref{eq.E} and \eqref{eq.vtt2}, it holds for every $s\in [t_2,t]$:
 \begin{align*}
&\big| v_s(\mathsf x)- v_{t_2}(\mathsf x) -( L^{+}_s(\delta)-L^{+}_{t_2}(\delta))  \big|\\
& \le C_{\mathbf B}\pmb{\Delta}_2+ C_{\mathbf B}\pmb{\Delta}_2 \, e^{C_{\mathbf B}\pmb{\Delta}_2} \Big[ 2|v_{t_2}(\mathsf x) | + C_{\mathbf B}\pmb{\Delta}_2 + 3\beta  + c_{\boldsymbol\epsilon_0} \Big] +  2\beta\\
  &\le  \mathbf r^{(3)}_\star:= C_{\mathbf B}\pmb{\Delta}_2+ C_{\mathbf B}  \, e^{C_{\mathbf B}\pmb{\Delta}_2} \Big[ 2 \pmb{\Delta}_2 \mathbf r^{(2)}_\star+ 4c_{\boldsymbol\epsilon_0} \frac{  \pmb{\Delta}_2}{\pmb{\Delta}_1} + C_{\mathbf B}\pmb{\Delta}_2^2 + 3\pmb{\Delta}_2 \beta + c_{\boldsymbol\epsilon_0}\pmb{\Delta}_2   \Big] +  2\beta.
 \end{align*}
Because $L^{+}_{t}(\delta)-L^{+}_{t_2}(\delta)= v_F-v_{t_2}(\mathsf x)+z$ with $|z|\le \beta$, this implies that 
 \begin{align*}
  | v_t(\mathsf x)-v_F|\le \mathbf r^{(3)}_\star + \beta,
  \end{align*}
  and for all $s\in  [t_2,t]$:
  $$| v_s(\mathsf x)|\le   \mathbf r^{(3)}_\star + |v_{t_2}(\mathsf x)|+| L^{+}_s(\delta)-L^{+}_{t_2}(\delta)| \le   \mathbf r^{(3)}_\star + 2 \mathbf r^{(2)}_\star+ \frac{4c_{\boldsymbol\epsilon_0}}{\pmb{\Delta}_1} +c_{\boldsymbol\epsilon_0}.$$
The latter inequality yields, for every   $s\in  [t_2,t]$  
 \begin{align*}
 |x_s(\mathsf x)- x_F|&\le   |x_{t_2}(\mathsf x)-x_F|+ \pmb{\Delta}_2\mathbf r^{(3)}_\star + 2  \pmb{\Delta}_2\mathbf r^{(2)}_\star+ \frac{4c_{\boldsymbol\epsilon_0} \pmb{\Delta}_2}{\pmb{\Delta}_1} + \pmb{\Delta}_2c_{\boldsymbol\epsilon_0} \\
 &\le \mathbf R^{(2)}_\star:= \mathbf R^{(1)}_\star+ \pmb{\Delta}_2\mathbf r^{(3)}_\star + 2  \pmb{\Delta}_2\mathbf r^{(2)}_\star+ \frac{4c_{\boldsymbol\epsilon_0} \pmb{\Delta}_2}{\pmb{\Delta}_1} + \pmb{\Delta}_2c_{\boldsymbol\epsilon_0}.
   \end{align*}
\underline{Conclusion}.  In the limit  $\beta+ \pmb{\Delta}_1 + \frac{  \pmb{\Delta}_0}{\pmb{\Delta}_1}+ \frac{  \pmb{\Delta}_2}{\pmb{\Delta}_1} \to 0^+$ (and thus $\pmb{\Delta}_0+\pmb{\Delta}_2\to 0$ as well)  we have that 
 $$  \mathbf R^{(0)}_\star +  \mathbf R^{(1)}_\star+   \mathbf R^{(2)}_\star +  (\mathbf r^{(3)}_\star+ \beta) \to 0.$$
 Therefore, recalling \eqref{eq.seg}, we obtain that  there exist $\beta_{ \boldsymbol\epsilon}^{\mathbf A}>0$ and $t^{\mathbf A}_{\boldsymbol\epsilon}\in (0,  1]$     such that for every
 \begin{enumerate}
 \item[-] $\mathsf x=(x,v)\in \mathscr B_{\mathbb R^{2d}}(\mathsf x_F,\boldsymbol\epsilon) $, 
 \item[-]   $t\in (0,t^{\mathbf A}_{\boldsymbol\epsilon}]$,   $\pmb{\Delta}_0\le t^{\mathbf A}_{\boldsymbol\epsilon} \pmb{\Delta}_1$, $\pmb{\Delta}_2\le t^{\mathbf A}_{\boldsymbol\epsilon}  \pmb{\Delta}_1$,    
  \item[-] $\delta\in (0,1]$,
 \end{enumerate} 
 we have that:
\begin{align*}
\mathscr E_{\textbf{A}} (\mathsf x,\delta,t_1,t_2,t,\beta_{\boldsymbol\epsilon}^{\textbf{A}}) \subset  \Big \{ |X_t(\mathsf x)- \mathsf x_F|\le \boldsymbol\epsilon \text{ and } x_s(\mathsf x)\in \mathscr O \text{ for all $s\in [0,t]$} \Big \}. 
\end{align*}
Choosing finally 
$$\delta= \delta_{\beta_{\boldsymbol\epsilon}^{\textbf{A}}}(1)\wedge \frac{ \beta_{\boldsymbol\epsilon}^{\textbf{A}}}2\wedge  1,$$
yields~\eqref{eq.Proba}.  This shows~\eqref{eq.P>} for every $t\in (0,t^{\mathbf A}_{\boldsymbol\epsilon}]$ and every   $\mathsf x=(x,v)\in \mathscr B_{\mathbb R^{2d}}(\mathsf x_F,\boldsymbol\epsilon) $. 
  \medskip 

\noindent
 \textbf{Step 4 (Case B)}. Proof of \eqref{eq.P>}    when $\mathsf x=\mathsf x_0$. 
 \medskip
 
 \noindent
 The case when   $|x_0-x_F|\le \boldsymbol\epsilon$ is treated  in the  previous  case.
 We  thus consider the case when 
 $$
 |x_0-x_F|>  \boldsymbol\epsilon.
 $$

\noindent
 \textbf{Step 4.1}. Preliminary analysis. 
 \medskip
 
 \noindent
 Using the curve $\gamma_0$ above, one deduces that there exist $N\ge 0$ and distinct points   $ x_1, \ldots,  x_N$   such that setting $ x_{N+1}= x_F$ and $x_0=x$, it holds for $k\in \{0, \ldots,N\}$ (if $N=0$, $\{x_1, \ldots,  x_N\}=\varnothing$ by convention):
  $$[x_k,x_{k+1}]\subset   \text{Ran}_{\boldsymbol\epsilon/2}(\gamma_0).$$ 
 Note that (see~\eqref{eq.h1})
 $$
 \Big \{y\in  \mathbb R^d, \mathbf d_{\mathbb R^d}\big (y, \bigcup_{k=0}^N[x_k,x_{k+1}] \big )\le   \boldsymbol\epsilon \Big  \} \subset \text{Ran}_{3\boldsymbol\epsilon}(\gamma_0) \ (\subset  \mathscr O).
 $$  
  We can also assume that $x_{k+1}-x_k$ is not collinear to $x_{k+2}-x_{k+1}$ (else we remove the middle point $x_{k+1}$). Note that  $x_0, x_1, \ldots,  x_N,x_F$ is a \textit{broken line} joining $x_0$ to $x_F$ in $\mathscr O$.  
 We now introduce a finite sequence of increasing times 
 $$t_0=0<t_1<t_2<t_3<t_4<t_5<t_6\ldots <t_{4+ 2N-1} =t.$$
 Set   for $i\ge 0$, 
 $$  \pmb{\Delta}_i= t_{i+1}-t_i.$$
 We now define several events on which we will work. On $(t_0,t_1)$ (recall \eqref{eq.lmp}) we consider 
  \begin{align*}
  &\mathscr E^{+}_{\text{jump}}(\delta, t_0,t_1,\beta) \\
  &:=\Big \{\big (L^{+}_s(\delta),s\in (t_0,t_1)\big ) \text{ has exactly one jump}\\
  &\quad \quad  \quad \text{and this jump lies in } \mathscr B_{\mathbb R^d}(v_{x_0\to x_1}(\pmb{\Delta}_1) -v_0,\beta) \Big \}. 
 \end{align*}
On  $(t_2,t_3)$, $(t_4,t_5)$,  $\ldots$,  we consider  recursively the events 
 \begin{align*}
  &\mathscr E^{+}_{\text{jump}}(\delta, t_{2p},t_{2p+1},\beta) \\
  &:=\Big \{\big (L^{+}_s(\delta)-L_{t_{2p}}^{+}(\delta),s\in (t_{2p},t_{2p+1})\big ) \text{ has exactly one jump}\\
  &\quad \quad  \quad \text{and this jump lies in } \mathscr B_{\mathbb R^d}(v_{x_p\to x_{p+1}}(\pmb{\Delta}_{2p+1})- v_{t_{2p}}(\mathsf x_0),\beta) \Big \}, 
 \end{align*}
  while    on the last time interval  $(t_{4+ 2N-2}, t)$, we consider 
  \begin{align*}
  &\mathscr E^{+}_{\text{jump}}(\delta, t_{4+ 2N-2}, t_{4+ 2N-1},\beta) \\
  &:=\Big \{\big (L^{+}_s(\delta)-L_{t_{4+ 2N-2}}^{+}(\delta),s\in (t_{4+ 2N-2}, t_{4+ 2N-1})\big ) \text{ has exactly one jump}\\
  &\quad \quad  \quad \text{and this jump lies in } \mathscr B_{\mathbb R^d}( v_F- v_{t_{4+ 2N-2}}(\mathsf x_0),\beta) \Big \}. 
 \end{align*}
  On the time intervals $[t_1,t_2]$, $[t_3,t_4]$,  $[t_5,t_6]$ $\ldots$, we consider  the events:
 \begin{align*}
 &\mathscr E^{+}_{\text{no jump}}(\delta,t_{2p+1},t_{2p+2})\\
 &:=\Big \{\big (L^{+}_s(\delta)-L_{t_{2p+1}}^{+}(\delta),s\in [t_{2p+1},t_{2p+2}]\big )  \text{ has no jumps}\Big \},
 \end{align*}
 We finally define   (recall~\eqref{eq.e-}) 
\begin{align*}
&\mathscr E_{\textbf{B}} (\delta,t_1,\ldots, t_{4+ 2N-2},t ,\beta) \\
&:= \mathscr E^{-}(\delta,t,\beta)\pmb{\cap} \mathscr E^{+}_{\text{jump}} (\delta, t_0,t_1,\beta)\pmb{\cap}  \mathscr E^{+}_{\text{no jump}}(\delta, t_1,t_2)  \ldots \\
&\quad \pmb{\cap} \mathscr E^{+}_{\text{no jump}} (\delta,t_{4+ 2N-3},t_{4+ 2N-2}) \pmb{\cap}  \mathscr E^{+}_{\text{jump}}(\delta,t_{4+ 2N-2},t,\beta).
\end{align*} 
We now prove   that if $0<\delta\le  \delta_\beta(1)\wedge \frac \beta2\wedge 1$,
  \begin{align}\label{eq.IN0}
  \mathbb P\big[\mathscr E_{\textbf{B}} (\delta,t_1,\ldots, t_{4+ 2N-2},t ,\beta)\big]>0.
    \end{align}
Recall that   $(L^{-}_s(\delta),s\ge 0)$ and $(L^{+}_s(\delta),s\ge 0)$ are independent, and that  the increments of the Lévy process $(L^{+}_s(\delta),s\ge 0)$ are also independent. Thus, it holds:
\begin{align*}
&\mathbb P\Big[\mathscr E^{-}(\delta,t,\beta)\pmb{\cap} \mathscr E^{+}_{\text{jump}} (\delta, t_0,t_1,\beta)\pmb{\cap}  \mathscr E^{+}_{\text{no jump}}(\delta, t_1,t_2)  \ldots \\
&\quad \pmb{\cap} \mathscr E^{+}_{\text{no jump}} (\delta,t_{4+ 2N-3},t_{4+ 2N-2}) \pmb{\cap}  \mathscr E^{+}_{\text{jump}}(\delta,t_{4+ 2N-2},t,\beta)\Big]\\
&=\underbrace{\mathbb P[\mathscr E^{-}(\delta,t,\beta)]}_{>0}\ \mathbb E\Big[\mathbf 1_{ \mathscr E^{+}_{\text{jump}} (\delta, t_0,t_1,\beta)\pmb{\cap}  \,  \ldots \, \pmb{\cap} \mathscr E^{+}_{\text{no jump}} (\delta,t_{4+ 2N-3},t_{4+ 2N-2}) }\\
&\quad \times \underbrace{\mathbb P[ \mathscr E^{+}_{\text{jump}}(\delta,t_{4+ 2N-2},t,\beta)|\mathcal F_{t_{4+ 2N-2}}]}_{\Phi( v_{t_{4+ 2N-2}}(\mathsf x))}\Big],
  \end{align*}
  where 
  $$\Phi (y) =   \pmb{\Delta}_{4+2N-2}  e^{-\lambda_\delta  \pmb{\Delta}_{4+2N-2}} \ \nu  \big( \mathscr B_{\mathbb R^d}(v_F - y, \beta) \cap \{ |z| > \delta \}\big ). $$
  Since $0<\delta\le   \frac \beta2\wedge 1$, it holds $\Phi ( y)>0$ for all $y\in \mathbb R^d$. 
  It is thus enough to prove that 
  \begin{align}\label{eq.IN}
 \mathbb P\Big[  \underbrace{ \mathscr E^{+}_{\text{jump}} (\delta, t_0,t_1,\beta)\pmb{\cap}  \,  \ldots \, \pmb{\cap} \mathscr E^{+}_{\text{no jump}} (\delta,t_{4+ 2N-3},t_{4+ 2N-2})}_{\mathscr A_N}  \Big]>0.
  \end{align}
  One proceeds by  induction on $N\ge 0$.   If $N=0$, we have 
   \begin{align*}
 \mathbb P[\mathscr A_0]=\mathbb P\Big[ \mathscr E^{+}_{\text{jump}} (\delta, t_0,t_1,\beta)\pmb{\cap}  \mathscr E^{+}_{\text{no jump}}(\delta, t_1,t_2) \Big]&= \mathbb P\Big[ \mathscr E^{+}_{\text{jump}} (\delta, t_0,t_1,\beta) \Big]\\
 &\quad \times \mathbb P\Big[  \mathscr E^{+}_{\text{no jump}}(\delta, t_1,t_2) \Big] >0.
   \end{align*}
   This proves \eqref{eq.IN} when $N=0$.
 Assume that \eqref{eq.IN} holds for some $N\ge 0$. 
  Write 
   \begin{align*}
   \mathscr A_{N+1}= \mathscr A_{N}\pmb{\cap} \mathscr E^{+}_{\text{jump}} (\delta,t_{4+ 2N-2},t_{4+ 2N-1})\pmb{\cap} \mathscr E^{+}_{\text{no jump}} (\delta,t_{4+ 2N-1},t_{4+ 2N}).
    \end{align*} 
    We have, conditioning   w.r.t. $\mathcal F_{t_{4+ 2N-2}}$, 
     \begin{align*}
   \mathbb P[\mathscr A_{N+1}]&= \mathbb P\Big[\mathscr A_{N}\pmb{\cap} \mathscr E^{+}_{\text{jump}} (\delta,t_{4+ 2N-2},t_{4+ 2N-1})\Big]  \, \mathbb P\big [\mathscr E^{+}_{\text{no jump}} (\delta,t_{4+ 2N-1},t_{4+ 2N})\big ]\\
   &= \mathbb E\Big[\mathbf 1_{\mathscr A_{N}}    \Upsilon( v_{t_{4+ 2N-2}}(\mathsf x_0)) \Big]   \times \underbrace{\mathbb P\big [\mathscr E^{+}_{\text{no jump}} (\delta,t_{4+ 2N-1},t_{4+ 2N})\big ]}_{>0},
    \end{align*} 
    where (Recall that $x_{N+2}=x_F$):
    $$\Upsilon(y) =   \pmb{\Delta}_{4+ 2N-2} e^{-\lambda_\delta  \pmb{\Delta}_{4+ 2N-2}}  \,      \nu \big( \mathscr B_{\mathbb R^d}(v_{x_{N+1}\to x_{N+2}}(\pmb{\Delta}_{4+ 2N-1}) - y, \beta) \cap \{ |z| > \delta \} \big) .$$
    Since $0<\delta\le   \frac \beta2\wedge 1$, it holds $\Upsilon ( y)>0$ for all $y\in \mathbb R^d$. Because $\mathbb P[\mathscr A_N]>0$ (by the induction hypothesis), we obtain that
    $\mathbb E [\mathbf 1_{\mathscr A_{N}}    \Upsilon( v_{t_{4+ 2N-2}}(\mathsf x_0))  ] >0$ and then that $ \mathbb P[\mathscr A_{N+1}]>0$. This completes the proof of \eqref{eq.IN} by induction. The proof of \eqref{eq.IN0} is complete. 
%
   \medskip 

\noindent
 \textbf{Step 4.2}. We now end the proof of~\eqref{eq.P>} when $\mathsf x=\mathsf x_0$.  
 \medskip
 
 \noindent
  The  case $N=0$, corresponding to the scenario where the segment $[x_0, x_F]$ is entirely contained within the domain $\mathscr O$, follows directly from the arguments established in \textbf{Case A} above. In the following, we provide a detailed derivation for the case $N=1$ (illustrated in Figures~\ref{fig.R} and \ref{fig.R2}).  The general case for $N \ge 2$ follows by a straightforward inductive extension of this logic. 
 \medskip
 
 \noindent
 \textbf{Step 4.2.a}. Reach $x_F$ starting from $x$.   The analysis of the trajectories of \eqref{eq.Lan} on $[t_0,t_1]$ and $[t_1,t_2]$  are very close to the one made in \textbf{Case A} above. 
  \medskip

  \noindent
 \underline{On $[t_0,t_1]$}.  The process $(L^{+}_s(\delta),s\in (0,t_1))$ has only one jump, and this jump is in the $\beta$ neighborhood of $v_{x \to x_1}(\pmb{\Delta}_1) - v_0$  (see Figure~\ref{fig.R2}). 
Consequently,  for every  $s\in [0,t_1]$ (recall \eqref{eq.COF} for the definition of $c_{\boldsymbol\epsilon_0}>0$):
 $$|L^{+}_s(\delta)|\le |v_{x_0\to x_1}(\pmb{\Delta}_1)-v_0|+\beta\le \frac{2c_{\boldsymbol\epsilon_0}}{\pmb{\Delta}_1}+ c_{\boldsymbol\epsilon_0}+ \beta.$$ 
 Recall the definition of $ \mathbf r^{(0)}_\star $ in \eqref{eq.r0star}. 
By \eqref{eq.E}, it then holds for $s\in [0,t_1]$:
  $s\in [0,t_1]$:  
  \begin{align*}
    \big |v_s(\mathsf x_0)- v_0  -  L^{+}_s(\delta) \big  |   \le \mathbf r^{(0)}_\star  \text{ and } 
 \big |v_s(\mathsf x_0)|\le   c_{\boldsymbol\epsilon_0}+\frac{2c_{\boldsymbol\epsilon_0}}{\pmb{\Delta}_1}+ c_{\boldsymbol\epsilon_0}+ \beta+ \mathbf r^{(0)}_\star.
 \end{align*}
 As a result,  for   $s\in [0,t_1]$:
 $$|x_s(\mathsf x_0)-x_0|\le\mathbf R^{(0)}_\star:=  2c_{\boldsymbol\epsilon_0}\frac{  \pmb{\Delta}_0}{\pmb{\Delta}_1}+ \pmb{\Delta}_0[2c_{\boldsymbol\epsilon_0}+ \beta+   \mathbf r^{(0)}_\star].$$
Finally, since $L^{+}_{t_1}(\delta)= v_{x\to x_1}(\pmb{\Delta}_1)-v_0 + z$ with $|z|\le \beta$, we get 
  \begin{align}
    \label{eq.vt1Bis}
 & \big |v_{t_1}(\mathsf x_0)-  v_{x_0\to x_1}(\pmb{\Delta}_1)|   \le \beta+ \mathbf r^{(0)}_\star.  
 \end{align}

 \noindent
\underline{On $[t_1,t_2]$}. 
The process $(L^{+}_s(\delta)-L^{+}_{t_1}(\delta),s\in [t_1,t_2])$ has no jumps. As a result, using  \eqref{eq.E} and \eqref{eq.vt1Bis},  we deduce that  for $s\in [t_1,t_2]$:
 \begin{align*}
 & \big |v_s(\mathsf x_0)- v_{t_1}(\mathsf x_0) \big  |  \le \mathbf r^{(1)}_\star, 
 \end{align*}
 where we recall  \eqref{eq.rstar1} for the definition of $\mathbf r^{(1)}_\star$.
 We then have using also \eqref{eq.vt1Bis} and the triangle inequality:
$$|v_s(\mathsf x_0) -  v_{x\to x_1}(\pmb{\Delta}_1)|\le \mathbf r^{(2)}_\star , \ s\in [t_1,t_2],$$
where we recall that $\mathbf r^{(2)}_\star= \mathbf r^{(1)}_\star+  \beta+ \mathbf r^{(0)}_\star$. 
In addition, because  $|x_{t_1}(\mathsf x_0)-x_0|\le \mathbf R^{(0)}_\star $, we deduce that:  
\begin{align*}\Big |x_s(\mathsf x_0)- x_0 -  \frac{s-t_1}{\pmb{\Delta}_1} (x_1-x_0)\Big |  &\le  \mathbf R^{(1)}_\star,
 \end{align*}
 where we recall that $\mathbf R^{(1)}_\star:= \mathbf R^{(0)}_\star + \pmb{\Delta}_1\,  \mathbf r^{(2)}_\star$.   In particular, we have at time $s=t_2$:
\begin{equation}\label{eq.vtt2bis}
|x_{t_2}(\mathsf x_0)-x_1|\le  \mathbf R^{(1)}_\star \text{ and }  |v_{t_2}(\mathsf x_0) | \le  \mathbf r^{(2)}_\star+ \frac{2c_{\boldsymbol\epsilon_0}}{\pmb{\Delta}_1}.
 \end{equation}

  \noindent
 \underline{On $[t_2,t_3]$}. The process $(L^{+}_s(\delta)-L^{+}_{t_2}(\delta),s\in (t_2,t_3))$ has only one jump, and this jump is in the $\beta$ neighborhood of $v_{x_1 \to x_2}(\pmb{\Delta}_3) - v_{t_2}(\mathsf x_0)$ (see Figure~\ref{fig.R2}). 
Therefore, for $s\in  [t_2,t_3]$: 
$$|L^{+}_s(\delta)-L^{+}_{t_2}(\delta)|\le   \frac{2c_{\boldsymbol\epsilon_0}}{\pmb{\Delta}_3}+ |v_{t_2}(\mathsf x_0) |+\beta.$$  
Then, using \eqref{eq.E} and \eqref{eq.vtt2bis}, it holds for every $s\in [t_2,t_3]$:
 \begin{align*}
&\big| v_s(\mathsf x_0)- v_{t_2}(\mathsf x_0) -( L^{+}_s(\delta)-L^{+}_{t_2}(\delta))  \big|\\
& \le C_{\mathbf B}\pmb{\Delta}_2+ C_{\mathbf B}\pmb{\Delta}_2 \, e^{C_{\mathbf B}\pmb{\Delta}_2} \Big[ 2|v_{t_2}(\mathsf x_0) | + C_{\mathbf B}\pmb{\Delta}_2 + 3\beta  + \frac{2c_{\boldsymbol\epsilon_0}}{\pmb{\Delta}_3}  \Big] +  2\beta\\
  &\le  \hat{\mathbf r}^{(3)}_\star:= C_{\mathbf B}\pmb{\Delta}_2+ C_{\mathbf B}  \, e^{C_{\mathbf B}\pmb{\Delta}_2} \Big[ 2 \pmb{\Delta}_2 \mathbf r^{(2)}_\star+ 4c_{\boldsymbol\epsilon_0} \frac{  \pmb{\Delta}_2}{\pmb{\Delta}_1} + 2c_{\boldsymbol\epsilon_0} \frac{  \pmb{\Delta}_2}{\pmb{\Delta}_3} + C_{\mathbf B}\pmb{\Delta}_2^2 + 3\pmb{\Delta}_2 \beta  \Big] +  2\beta.
 \end{align*}
Because $L^{+}_{t_3}(\delta)-L^{+}_{t_2}(\delta)=v_{x_1 \to x_2}(\pmb{\Delta}_3) - v_{t_2}(\mathsf x_0)+z$ with $|z|\le \beta$, this implies that at time $s=t_3$:
 \begin{align}
 \label{eq.vvP}
  | v_{t_3}(\mathsf x_0)-v_{x_1 \to x_2}(\pmb{\Delta}_3)|\le \hat{\mathbf r}^{(3)}_\star + \beta
  \end{align}
  and for all $s\in  [t_2,t_3]$
 \begin{align}
 \nonumber
 | v_s(\mathsf x_0)|&\le   \hat{\mathbf r}^{(3)}_\star + |v_{t_2}(\mathsf x_0)|+| L^{+}_s(\delta)-L^{+}_{t_2}(\delta)| \\
 \label{eq.vt3}
 &\le   \hat{\mathbf r}^{(3)}_\star + 2  \mathbf r^{(2)}_\star+   \frac{ 4c_{\boldsymbol\epsilon_0} }{\pmb{\Delta}_1} +   \frac{  2c_{\boldsymbol\epsilon_0} }{\pmb{\Delta}_3}.
   \end{align}
The latter inequality yields, for every   $s\in  [t_2,t_3]$ (recall \eqref{eq.vtt2bis}):  
 \begin{align*}
 |x_s(\mathsf x_0)- x_1|&\le   |x_{t_2}(\mathsf x_0)-x_1|+ \pmb{\Delta}_2\hat{\mathbf r} ^{(3)}_\star + 2  \pmb{\Delta}_2\mathbf r^{(2)}_\star+ \frac{4c_{\boldsymbol\epsilon_0} \pmb{\Delta}_2}{\pmb{\Delta}_1} + \pmb{\Delta}_2c_{\boldsymbol\epsilon_0} \\
 &\le \hat{\mathbf R}^{(2)}_\star:= \mathbf R^{(1)}_\star+ \pmb{\Delta}_2\hat{\mathbf r}^{(3)}_\star + 2  \pmb{\Delta}_2\mathbf r^{(2)}_\star+ 4c_{\boldsymbol\epsilon_0} \frac{  \pmb{\Delta}_2}{\pmb{\Delta}_1} + 2c_{\boldsymbol\epsilon_0} \frac{  \pmb{\Delta}_2}{\pmb{\Delta}_3}.
   \end{align*}

 \noindent
  \underline{On $[t_3,t_4]$}. Because $(L^{+}_s(\delta)-L^{+}_{t_3}(\delta),s\in [t_3,t_4])$ has  no jumps, we have using \eqref{eq.vt3} and \eqref{eq.E}, for every $s\in [t_3,t_4]$: 
   \begin{align}
  \nonumber
 & \big |v_s(\mathsf x_0)- v_{t_3}(\mathsf x_0) \big  | \\
  \nonumber
 &\le C_{\mathbf B}\pmb{\Delta}_3+ C_{\mathbf B}  \, e^{C_{\mathbf B}\pmb{\Delta}_3} \pmb{\Delta}_3 \Big[ |v_{t_3}(\mathsf x_0)|  +  C_{\mathbf B}\pmb{\Delta}_3 + 2\beta     \Big] +  2\beta \\
   \nonumber
 & \le \mathbf r^{(4)}_\star:= C_{\mathbf B}\pmb{\Delta}_3+  C_{\mathbf B}  \, e^{C_{\mathbf B}\pmb{\Delta}_3}  \Big[    2c_{\boldsymbol\epsilon_0} +  \pmb{\Delta}_3\hat{\mathbf r}^{(3)}_\star + 2 \pmb{\Delta}_3 \mathbf r^{(2)}_\star+  4c_{\boldsymbol\epsilon_0} \frac{   \pmb{\Delta}_3}{\pmb{\Delta}_1} \\
 \label{eq.rstar1bis}
 &\quad \quad \quad \quad + \pmb{\Delta}_3(C_{\mathbf B}\pmb{\Delta}_3 + 2\beta)\Big] +  2\beta.
 \end{align}
By the triangle inequality and using \eqref{eq.vvP}, it then holds:
$$
  |v_s(\mathsf x_0)-v_{x_1 \to x_2}(\pmb{\Delta}_3)|\le \mathbf r^{(4)}_\star+ \hat{\mathbf r}^{(3)}_\star + \beta, \ s\in [t_3,t_4]. 
  $$
  Recall that $v_{x_1 \to x_2}(\pmb{\Delta}_3)= (x_2-x_1)/\pmb{\Delta}_3$. Hence, because $|x_s(\mathsf x_0)- x_1| \le  \hat{\mathbf R}^{(2)}_\star$, we get by the triangle inequality, for every $s\in [t_3,t_4]$:  
\begin{align*}
\Big |x_s(\mathsf x_0)- x_1 -  \frac{s-t_3}{\pmb{\Delta}_3} (x_2-x_1)\Big |  &\le 
\mathbf R^{(3)}_\star:= \hat{\mathbf R}^{(2)}_\star+  \pmb{\Delta}_3 (\mathbf r^{(4)}_\star+ \hat{\mathbf r}^{(3)}_\star + \beta).  
 \end{align*} 
   In particular,  at time $s=t_4$, the two previous upper bounds  imply that:
$$
|x_{t_4}(\mathsf x_0)-x_2|\le  \mathbf R^{(3)}_\star  \text{ and }  |v_{t_4}(\mathsf x_0) | \le  \frac{2c_{\boldsymbol\epsilon_0}}{\pmb{\Delta}_3 }+ \mathbf r^{(4)}_\star+ \hat{\mathbf r}^{(3)}_\star + \beta.
$$
 \noindent
 \textbf{Step 4.2.b}. Reach $v_F$ during the time interval $[t_4,t_5]$ (Recall that $t_5=t$). 
   The purpose of this final step is to  make   $v_t(\mathsf x_0)$  fall into the   $\boldsymbol\epsilon/2$ neighborhood of  $v_F$  and  $x_t(\mathsf x_0)$ fall into the    $\boldsymbol\epsilon/2$ of  $x_F$, while ensuring that     $(x_s,s\in [t_4,t_5])$ stays close to $x_2=x_F$ (so that it remains inside $\mathscr O$).  
   \medskip
  
 \noindent
 \underline{On $[t_4,t]$}.
 For  $s\in [t_4,t]$,  
 The process $(L^{+}_s(\delta)-L^{+}_{t_4}(\delta),s\in (t_4,t))$ has only one jump, and this jump is in the $\beta$ neighborhood of $v_F-v_{t_4}(\mathsf x_0)$.   
Thus, for $s\in  [t_4,t]$: 
$$|L^{+}_s(\delta)-L^{+}_{t_4}(\delta)|\le |v_F|+ |v_{t_4}(\mathsf x_0) |+\beta\le c_{\boldsymbol\epsilon_0}+ \frac{2c_{\boldsymbol\epsilon_0}}{\pmb{\Delta}_3 }+ \mathbf r^{(4)}_\star+ \hat{\mathbf r}^{(3)}_\star + 2\beta.$$ 
Then, using \eqref{eq.E}, it holds for every $s\in [t_4,t]$:
 \begin{align*}
&\big| v_s(\mathsf x_0)- v_{t_4}(\mathsf x_0) -( L^{+}_s(\delta)-L^{+}_{t_4}(\delta))  \big|\\
& \le  \mathbf r^{(5)}_\star:=  C_{\mathbf B}\pmb{\Delta}_4+ C_{\mathbf B}\, e^{C_{\mathbf B}\pmb{\Delta}_4 }  4c_{\boldsymbol\epsilon_0}  \frac{ \pmb{\Delta}_4 }{\pmb{\Delta}_3 } \\
&\quad +  \pmb{\Delta}_4( 2\mathbf r^{(4)}_\star+ 2\hat{\mathbf r}^{(3)}_\star +   C_{\mathbf B}\pmb{\Delta}_4 + 5\beta  +c_{\boldsymbol\epsilon_0} )  \Big] +  2\beta. 
 \end{align*}
Because $L^{+}_t(\delta)-L^{+}_{t_4}(\delta)= v_F-v_{t_4}(\mathsf x_0)+z$ with $|z|\le \beta$, we deduce that:
 \begin{align*}
  | v_t(\mathsf x_0)-v_F|\le \mathbf r^{(5)}_\star + \beta,
  \end{align*}
  and for all $s\in  [t_4,t]$:
  $$| v_s(\mathsf x_0)|\le   \mathbf r^{(5)}_\star + |v_{t_4}(\mathsf x_0)|+| L^{+}_s(\delta)-L^{+}_{t_4}(\delta)| \le   \mathbf r^{(5)}_\star + \frac{4c_{\boldsymbol\epsilon_0}}{\pmb{\Delta}_3 }+ 2 \mathbf r^{(4)}_\star+ 2\hat{\mathbf r}^{(3)}_\star +3 \beta+c_{\boldsymbol\epsilon_0}.$$
The latter inequality yields, for every   $s\in  [t_4,t]$  
 \begin{align*}
 |x_s(\mathsf x_0)- x_F|&\le   |x_{t_4}(\mathsf x_0)-x_F| + \pmb{\Delta}_4\mathbf r^{(5)}_\star + 4c_{\boldsymbol\epsilon_0}\frac{ \pmb{\Delta}_4}{\pmb{\Delta}_3 }\\
 &\quad + 2\pmb{\Delta}_4 \mathbf r^{(4)}_\star+ 2\pmb{\Delta}_4\hat{\mathbf r}^{(3)}_\star +\pmb{\Delta}_4(3 \beta+c_{\boldsymbol\epsilon_0})\\
 &\le \mathbf R^{(4)}_\star:=\mathbf R^{(3)}_\star + 4c_{\boldsymbol\epsilon_0}\frac{ \pmb{\Delta}_4}{\pmb{\Delta}_3 }  + 2\pmb{\Delta}_4 \mathbf r^{(4)}_\star+ 2\pmb{\Delta}_4\hat{\mathbf r}^{(3)}_\star +\pmb{\Delta}_4(3 \beta+c_{\boldsymbol\epsilon_0}).
   \end{align*}
\underline{Conclusion}. In the limit  $\beta+ \pmb{\Delta}_1 +  \pmb{\Delta}_3+  \frac{  \pmb{\Delta}_0}{\pmb{\Delta}_1}+ \frac{  \pmb{\Delta}_2}{\pmb{\Delta}_1} + \frac{  \pmb{\Delta}_2}{\pmb{\Delta}_3}+  \frac{  \pmb{\Delta}_4}{\pmb{\Delta}_3}       \to 0^+$ (and thus $\pmb{\Delta}_0+\pmb{\Delta}_2+ \pmb{\Delta}_4\to 0$ as well)  we have that 
 $$\mathbf R^{(0)}_\star + \mathbf R^{(1)}_\star+\hat{ \mathbf R}^{(2)}_\star+ \mathbf R^{(3)}_\star+ \mathbf R^{(4)}_\star+ \mathbf r^{(5)}_\star \to 0.$$
We then  deduce that  there exist $\beta^{\mathbf B} _{\boldsymbol\epsilon}>0$ and  $t^{\mathbf B}_{\boldsymbol\epsilon}\in (0,1]$   such that for every   
 \begin{enumerate}
 \item [-] 
 $t\in (0,t^{\mathbf B}_{\boldsymbol\epsilon}]$,   $\pmb{\Delta}_0\le t^{\mathbf B}_{\boldsymbol\epsilon} \pmb{\Delta}_1$, $\pmb{\Delta}_2\le t^{\mathbf B}_{\boldsymbol\epsilon} \min (\pmb{\Delta}_1,\pmb{\Delta}_3)$, and $\pmb{\Delta}_4\le t^{\mathbf B}_{\boldsymbol\epsilon} \pmb{\Delta}_3$,  
 \item [-]  $\delta\in (0,1]$,
 \end{enumerate}
 we have:
 \begin{align*}
 &\mathscr E_{\textbf{B}} (\delta,t_1,\ldots, t_4,t ,\beta_{\boldsymbol\epsilon}^{\mathbf B}) \\
 &\subset  \Big \{ |X_t(\mathsf x_0)- \mathsf x_F|\le \boldsymbol\epsilon \text{ and } x_s(\mathsf x_0)\in \mathscr O \text{ for all $s\in [0,t]$} \Big \}. 
 \end{align*}
  Choosing in addition $\delta= \delta_{\beta_{\boldsymbol\epsilon}^{\mathbf B}}(1)\wedge \beta_{\boldsymbol\epsilon}^{\mathbf B}/2\wedge  1$, we  get  that 
 $$\mathbb P[\mathscr E_{\textbf{B}} (\delta,t_1,\ldots, t_4,t ,\beta_{\boldsymbol\epsilon}^{\mathbf B})]>0.$$ This implies~\eqref{eq.P>} when $\mathsf x=\mathsf x_0$ and  $t\in (0,t^{\mathbf B}_{\boldsymbol\epsilon}]$, the desired result.   
 The proof of Proposition~\ref{pr.C5-infty} is complete  by considering 
 $$t_{\boldsymbol\epsilon}=\min( t_{\boldsymbol\epsilon}^{\mathbf A}, t_{\boldsymbol\epsilon}^{\mathbf B}).$$ 
 \end{proof}

 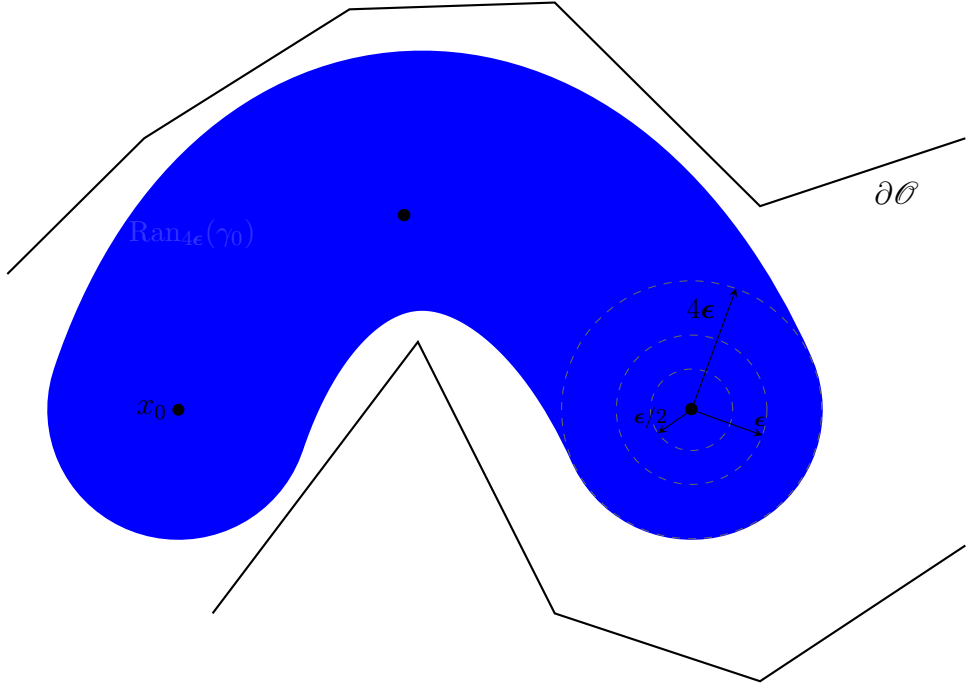
\begin{figure}[h]
\centering
\begin{tikzpicture}[scale=0.9]

 
\draw[thick]  (-2,0) -- (1, 4) -- (3,0) -- (6,-1) -- (9,1);
\draw[thick] (-5,5) -- (-3,7) -- (0,8.9) -- (3,9) -- (6,6) -- (9,7);

\node at (8,6.2) {  $\partial \mathscr{O}$};

\coordinate (x0) at (-2.5, 3);
\coordinate (xF) at (5, 3);
  \draw[scale=0.8]  (6.3, 3.3)  node[]{{\footnotesize $x_F$}};
 
\coordinate (x1) at (0.8, 5.9);
  \draw[scale=0.8]  (01, 6.5)  node[]{$x_1$};
\draw[thick, densely dotted] (x1) -- (xF);
\draw[thick, densely dotted] (x0) -- (x1);
\begin{scope}
    
    \draw[blue, line width=3.45cm, opacity=0.15, line cap=round] 
        (x0) .. controls (-1, 7.5) and (3, 7.5) .. (xF);
\end{scope}

\draw[blue, very thick] (x0) .. controls (-1, 7.5) and (3, 7.5) .. (xF);
\node[blue, font=\small] at (1, 7) {$\gamma_0$};
\node[text width=3cm, align=center, font=\small, blue!80] at (-2.3, 5.6) {$\text{Ran}_{4\boldsymbol\epsilon}(\gamma_0) $};

\def\rOne{0.6}   
\def\rTwo{1.1}   
\def\rThree{1.9} 

\draw[dashed, black!60] (xF) circle (\rOne);
\draw[dashed, black!60] (xF) circle (\rTwo);
\draw[dashed, black!60] (xF) circle (\rThree);

\draw[<-, >=stealth, thin] ($(xF)+(215:\rOne)$) -- (xF);
  \draw[scale=0.8]  (5.5,3.6)  node[]{{\tiny ${\boldsymbol\epsilon}/{2}$}};

\draw[<-, >=stealth, thin] ($(xF)+(-20:\rTwo)$) -- (xF);
  \draw[scale=0.8]  (7.5,3.54)  node[]{{\tiny ${\boldsymbol\epsilon}$}};

\draw[<-, >=stealth, thin] ($(xF)+(70:\rThree)$) -- (xF);
\node[scale=0.9] at ($(xF)+(85:1.5)$) {$4\boldsymbol\epsilon$};

\fill (x0) circle (2.5pt) node[left] {$x_0$};
  \draw   (xF)   node[]{$\bullet$};
    \draw   (0.8,5.85)   node[]{$\bullet$};
\end{tikzpicture}
\caption{Schematic representation (not to scale) of the  condition~\eqref{eq.h1} and $\text{Ran}_{4\boldsymbol\epsilon}(\gamma_0) $, when $N=1$.}
\label{fig.R}
\end{figure}

 
 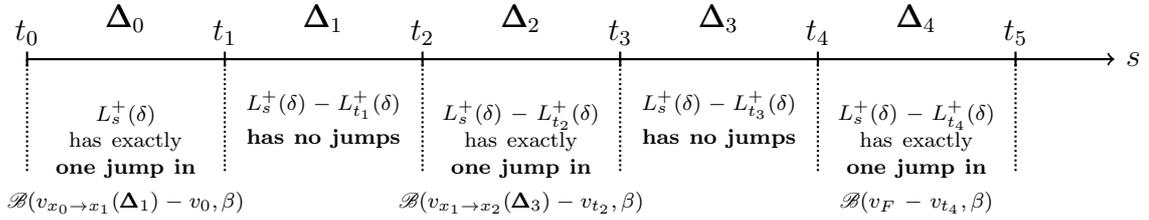
\begin{figure}
\begin{tikzpicture}[xscale=2.6]

\draw[thick, ->] (-0,0) -- (5.5,0) node[right] {$s$};

\foreach \i in {0,1,2,3,4,5} {
    \draw[thick] (\i,0.1) -- (\i,-0.1) node[above=4pt] {$t_{\i}$};
}


\node[above=6pt] at (0.5,0) {$\pmb{\Delta}_0$};
\node[above=6pt] at (1.5,0) {$\pmb{\Delta}_1$};
\node[above=6pt] at (2.5,0) {$\pmb{\Delta}_2$};
\node[above=6pt] at (3.5,0) {$\pmb{\Delta}_3$};
\node[above=6pt] at (4.5,0) {$\pmb{\Delta}_4$};

\node[text width=3.8cm, align=center, font=\small] at (0.5, -1.3) {
    {\tiny $L^{+}_s(\delta)$ \\
    has 
    exactly \\
    \textbf{one jump in } \\
      $\mathscr{B}(v_{x_0 \to x_1}(\pmb{\Delta}_1) - v_0, \beta)$}
};

\node[text width=3cm, align=center, font=\small] at (1.5, -0.8) {
     {\tiny $L^{+}_s(\delta)-L^{+}_{t_1}(\delta)$\\
      \textbf{has no jumps}}
};

\node[text width=3.8cm, align=center, font=\small] at (2.5, -1.3) {
    {\tiny $L^{+}_s(\delta)-L^{+}_{t_2}(\delta)$ \\
    has  
    exactly\\
     \textbf{one jump in} \\ 
       $\mathscr{B}(v_{x_1 \to x_2}(\pmb{\Delta}_3) - v_{t_2}, \beta)$}
};

\node[text width=3cm, align=center, font=\small] at (3.5, -0.8) {
    {\tiny $L^{+}_s(\delta)-L^{+}_{t_3}(\delta)$\\
      \textbf{has no jumps}}
};

\node[text width=3.8cm, align=center, font=\small] at (4.5, -1.3) {
    {\tiny $L^{+}_s(\delta)-L^{+}_{t_4}(\delta)$ \\
    has 
    exactly \\
    \textbf{one jump in } 
    \\
    $\mathscr{B}(v_F - v_{t_4}, \beta)$}
};

\foreach \x in {0,1,2,3,4,5} {
    \draw[thick, densely dotted] (\x, 0) -- (\x, -1.5);
}

\end{tikzpicture}
\caption{Schematic representation of the  cascade of events  for the increments of $(L^{+}_s(\delta),s\ge 0)$ when $N=1$ ($t_0=0$ and $t_5=t$) and when $X_0=\mathsf x_0$.}
\label{fig.R2}
 \end{figure}

\section{Weak well-posedness  of \eqref{eq.Lan}  driven by a RI$\alpha$S  process with $\alpha\in (1,2)$: the case of a bounded drift    $\mathbf B$}

\label{sec.Eu}
 
In this section, we prove the well-posedness   and the weak continuity property w.r.t. initial conditions of~\eqref{eq.Lan}   when 
 \begin{enumerate}
   \item[]{\rm\textbf{[{\footnotesize H$^{\alpha\in (1,2)}_{\text{bounded}}$}]}}  $\mathbf B$ is  measurable and bounded, the driving noise  in \eqref{eq.Lan} is a RI$\alpha$S  process,    and $\alpha\in (1,2)$.   
\end{enumerate} 
In this setting, we also  show the existence of a density in some $L^m$ spaces, $m>1$. 
Let us emphasize that several intermediate results we derive in this section  are also crucial  to prove  the strong Feller property of this process 
under {\rm\textbf{[{\footnotesize H$^{\alpha\in (1,2)}_{\text{bounded}}$}]}} in the next section, see the proof of Theorem~\ref{th.FF}. 
  The main results of this section are Theorems \ref{th.ex} and \ref{th.uni}.

\begin{theorem}\label{th.ex}
Assume {\rm\textbf{[{\footnotesize H$^{\alpha\in (1,2)}_{\text{bounded}}$}]}}.  Then, for every $\mu \in \mathcal P(\mathbb R^{2d})$, there is a unique weak solution   to~\eqref{eq.Lan}  with initial law $\mu$. In addition,  for $t\ge 0$,  the mapping $\mathsf x\mapsto \mathbb P_{\mathsf x}[X_t\in \cdot]$ (resp. $\mathsf x\mapsto \mathbb P_{\mathsf x}[X\in \cdot]$) is continuous for the weak convergence in $\mathcal P(\mathbb R^{2d})$ (resp. in $\mathcal P(D([0,+\infty),\mathbb R^{2d})))$.  
\end{theorem}
  \medskip

As already evoked in the introduction, the technical point in the proof of Theorem~\ref{th.ex} lies in the fact that  $\mathbf B$ is not continuous which is combined with the full degeneracy of~\eqref{eq.Lan} (in the sense that the noise only acts on the velocity variable). Moreover,  the Girsanov formula is not relevant  in this    setting to circumvent this lack of regularity (though it is a powerful tool often used  for SDEs with non-smooth coefficients which are  driven by a Brownian noise, see e.g.~\cite{jacod2013limit,Wu2001,chaudru}).   
\vspace{0.2cm}

The literature on the well-posedness of SDEs driven by Lévy noise is extensive. The following short review is by no means exhaustive and mainly focuses on results that are closest to our setting.  When the drift is bounded, and the equation is elliptic (noise acts in every direction) and driven by a Brownian motion, we refer to~\cite{zvonkin1974transformation,veretennikov1981strong,stroock1975diffusion} (see also~\cite[Section 5.5]{situ2005theory} and~\cite{Lepeltier,krylov2005strong,chen2016uniqueness}). We also refer to~\cite{stroock1975diffusion,kallenberg,kurtz2010equivalence,kulinich2014strong} where continuity of the drift is assumed for existence. Kinetic equations driven by a Brownian motion and  with low regularity coefficients have been studied in~\cite{de2023heat}.  Very close to our setting are   the recent works~\cite{marino} and~\cite[Theorem 4.1]{RocknerZhang} where    kinetic equations driven by $\alpha$-stable processes ($\alpha\in (1,2)$) with  Hölder continuous coefficients are considered. We emphasize that the work \cite{RocknerZhang} is concerned with  kinetic McKean-Vlasov equations driven by $\alpha$-stable processes (see also~\cite{de2023multidimensional}).   For another detailed investigation of  the strong well-posedness of \eqref{eq.Lan} under the assumption of Hölder continuous drifts, we refer the reader to \cite{hao2020schauder}.   Let us  also mention  the  elliptic example provided in  \cite[Theorem 3.2]{tanaka1974perturbation} where weak uniqueness fails when $0\le \alpha<1$ (see also~\cite{priola2012pathwise} or~\cite[Theorem 2.9]{marino} for an example of non uniqueness in a fully degenerate setting).


      This section is organized as follows.   In Section~\ref{eq.weakU}, we consider existence of density for weak solutions and derive the weak uniqueness of \eqref{eq.Lan} when  $\mathbf B$ is measurable and bounded, see Theorem~\ref{th.uni}. Weak existence is proved in  Section~\ref{sec.weakEx}. 
      Our main ingredients to prove Theorems \ref{th.ex} and  \ref{th.uni} are Krylov's type estimates. Such estimates, as originally derived in~\cite{krylov1969ito},  are powerful tools   to derive the  existence of  solutions to SDEs in low regularity settings (see e.g.~\cite{stroock1975diffusion,Lepeltier,mel1983stochastic,krylov2005strong,situ2005theory,zhangAIHP2020,chaudru,marino}).
   Throughout this section, we assume   {\rm\textbf{[{\footnotesize H$^{\alpha\in (1,2)}_{\text{bounded}}$}]}}. 



\subsection{Weak uniqueness to~\eqref{eq.Lan}} 
\label{eq.weakU}
In this section, we prove the following result. 

\begin{theorem}\label{th.uni}
Assume  {\rm\textbf{[{\footnotesize H$^{\alpha\in (1,2)}_{\text{bounded}}$}]}}.  The following holds:
\begin{enumerate}
\item[$\mathfrak 1$.] Consider a weak solution  $(X_t(\mathsf x), t\ge 0)$  to~\eqref{eq.Lan}  starting at $X_0=\mathsf x \in \mathbb R^{2d}$ (see Definition~\ref{de.w}). Then, for every $t>0$, $X_t(\mathsf x)$   has a density  $p_t(\mathsf x,\mathsf x')$ w.r.t. the Lebesgue measure $d\mathsf x'$. In addition,  for every $p>1$ such that
$$
 1+ (1+\alpha d)p^{-1}+ dp^{-1} <\alpha,$$ 
 the density $\mathsf x'\mapsto p_t(\mathsf x,\mathsf x')$ belongs to   $L^{p'}(\mathbb R^{2d})$ (here $p'= p/(p-1)$). 
 Furthermore, for any $t>0$,  
 $$\sup_{\mathsf{x} \in \mathbb{R}^{2d}} \Vert p_t(\mathsf{x}, \cdot)\Vert_{L^{p'}} \le  c_p\Big[(1\wedge t )^{-\frac{ 1+\mathscr A_0(p)}{\alpha}} + \int_0^t (1\wedge u )^{-\frac{1+\mathscr A_0(p)}{\alpha}} du\Big],$$
 for some $c_p>0$ independent of $t\ge 0$, and where,  
 $$\mathscr A_0(p)=(1+\alpha d)p^{-1}+ dp^{-1} .$$

\item[$\mathfrak 2$.] Weak uniqueness holds for~\eqref{eq.Lan} for every initial condition $\mu \in \mathcal P(\mathbb R^{2d})$. 
\end{enumerate}
\end{theorem}
\noindent
\textbf{Note}. We also  prove that for any $T>0$ and $\mathsf x\in \mathbb R^{2d}$,  $(t,\mathsf x')\mapsto p_t(\mathsf x,\mathsf x') \in L^{q'}_TL^{p'}(\mathbb R^{2d})$ for every   $p>1$  satisfying~\eqref{eq.p0} and $q>1$  satisfying~\eqref{eq.qp}.

\begin{remark}
Let us recall that by Yamada-Watanabe type theorem for solutions to SDEs with jumps  (see~\cite[Corollary 140]{situ2005theory} or~\cite{Lepeltier,barczy2015yamada}),  if   pathwise uniqueness holds for~\eqref{eq.Lan}, then weak uniqueness holds.  Unfortunately, pathwise uniqueness  under the sole assumption that $\mathbf B$ is bounded and measurable  is hard to investigate. 
\end{remark}

Consider a    weak solution   $(X_t(\mathsf x), t\ge 0)$ of \eqref{eq.Lan} starting at $\mathsf x\in \mathbb R^{2d}$. 
To prove Item $\mathfrak 1$ in Theorem~\ref{th.uni}, 
 the goal is to show that   for every $t>0$ and $p\ge p_0$ (see~\eqref{eq.p0}), there exists $C_{p,t}>0$ such that  for all  $\mathsf x\in \mathbb R^{2d}$ and $f\in L^{p}(\mathbb R^{2d})$:
\begin{equation}\label{eq.Goal}
|\mathbb E_{\mathsf x} [ f(X_{t})]|\le C_{p,t} \Vert f \Vert_{L^{p}}.
\end{equation}  
The proof of~\eqref{eq.Goal} relies on a Duhamel formula (see~\eqref{eq.Sfd} below) relating   the non-killed semigroup of $(X_t,t\ge 0)$ with the one of an auxiliary degenerate Ornstein-Uhlenbeck process introduced in Section~\ref{sec.aux} (see  $(\hat X_t,t\ge 0)$ below). This strategy for deriving information on the semigroup of $(X_t,t\ge 0)$  is not    new (see~\cite[Chapter 7]{stroock2007multidimensional} or more recently~\cite{marino,RocknerZhang}) and  is well adapted as soon as $\alpha>1$ since in that case the drift term $\mathbf B\cdot \nabla_ v$ is a perturbation of $\mathcal S_v^{\nu_\alpha}$ in the operator infinitesimal generator  $\mathcal L^\alpha$ whose expression is given (see \eqref{eq.gene})
\begin{equation}\label{eq.La}
{\mathcal L}^\alpha= \mathcal S_v^{\nu_\alpha} + v\cdot \nabla_x + \mathbf B \cdot \nabla_v,
\end{equation}
since $\mathcal S_v^{\nu_\alpha}=\Delta_v^{\alpha/2}$.  
The derivation of the Duhamel formula is the purpose of Section~\ref{sec.duha} and we finally conclude the proof of Item $\mathfrak 1$ in Theorem~\ref{th.uni} in Section~\ref{sec.Item1}. 
We then prove  $\mathfrak 2$ in Theorem~\ref{th.uni} in Section~\ref{sec.Kl}. Recall that in all this section,   we assume {\rm\textbf{[{\footnotesize H$^{\alpha\in (1,2)}_{\text{bounded}}$}]}}.   


\subsubsection{An auxiliary process}
\label{sec.aux}

Consider  the strong solution   $(\hat X_t=(\hat x_t,\hat v_t),t\ge 0)\in \mathbb R^d \times \mathbb R^d $  to the kinetic Langevin equation
\begin{equation}\label{eq.L00}
 d\hat x_t=\hat v_tdt, \ d\hat v_t=   dL^\alpha_t.
 \end{equation} 
This process is denoted by $( \hat X_t(\mathsf x)=(\hat x _t(\mathsf x), \hat v _t(\mathsf x)),t\ge 0)$ when $\hat X_0= \mathsf x$.  Denote by $\hat{\mathcal L}^\alpha$ its infinitesimal generator and by $(\hat P_t,t\ge 0)$ its semigroup:
 $$\hat{\mathcal L}^\alpha = \mathcal S_v^{\nu_\alpha} + v\cdot \nabla_x\  \text{ and }\  \hat P_tf(\mathsf x)= \mathbb E_{\mathsf x}[f(\hat X_t)], \ f \in bB(\mathbb R^{2d}), \mathsf x\in \mathbb R^{2d}.
 $$       
We now gather known results on the semigroup of the process  $(\hat X_t,t\ge 0)$.  
  By~\cite{priola2009densities,zhang2017fundamental}, for $f \in C_c^\infty(\mathbb R^{2d})$, we have that 
\begin{equation}\label{eq.regP}
  (t,\mathsf x)\mapsto \hat P_tf(\mathsf x)\in C^\infty(   \mathbb R_+^*\times \mathbb R^{2d}), \ \forall t>0, \  \hat P_tf \in C^\infty_b(\mathbb R^{2d}  ). 
\end{equation} 
and  in the classical sense, it holds: 
\begin{equation}\label{eq.Es-p}
 \partial _t   \hat P_{t}f(\mathsf x)= \hat{\mathcal L}^\alpha  \hat P_{t}f(\mathsf x), \ \forall  (t,\mathsf x ) \in \mathbb R_+^* \times \mathbb R^{2d}.
\end{equation}   
We  also need estimates in the supremum norm over $\mathbb R^{2d}$ of both $ \hat P_t$ and $ \nabla _v  \hat P_t$. This is the content of what follows. 
Set   for $p\in [1,+\infty]$,   $\mathscr A_0(p)= (1+\alpha d)p^{-1}+ dp^{-1} $.      
 By~\cite[Lemma 2.16]{RocknerZhang}, with $\mathbf p'=(+\infty,+\infty)$, $\beta'=1$, $\beta=0$ and $\mathbf p_0= (p,p)\in [1,+\infty]^2$ there, 
there exists $C_p>0$ (independent of $f$) such that  for all $t>0$: 
\begin{align}\label{eq.LpB}
 \Vert \nabla _v  \hat P_tf \Vert_{ B_{\mathbf p';\mathbf a}^{1,1}} + \Vert \hat P_t f\Vert_{ B_{\mathbf p';\mathbf a}^{1,\infty}} \le C_p (1\wedge t)^{-\frac{1+\mathscr A_0(p)}{\alpha}}\Vert f \Vert_{  B_{\mathbf p_0;\mathbf a}^{0,\infty}}.
  \end{align} 
 Above, $ B_{\mathbf p; \mathbf a}^{s,q}$ denotes the anisotropic Besov space over $\mathbb R^{2d}$, see~\cite[Definition 2.2]{RocknerZhang}.  
We also recall  the following Besov's embedding (see~\cite[Item (i) in Lemma 2.6]{RocknerZhang}) valid for all $\mathbf p =(p,p)$ with $p\in [1,+\infty]$, 
$$ B^{1,\infty}_{\mathbf p;\mathbf a}\hookrightarrow  B^{0,1}_{\mathbf p;\mathbf a} \hookrightarrow  L^{ p}(\mathbb R^{2d}) \hookrightarrow  B^{0,\infty}_{\mathbf p;\mathbf a}.$$ 
Note that   because $\mathsf x\in \mathbb R^{2d}\mapsto \nabla _v  \hat P_t f(\mathsf x)$ is continuous for $t>0$, it holds that $\Vert \nabla _v  \hat P_t f\Vert_{L^\infty}= \Vert \nabla _v  \hat P_t f\Vert_{\infty}$ where we recall that $ \Vert \nabla _v  \hat P_t f\Vert_{\infty}=\sup_{\mathbb R^{2d}} |\nabla _v  \hat P_t f|  $.
We then  deduce (see   also~\cite[Remark 1.6]{RocknerZhang}) that  for all $p\in [1,+\infty]$ satisfying \eqref{eq.p0}, there exists $C_p>0$ (independent of $f$) such that  for all $t>0$: 
 \begin{align}\label{eq.LpB}
 \Vert \nabla _v  \hat P_tf \Vert_{\infty} + \Vert \hat P_t f\Vert_{\infty} \le C_p (1\wedge t)^{-\frac{1+\mathscr A_0(p)}{\alpha}}\Vert f \Vert_{ L^{  p}}.
  \end{align} 
In what follows, we consider $p>1$ is such that   (recall that $\alpha\in (1,2)$)
 \begin{align}\label{eq.p0}
 1+ \mathscr A_0(p)<\alpha. 
  \end{align}
 Finally, we also mention that~\eqref{eq.LpB} also holds for $p=+\infty$ (in that case $\mathscr A_0(p)=0)$.  


\subsubsection{A Duhamel formula}
\label{sec.duha}

 To show~\eqref{eq.Goal}, we prove the perturbative   formula~\eqref{eq.Sfd} below. The proof of this formula requires special attention near $0^+$ since even if $f$ is smooth,    $\hat P_tf$ and its   derivatives  are singular when $t\to 0^+$ (and these singularities depend  on $\alpha$, see~\eqref{eq.LpB}). This is where the condition $\alpha\in (1,2)$ appears.   This Duhamel formula is proved using Itô calculation. 
 \medskip
 
 Consider a weak solution $(X_t,t\ge 0)$ of \eqref{eq.Lan}. 
Recall that  $ {\mathcal L}^\alpha= \mathcal S_v^{\nu_\alpha} + v\cdot \nabla_x + \mathbf B \cdot \nabla_v  $ denotes  the infinitesimal generator associated with equation~\eqref{eq.Lan}, see \eqref{eq.La}. 
 Let $f\in C_c^{\infty}(\mathbb R^{2d})$, $t>0$, and $0<\epsilon<t$.
For $s\in[0,t]$, define the function
$$
\Phi(s,\mathsf x):=\hat P_{t-s}f(\mathsf x).
$$
Let $t>0$. 
Applying It\^o's formula~\cite[Theorem 4.4.7]{applebaum2009levy} to the smooth function (see indeed~\eqref{eq.regP})
$(s,\mathsf x)\in [0,t)\times \mathbb R^{2d} \mapsto \Phi(s,\mathsf x)$ along the process $(X_s,s\ge 0)$, we obtain
\begin{align}
\nonumber
d \Phi(s,X_s)
&=\partial_s\Phi(s,X_s)\,ds
+\mathcal L^\alpha\Phi(s,X_s)\,ds
+dM_s(f) \\
\nonumber
&=\big(-\hat{\mathcal L}^\alpha\hat P_{t-s}f
+\mathcal L^\alpha\hat P_{t-s}f\big)(X_s)\,ds
+dM_s(f) \\
\nonumber
&=(\mathcal L^\alpha-\hat{\mathcal L}^\alpha)\hat P_{t-s}f(X_s)\,ds
+dM_s(f) \\
\label{eq.Ito1}
&=\mathbf B(X_s)\cdot\nabla_v
\hat P_{t-s}f(X_s)\,ds
+dM_s(f),
\end{align}
where 
$$M_s(f) = \int_0^s \int_{\mathbb{R}^d \setminus \{0\}} \Big[\hat{P}_{t-r}f(x_{r-}, v_{r-} + z) - \hat{P}_{t-r}f(x_{r-}, v_{r-}) \Big] \tilde{N}_\alpha(dr, dz).$$
We recall that $(M_s(f),s\in [0,t-\epsilon])$ is a $L^2$-martingale provided that
 $$\mathbb{E} \left[ \int_0^{t-\epsilon} \int_{\mathbb{R}^d \setminus \{0\}} | \hat{P}_{t-r}f(x_{r-}, v_{r-} + z) - \hat{P}_{t-r}f(x_{r-}, v_{r-})  |^2 \, \nu_\alpha(dz) dr \right] < \infty.$$
 Note that this condition is satisfied for the following reasons:
 \begin{enumerate}
 \item[(M1)]  Since for all $u\ge 0$, $|\hat{P}_{u}f (\mathsf y)|\le  \Vert f\Vert_\infty$, it holds for every $0<\epsilon \le r\le t$ and $(x,v)\in \mathbb R^{2d}$:
 $$ \int_{|z|\ge1} \big | \hat{P}_{t-r}f(x, v  + z) - \hat{P}_{t-r}f(x , v ) \big  |^2 \, \nu_\alpha(dz)  \le 4 \Vert f\Vert_\infty \int_{|z|\ge 1} \nu_\alpha(dz). $$
  \item[(M2)] For every $\epsilon \le r\le t$ and $(x,v)\in \mathbb R^{2d}$, by the mean value theorem:
\begin{align*}
 \int_{|z|<1} | \hat{P}_{t-r}f(x, v  + z) - \hat{P}_{t-r}f(x , v )  |^2 \, \nu_\alpha(dz)  \le 2 \Vert \nabla_v \hat P_{t-r}f\Vert^2_\infty  \int_{|z|<1} |z|^2 \nu_\alpha(dz). 
 \end{align*}
 The r.h.s of the previous inequality is integrable over $[0,t-\epsilon]$ by  \eqref{eq.LpB}.  
 \end{enumerate}
 Consequently, $(M_s(f),s\ge 0)$ is a $L^2$-martingale. 
 Integrating over $s\in  [0,t-\epsilon]$ the equality \eqref{eq.Ito1}, we obtain
$$
\hat P_{\epsilon}f(X_{t-\epsilon})
-\hat P_t f(X_0)
=
\int_0^{t-\epsilon}
\mathbf B(X_s)\cdot
\nabla_v\hat P_{t-s}f(X_s)\,ds
+M_{t-\epsilon}(f).
$$
Taking expectations under $\mathbb P_{\mathsf x}$ and using
$\mathbb E_{\mathsf x}[M_{t-\epsilon}(f)]=0$, we finally obtain
\begin{equation}\label{eq.RS}
\mathbb E_{\mathsf x}\big [\hat P_{\epsilon}f(X_{t-\epsilon})\big ]
=
\hat P_t f(\mathsf x)
+
\mathbb E_{\mathsf x}\Big[
\int_0^{t-\epsilon}
\mathbf B(X_s)\cdot
\nabla_v\hat P_{t-s}f(X_s)\,ds
\Big].
\end{equation}
We now want to pass to the limit $\epsilon\to 0$ in the previous equation. 
 Using~\eqref{eq.LpB},  for any  $p>1$  satisfying~\eqref{eq.p0}, there exists $C>0$ (independent of $f$) such that for all $t>0$ and $\mathsf x\in \mathbb R^{2d}$:
   \begin{align*} 
  \big| \,   \mathbb E_{\mathsf x} [\mathbf B(X_s)\cdot    \nabla _v  \hat P_{t-s}f (X_s)]  \,  \big| \le C (1\wedge (t-s))^{-\frac{1+\mathscr A_0(p)}{\alpha}} \Vert \mathbf B\Vert_{\infty} \Vert f \Vert_{L^{p}},
 \end{align*}
 which is integrable over $[0,t]$ because  $1+ \mathscr A_0(p)<\alpha$. As a result,
 $$
\lim_{\epsilon\to 0 }  \int_0^{t-\epsilon}  \mathbb E_{\mathsf x}\Big[\mathbf B(X_s)\cdot    \nabla _v  \hat P_{t-s}f (X_s)   \Big] ds= \int_0^{t}  \mathbb E_{\mathsf x}\Big[\mathbf B(X_s)\cdot    \nabla _v  \hat P_{t-s}f (X_s)   \Big] ds.
$$
 The term  $ \mathbb E_{\mathsf x} [\hat P_\epsilon f(X_{t-\epsilon})]$ requires more attention as it involves a double limit and that $\hat P_\epsilon$ becomes singular when $\epsilon \to 0$. 
 We write 
   \begin{align} 
   \label{eq.deco}
   \mathbb E_{\mathsf x}[\hat P_\epsilon f(X_{t-\epsilon})]- \mathbb E_{\mathsf x}[ f(X_{t})]&=\mathbb E_{\mathsf x}[\hat P_\epsilon f(X_{t-\epsilon})] - \mathbb E_{\mathsf x}[\hat P_\epsilon f(X_{t})]\\
   \label{eq.deco2}
   &\quad  +  \mathbb E_{\mathsf x}[\hat P_\epsilon f(X_{t})]- \mathbb E_{\mathsf x}[ f(X_{t})].
     \end{align}
 We first deal with the  term appearing in~\eqref{eq.deco2}.
 Note that since $f$ is continuous and bounded, and a.s. $\hat X_s(\mathsf y)\to \mathsf y$ as $s\to 0^+$, $ \hat P_\epsilon f(\mathsf y )\to f(\mathsf y )$ for all $\mathsf y\in \mathbb R^{2d}$. Thus,  by dominated convergence theorem,  as $\epsilon \to 0$, it holds: 
 $$\mathbb E[ \hat P_\epsilon f(X_{t}({\mathsf x}) )]\to \mathbb E[ f(X_t(\mathsf x) )].$$ 
 We now deal with the  term in the right-hand side of~\eqref{eq.deco}.
 Recall the definitions of the anisotropic Hölder-Zygmund space $C_{\mathbf a}^\beta(\mathbb R^{2d})$ in Section~\ref{sec.Nota}.   
 We have  for every $\epsilon>0$, 
 \begin{align}
 \nonumber
 \mathbb E_{\mathsf x}\big[\big |\hat P_\epsilon f(X_{t-\epsilon}) -  \hat P_\epsilon f(X_{t}) \big  | \big]&\le   \Vert \hat P_\epsilon f \Vert_{C_{\mathbf a}^\beta} \ \mathbb E_{\mathsf x}[\, |X_{t}-X_{t-\epsilon}|^\beta _{\mathbf a} \, ]\\
 \label{eq.Cab}
 &\le \Vert \hat P_\epsilon f \Vert_{C_{\mathbf a}^\beta}  \Big [ \mathbb E_{\mathsf x}[|x_{t}-x_{t-\epsilon}|^{\frac{\beta}{1+\alpha}}] + \mathbb E_{\mathsf x}[|v_{t}-v_{t-\epsilon}|^{\beta}  ]\Big].  
  \end{align}
  By ~\cite[Lemma 2.16]{RocknerZhang} (with $\mathbf p'=\mathbf p_0=(+\infty,+\infty)$, $\beta'=\beta\in (0,1)$),  there exists $C_\beta>0$ (independent of $f$ and $\epsilon \in (0,1)$) such that  
    \begin{align*}
\mathbb E_{\mathsf x}\big[\big |\hat P_\epsilon f(X_{t-\epsilon}) -  \hat P_\epsilon f(X_{t}) \big  | \big] \le C_\beta \Vert f \Vert_{C_{\mathbf a}^\beta}  \Big [ \mathbb E_{\mathsf x}[|x_{t}-x_{t-\epsilon}|^{\frac{\beta}{1+\alpha}}] + \mathbb E_{\mathsf x}[|v_{t}-v_{t-\epsilon}|^{\beta}  ]\Big].  
  \end{align*}
Recall that $v_{t}(\mathsf x)= v+ \int_0^t \mathbf B(X_s) ds+ L_t^\alpha$, $t\ge 0$. 
  Since $\mathbf B$ is bounded,
  $$\mathbb E_{\mathsf x} [ |v_{t-\epsilon} -v_t |^{\beta}] \le C\Big [ \epsilon^\beta + \mathbb E \big [ |L^\alpha _{ \epsilon} |^{\beta}\big ] \Big ]  .$$
  Moreover,
  $$\mathbb E_{\mathsf x}\big [|x_t-x_{t-\epsilon}| ^{\frac{\beta}{1+\alpha}}\big ]=  \mathbb E_{\mathsf x}\Big [ \big |\int_{t-\epsilon}^t    v_s   \,   ds \big|^{\frac{\beta}{1+\alpha}} \Big ] \le \epsilon^{\frac{\beta}{1+\alpha}}\,   \mathbb E_{\mathsf x}\big [ \sup_{s\in [0,t]}|v_s|^{\frac{\beta}{1+\alpha}}\big  ].$$ 
For  $s\in[0,t]$,   it holds since $\mathbf B$ is bounded:
$$\mathbb E[|v_{s}(\mathsf x)|^{\frac{\beta}{1+\alpha}}]\le  |v|^{\frac{\beta}{1+\alpha}}+ Ct^{\frac{\beta}{1+\alpha}}   +\mathbb E[\sup_{0\le s\le t} |L^\alpha_s|^{\frac{\beta}{1+\alpha}}].$$
 Moreover, using Doob’s inequality in $L^m$  (for some fixed $m>1$ with $ \frac{m \beta}{1+\alpha}<\alpha )$, see~\cite[Proposition 3.15 (ii)]{le2016brownian}), we have that $$ \mathbb  E  \big [ \sup_{s\in [0,t]}|L^\alpha_s|^{\frac{\beta}{1+\alpha}}  \big ] \le  \mathbb  E  \big [ \sup_{s\in [0,t]}|L_s^\alpha| ^{\frac{m\beta}{1+\alpha}}  \big ]^{\frac 1m }  \le C  \mathbb  E \big [  |L_t^\alpha|^{\frac{m\beta}{1+\alpha}}  \big ]^{\frac 1m} .$$  As a result, for all $\mathsf x\in \mathbb R^{2d}$ and $t>0$, 
$$\lim_{\epsilon\to 0}\mathbb E_{\mathsf x}[|x_{t}-x_{t-\epsilon}|^{\frac{\beta}{1+\alpha}}] + \mathbb E_{\mathsf x}[|v_{t}-v_{t-\epsilon}|^{\beta}  ]= 0.$$ 
Hence, for every $f\in C_c^\infty(\mathbb R^{2d})$ and $t>0$:
\begin{equation}\label{eq.Eee}
\lim_{\epsilon\to 0} \mathbb  E_{\mathsf x} \Big[\big |\hat P_\epsilon f(X_{t-\epsilon}) -  \hat P_\epsilon f(X_{t}) \big  | \Big]=0 .
\end{equation}
  In conclusion (see~\eqref{eq.deco} and~\eqref{eq.deco2}),  $ \lim_{\epsilon\to 0} \mathbb E_{\mathsf x} [\hat P_\epsilon f(X_{t-\epsilon})]=\mathbb E_{\mathsf x} [ f(X_{t})]$.  We thus have for all $\mathsf x\in \mathbb R^{2d}$, $t>0$ and $f\in  C_c^\infty(\mathbb R^{2d})$: 
   \begin{align}\label{eq.Sfd} 
   \mathbb E_{\mathsf x} [ f(X_{t})]
=
\hat P_t f(\mathsf x)
+
\mathbb E_{\mathsf x}\Big[
\int_0^{t}
\mathbf B(X_s)\cdot
\nabla_v\hat P_{t-s}f(X_s)\,ds
\Big].
   \end{align}

\subsubsection{Proof of Item $\mathfrak 1$ in Theorem~\ref{th.uni}}
\label{sec.Item1}
Let   $(X_t(\mathsf x),t\ge 0)$ be a weak solution to~\eqref{eq.Lan}    starting at $\mathsf x$. 
In this step, we prove  Item $\mathfrak 1$  in Theorem~\ref{th.uni}. 
To this end, we show~\eqref{eq.Goal} and  then prove  that $X_t$ has density w.r.t. the Lebesgue measure. 
\medskip

Let $p>1$  satisfying~\eqref{eq.p0}. 
Using~\eqref{eq.LpB}, and since $1+ \mathscr A_0(p)<\alpha$ and $t>0$, there exists $C_{p,t}>0$ (independent of $f$) such that for all $t>0$ and $\mathsf x\in \mathbb R^{2d}$:
   \begin{align*} 
  \Big| \int_0^t  \mathbb E_{\mathsf x} [\mathbf B(X_s)\cdot    \nabla _v  \hat P_{t-s}f (X_s)] ds \Big| \le C_{p,t}\Vert \mathbf B\Vert_{\infty} \Vert f \Vert_{L^{p}},\ \forall \mathsf x\in \mathbb R^{2d}.
 \end{align*}
Consequently, using~\eqref{eq.Sfd} and~\eqref{eq.LpB} for the term $\hat P_t f$, for every $t>0$ and $p>1$  satisfying~\eqref{eq.p0}, there exists $C_{p,t}>0$ such that for all $\mathsf x\in \mathbb R^{2d}$  and  $f\in  C_c^\infty(\mathbb R^{2d})$: 
\begin{equation}\label{eq.Dua}
|  \mathbb E_{\mathsf x} [ f(X_{t})]|  \le  C_{p,t}  \Vert f \Vert_{ L^{  p}}. 
\end{equation}
Note that it is straightforward to get that one can choose   $$C_{p,t}= c_p\Big[(1\wedge t )^{-\frac{ 1+\mathscr A_0(p)}{\alpha}} + \int_0^t (1\wedge u )^{-\frac{1+\mathscr A_0(p)}{\alpha}} du\Big],$$
 for some $c_p>0$ independent of $t> 0$. 
 This is precisely~\eqref{eq.Goal} for function $f\in  C_c^\infty(\mathbb R^{2d})$.  
 Because for $f\in  C_c^\infty(\mathbb R^{2d})$, $P_tf(\mathsf x)= \int_{\mathbb R^{2d}} f(\mathsf x') P_t(\mathsf x, d\mathsf x')$ (where $P_t(\mathsf x, d\mathsf x')$ is the law of $X_t(\mathsf x)$),~\eqref{eq.Dua} implies that the distribution $P_t(\mathsf x, d\mathsf x')$ lies in $L^{  p'}(\mathbb R^{2d})$    (where $p'=p/(p-1)$). Thus for all  $p>1$  satisfying~\eqref{eq.p0}, 
 there exists $p_t(\mathsf x, \cdot)\in L^{  p'}(\mathbb R^{2d})$
 and for all $f\in L^{p}( \mathbb R^{2d})$,
 $$\mathbb E_{\mathsf x} [ f(X_{t})]= \int_{\mathbb R^{2d}} f(\mathsf x') p_t(\mathsf x,  \mathsf x')d\mathsf x'.$$
 Note that  the density  $p_t(\mathsf x, \cdot)$ does not depend on  $p\ge p_0$.  
   Moreover, $p_t(\mathsf x,\cdot)$ is non-negative  and belongs to $L^1(\mathbb R^{2d})$.  This is standard, but we write the arguments for the sake of completeness. Set $\mathscr E_{t,\mathsf x}:=\{p_t(\mathsf x,\cdot)<0\}$ and observe that    for every  bounded Borel set $\mathscr A$,  $\mathbb E_{\mathsf x} [ \mathbf 1_{\mathscr A}(X_{t})] = \int_{\mathscr A} p_t(\mathsf x, \mathsf x') d\mathsf x'$ (as $\mathbf 1_{\mathscr A}\in L^{m}(\mathbb R^{2d})$, for all $m\ge 1$). 
Then, we have: 
\begin{align*}
 0\le   \mathbb E_{\mathsf x} [ \mathbf  1_{\mathscr E_{t,\mathsf x}\cap  \mathscr B_{\mathbb R^{2d}}(0,R) } (X_t)] =  \int_{\mathscr E_{t,\mathsf x}\cap  \mathscr B_{\mathbb R^{2d}}(0,R) } p_t(\mathsf x, \mathsf x') d\mathsf x'\le 0.
\end{align*}
so that for all $t>0$ and $\mathsf x\in \mathbb R^{2d}$, $p_t(\mathsf x,\mathsf x')\ge 0$ ($d\mathsf x'$-a.e.). By monotone convergence theorem, we also have:
\begin{align*}
1=  \mathbb E_{\mathsf x} [ \mathbf  1  (X_t)]&= \lim_{R\to +\infty} \mathbb E_{\mathsf x} [ \mathbf  1_{  \mathscr B_{\mathbb R^{2d}}(0,R) }(X_t)] = \lim_{R\to +\infty}\int_{ \mathscr B_{\mathbb R^{2d}}(0,R) } p_t(\mathsf x, \mathsf x') d\mathsf x' .
\end{align*}
Hence, for all $t>0$ and $\mathsf x\in \mathbb R^{2d}$, $p_t(\mathsf x, \mathsf x') d\mathsf x' \in \mathcal P(\mathbb R^{2d})$ and, since the set $C_c(\mathbb R^{2d})$ is \textit{separating}~\cite[Proposition 4.4 in Chapter 3]{EK},   the measures  $p_t(\mathsf x, \mathsf x') d\mathsf x'$ and $P_t(\mathsf x,d\mathsf x')$ are equal.  We have thus  proved that for all $\mathsf x\in \mathbb R^{2d}$ and $t>0$, $X_t(\mathsf x)$   has a density $ p_t(\mathsf x, \mathsf x') $ w.r.t. $d\mathsf x'$ and   $p_t(\mathsf x, \cdot)\in L^{  p'}(\mathbb R^{2d})$ for any $p>1$  satisfying~\eqref{eq.p0}.  Note that by a density argument, the inequality~\eqref{eq.Dua} then extends to every $f\in L^p(\mathbb R^{2d})$, for every $p>1$  satisfying~\eqref{eq.p0}. 
 The proof of Item $\mathfrak 1$ in Theorem~\ref{th.uni} is thus complete. 
Note that the last statement of Item $\mathfrak 1$ follows from 
$$\Vert p_t(\mathsf{x}, \cdot)\Vert _{L^{p'}} = \sup_{f \in L^p(\mathbb R^{2d}), f \neq 0} \frac{\left| \int_{\mathbb R^{2d}} f(\mathsf{y}) p_t(\mathsf{x}, \mathsf{y}) d\mathsf{y} \right|}{\Vert f\Vert _{L^p}} \le C_{p,t}, \ \forall t>0.$$



\subsubsection{Proof of Item $\mathfrak 2$ in Theorem~\ref{th.uni}}
\label{sec.Kl}
Let $T>0$ and $0\le t<T$. 
The proof is divided into several steps.
\medskip

\noindent
\textbf{Step 1}. In this step we prove the integrated perturbation formula~\eqref{eq.DI} below for test functions depending also on time. This formula is the cornerstone to derive so-called Krylov's estimates which play a key role in the proof of $\mathfrak 2$ in Theorem~\ref{th.uni} (and also in the proof below of the existence of a weak solution to~\eqref{eq.Lan}). Let $0<\epsilon<T$. 
For $t\in [0,T-\epsilon]$ and $f\in C_c^{\infty}([0,T]\times\mathbb R^{2d})$,
we consider the Green function
$$
\hat G_\epsilon f(t,\mathsf x)
:= 
\int_{t+\epsilon}^T
\hat P_{s-t}f(s,\cdot)(\mathsf x)\,ds .
$$
Note that $ \hat G_\epsilon f(t,\mathsf x) = \int_{\epsilon}^{T-t} \hat P_u f(t+u, \cdot)(\mathsf x) \, du $. 
Since $\epsilon>0$, this function is smooth and we have  thanks to~\eqref{eq.Es-p},  
\begin{equation}\label{eq:GEpsPDE}
(\partial_t+\hat{\mathcal L}^\alpha)\hat G_\epsilon f(t,\mathsf x)
=
-\hat P_\epsilon f(t+\epsilon,\cdot)(\mathsf x).
\end{equation}
Applying It\^o's formula  to
$(t,\mathsf x)\mapsto \hat G_\epsilon f(t,\mathsf x)$ along $X_t$ yields
$$
\begin{aligned}
d\hat G_\epsilon f(t,X_t)=\partial_t \hat G_\epsilon f(t,X_t)dt+\hat{\mathcal L}^\alpha \hat G_\epsilon f(t,X_t)\,dt+ 
 (\mathcal L^\alpha-\hat{\mathcal L}^\alpha)\hat G_\epsilon f(t,X_t)\,dt
+dM^*_t(f),
\end{aligned}
$$
where  
$$M^*_t(f) = \int_0^t \int_{\mathbb{R}^d \setminus \{0\}} \Big[\hat G_\epsilon f(r,(x_{r-}, v_{r-} + z)) - \hat G_\epsilon f(r,(x_{r-}, v_{r-})) \Big] \tilde{N}_\alpha(dr, dz).$$
Using~\eqref{eq:GEpsPDE}, we obtain that:
$$
d\hat G_\epsilon f(t,X_t)
=
-\hat P_\epsilon f(t+\epsilon,X_t)\,dt
+\mathbf B(x_t,v_t)\cdot\nabla_v\hat G_\epsilon f(t,X_t)\,dt
+dM^*_t(f).
$$
The process $(M^*_t(f),t\in [0,T-\epsilon])$ is a $L^2$-martingale since 
 $$\mathbb{E} \left[ \int_0^{T-\epsilon} \int_{\mathbb{R}^d \setminus \{0\}} | \hat G_\epsilon f(r,(x_{r-}, v_{r-} + z)) - \hat G_\epsilon f(r,(x_{r-}, v_{r-})) |^2 \, \nu_\alpha(dz) dr \right] < \infty.$$
 This is proved using the same arguments as in Items (M1) and (M2) above and using also the fact that for all $\mathsf x\in \mathbb R^{2d}$ and $r\in [0,T-\epsilon]$,  $|\hat G_\epsilon f(r,\mathsf x) |\le T \sup_{[0,T]\times \mathbb R^{2d}}|f|$ and 
 $$|\nabla_v\hat G_\epsilon f(r,\mathsf x) |\le \int_{\epsilon}^{T} \sup_{\mathsf y \in \mathbb R^{2d}}|\nabla _v \hat P_u f(r+u,\cdot)(\mathsf y)| du\le  \int_{\epsilon}^{T} (1\wedge u)^{-1/\alpha}   du \, \sup_{[0,T]\times \mathbb R^{2d}}|f| .$$
  As a consequence, integrating over $[0,T-\epsilon]$ and taking expectations, since in addition $\hat G_\epsilon f(T-\epsilon,\cdot )=0$, we deduce that:
\begin{align*}
&\hat G_\epsilon f(0,\mathsf x)\\
&=
\mathbb E_{\mathsf x}\Big[\int_0^{T-\epsilon}
\hat P_\epsilon f(t+\epsilon,\cdot)(X_t)\,dt\Big]
-
\mathbb E_{\mathsf x}\Big[\int_0^{T-\epsilon}
\mathbf B(X_t)\cdot\nabla_v\hat G_\epsilon f(t,X_t)\,dt\Big]
\\
&=\mathbb E_{\mathsf x}\Big[\int_0^{T-\epsilon}
\hat P_\epsilon f(t+\epsilon,\cdot)(X_t)\,dt\Big]
- \mathbb E_{\mathsf x}\Big[\int_0^{T-\epsilon}
\mathbf B(X_t)\cdot \int_{t+\epsilon}^T\nabla_v
\hat P_{s-t}f(s,\cdot)(X_t)\,ds  \,dt\Big].
\end{align*}
Let us now   pass to the limit $\epsilon\to 0$ in the previous equation. It is not difficult to prove that, 
$$  \int_0^{T-\epsilon}
\mathbb E_{\mathsf x}[\hat P_\epsilon f(t+\epsilon,\cdot)(X_t)]\, dt=   \int_0^{T}
\mathbb E_{\mathsf x}[\hat P_\epsilon f(t+\epsilon,\cdot)(X_t)]\, dt + o_\epsilon(1).$$
Note that $\mathbb E_{\mathsf x}[\hat P_\epsilon f(t+\epsilon,\cdot)(X_t)]= \mathbb E_{\mathsf x}[\mathbb E_{X_t}[ f(t+\epsilon,\hat X_\epsilon)]]$. 
In addition,  for all $\mathsf y\in \mathbb R^{2d}$, $\mathbb E_{\mathsf y}[ f(t+\epsilon,\hat X_\epsilon)]\to f(t,\mathsf y)$.  
Hence, by dominated convergence theorem, we obtain that:
$$\lim_{\epsilon\to 0^+} \int_0^{T-\epsilon}
\mathbb E_{\mathsf x}[\hat P_\epsilon f(t+\epsilon,X_t)]\, dt= \int_0^T \mathbb E_{\mathsf x}[f(t,X_t)]\, dt.$$
On the other hand, write $\int_{t+\epsilon}^T\nabla_v
\hat P_{s-t}f(s,\cdot)(\mathsf x) \, ds= \int_{ \epsilon}^{T-t}\nabla_v
\hat P_{u}f(t+u ,\cdot)(\mathsf x) \, du$. 
Moreover, using~\eqref{eq.LpB}  with $p=+\infty$ (recall that $\alpha>1$), there exists $C>0$, such that for all $f\in C_c^{\infty}([0,T]\times\mathbb R^{2d})$,  $\mathsf x\in \mathbb R^{2d}$, and every $t\in [0,T]$,  
$$  \int_{0 }^{T-t}\big|\nabla_v
\hat P_{u}f(t+u ,\cdot)(\mathsf x) \big|\,  du \le   C \sup_{[0,T]\times \mathbb R^{2d}}|f|\, \Big[ \frac{\alpha}{\alpha-1}\big((T-t)\wedge 1\big)^{\frac{\alpha-1}{\alpha}} 
+(T-t-1)^+\Big].$$
Hence, we have:
$$\lim_{\epsilon\to 0^+} \int_{ \epsilon}^{T-t}\nabla_v
\hat P_{u}f(t+u ,\cdot)(\mathsf x) \, du=  \int_{0 }^{T-t}\nabla_v
\hat P_{u}f(t+u ,\cdot)(\mathsf x) \, du.$$ 
 Since $\mathbf B$ is bounded, we then get that
\begin{align*}
&\lim_{\epsilon\to 0^+}  \mathbb E_{\mathsf x}\Big[\int_0^{T-\epsilon}
\mathbf B(X_t)\cdot \int_{t+\epsilon}^T\nabla_v
\hat P_{s-t}f(s,\cdot)(X_t)\,ds  \,dt\Big]\\
&= \mathbb E_{\mathsf x}\Big[\int_0^{T}
\mathbf B(X_t)\cdot  \int_{0 }^{T-t}\nabla_v
\hat P_{u}f(t+u ,\cdot)(X_t) \, du   \,dt\Big].
\end{align*}
In conclusion, for any $\mathsf x\in \mathbb R^{2d}$, $T>0$, and $f\in C_c^{\infty}([0,T]\times\mathbb R^{2d})$ we have that 
\begin{align}
\nonumber
 \int_{0 }^T
\hat P_{s}f(s,\cdot)(\mathsf x)\,ds&=\mathbb E_{\mathsf x}\Big[\int_0^T f(t,X_t)\,dt\Big]  \\
\label{eq.DI}
&\quad +\mathbb E_{\mathsf x}\Big[\int_0^{T}
\mathbf B(X_t)\cdot  \int_{0 }^{T-t}\nabla_v
\hat P_{u}f(t+u ,\cdot)(X_t) \, du   \,dt\Big].
\end{align}


\noindent
\textbf{Step 2}. We now derive Krylov's type estimates using ~\eqref{eq.LpB} and~\eqref{eq.DI}. For ease of notation, denote by $E_2(T,f,\mathsf x)$ the last  term in
~\eqref{eq.DI}. 
 We recall that, for $p,q\in [1,+\infty]$ the Bochner space $L^q_TL^p(\mathbb R^{2d})$ is defined as the set of Borel functions $f: [0,T]\times \mathbb R^{2d}\to \mathbb R$ such that 
 $$\Vert f\Vert_{L^q_TL^p}:=\Big[\int_0^T \Vert f(s,\cdot)\Vert_{L^p} ^q ds\, \Big]^{1/q}<+\infty,$$
 where we recall (see Section~\ref{sec.Nota}) that $\Vert f(s,\cdot)\Vert_{L^p}=  \big [\int_{\mathbb R^{2d}}|f(s,x,v)|^p\,dx\,dv \big ]^{1/p}$, with obvious changes if $p$ or $q$ are infinite. 
  Let  $T>0$ and $f\in C_c^{\infty}([0,T]\times\mathbb R^{2d})$. We have using ~\eqref{eq.LpB} and since $\mathbf B$ is bounded, for $t\in [0,T]$ and $p>1$  satisfying~\eqref{eq.p0}, 
 \begin{align*}
 &\int_0^{T-t}  \sup_{\mathsf y\in \mathbb R^{2d}}  |\mathbf B(\mathsf y)\cdot  \nabla_v \hat P_{u}f(t+u ,\cdot)(\mathsf y )|    \, du  \\
 &\le \Vert \mathbf B\Vert_\infty C_p \int_0^{T-t} (1\wedge u)^{-\frac{1+\mathscr A_0(p)}{\alpha}}\Vert f (t+u,\cdot) \Vert_{ L^{  p}}    \, du  .
 \end{align*}
 Let $p>1$  satisfying~\eqref{eq.p0}. 
Consider  $q>1$ such that
\begin{equation}\label{eq.qp}
 \frac{1+\mathscr A_0(p)}{\alpha}\frac{q}{q-1} <1.
\end{equation}
Note that this is possible since $1+\mathscr A_0(p) <\alpha$. 
 By Hölder's inequality in time, we obtain that: 
 \begin{align*}
 &\int_0^{T-t} (1\wedge u)^{-\frac{1+\mathscr A_0(p)}{\alpha}}\Vert f (t+u,\cdot) \Vert_{ L^{  p}}    \, du \\
 &\le  \underbrace{\Big[ \int_0^{T} (1\wedge u)^{-\frac{1+\mathscr A_0(p)}{\alpha}\frac{q}{q-1} }\, du \Big]^{\frac{q-1}{q} }}_{:=\eta(T,p,q)<+\infty} \Big[ \int_0^{T}\Vert f (s,\cdot) \Vert_{ L^{  p}} ^q\, ds \Big]^{\frac{1}{q} }.
 \end{align*}
Consequently, we obtain that:
$$E_2(T,f,\mathsf x)\le \Vert \mathbf B\Vert_\infty C_p  T \, \eta(T,p,q) \,  \Vert f\Vert_{L^q_TL^p}.$$ 
On the other hand, with the same arguments, we also have thanks to~\eqref{eq.LpB}, for every  $p>1$  satisfying~\eqref{eq.p0} and $q>1$  satisfying~\eqref{eq.qp}:
$$\Big | \int_{0 }^T
\hat P_{s}f(s,\cdot)(\mathsf x)\,ds\Big|\le  C_p\, \eta(T,p,q) \,   \Vert f\Vert_{L^q_TL^p}.$$
In conclusion, by~\eqref{eq.DI},  we have the following Krylov's estimate: for every $\mathsf x\in \mathbb R^{2d}$,     $p>1$  satisfying~\eqref{eq.p0}, $q>1$  satisfying~\eqref{eq.qp}, and every $f\in C_c^{\infty}([0,T]\times\mathbb R^{2d})$:
  \begin{equation}\label{eq.Kt}
  \Big| \mathbb E_{\mathsf x}\Big[\int_0^T f(t,X_t)\,dt\Big] \Big| \le  [\Vert \mathbf B\Vert_\infty +1]C_p  T \, \eta(T,p,q) \,  \Vert f\Vert_{L^q_TL^p}.
  \end{equation}
Because  $p$ and $q$ are    finite,  we thus deduce by a density argument, that the inequalities~\eqref{eq.Kt} and~\eqref{eq.DI} extend to every $f\in L^{q}_TL^{p}(\mathbb R^{2d})$,  and that 
\begin{equation}\label{eq.dPp}
(t,\mathsf x')\mapsto p_t(\mathsf x,\mathsf x') \in L^{q'}_TL^{p'}(\mathbb R^{2d}).
\end{equation}

\noindent
\textbf{Step 3}.  We now conclude the proof of the uniqueness  of the weak solution to~\eqref{eq.Lan} (namely the proof of Item $\mathfrak 2$ in Theorem~\ref{th.uni}). Assume there are two weak solutions $(X_t(\mathsf x),t\ge 0)$ and $(Y_t(\mathsf x),t\ge 0)$ to~\eqref{eq.Lan} with the same  initial condition $\mathsf x$ (which might be defined on different probably spaces with respective probability $\mathbb P$ and $\bar{\mathbb P}$). By 
Item~$\mathfrak 1$ in Theorem~\ref{th.uni} and~\eqref{eq.dPp}, for every   $p>1$  satisfying~\eqref{eq.p0} and $q>1$  satisfying~\eqref{eq.qp},  both solutions have a density, denoted   by 
$$p^X_\cdot(\mathsf x,\cdot)\in L^{q'}_TL^{p'}(\mathbb R^{2d}) \text{ and } p^Y_\cdot(\mathsf x,\cdot)\in L^{q'}_TL^{p'}(\mathbb R^{2d}).$$
 By~\eqref{eq.DI}, we have for every $T>0$, $\mathsf x\in \mathbb R^{2d}$, and $f\in L^{q}_TL^{p}(\mathbb R^{2d})$ with     $p>1$  satisfying~\eqref{eq.p0} and $q>1$  satisfying~\eqref{eq.qp}:
\begin{align*}
  &\int_0^T \int_{\mathbb R^{2d}} f(t,\mathsf y ) \, q_t(\mathsf x,\mathsf y)\, d\mathsf y \, dt\\
  &=  \int_0^{T}
\int_{\mathbb R^{2d}} \underbrace{\Big[\mathbf B(\mathsf y)\cdot  \int_{0 }^{T-t}\nabla_v
\hat P_{u}f(t+u ,\cdot)(\mathsf y) \, du \Big]}_{R_T f(t,\mathsf y)}  \, q_t(\mathsf x,\mathsf y) \, d\mathsf y \,dt ,
\end{align*}
where 
$q_t(\mathsf x,\mathsf y):= p^X_t(\mathsf x,\mathsf y)-p^Y_t(\mathsf x,\mathsf y)$.  
The previous equality rewrites:
\begin{align}\label{eq.q=}
 \int_0^T \int_{\mathbb R^{2d}} [\mathbb{I}- R_T]f(t,\mathsf y ) \, q_t(\mathsf x,\mathsf y)\, d\mathsf y \, dt=0, \ \, \forall f\in L^{q}_TL^{p}(\mathbb R^{2d}). 
\end{align}
By~\cite[Lemma 2.16]{RocknerZhang}, with $\mathbf p'=(p,p)\in [1,+\infty]^2$, $\beta'=1$, $\beta=0$ and $\mathbf p_0= \mathbf p'$ there, there exists $C_p>0$ (independent of $f$) such that  for all $t>0$:
$$
\Vert \nabla _v  \hat P_tf \Vert_{L^p}   \le C_p (1\wedge t)^{-\frac{1}\alpha}\Vert f \Vert_{ L^{  p}}.
$$ 
Let  $p\in [1,+\infty]$ and  $T\in [0,1]$. By Minkowski’s integral inequality and since $\mathbf B$ is bounded
 $$\|R_T f(t, \cdot)\|_{L^p} \le C_p \|\mathbf{B}\|_\infty   \int_0^{T-t} u^{-\frac{1}\alpha} \|f(t+u, \cdot)\|_{L^p} du.$$ 
We apply Hölder's inequality in time and we have  for $t\in [0,T]$:
$$\|R_T f(t, \cdot)\|_{L^p} \le \underbrace{C_p \|\mathbf{B}\|_\infty    \left( \int_0^{T} u^{-\frac{1}\alpha \frac{q}{q-1}} du \right)^{\frac{q-1}{q}}}_{K_T} \|f\|_{L^q_T L^p}.$$
For $K_T$ to be finite, it is necessary and sufficient that 
$$q > \frac{\alpha}{\alpha - 1} \ \ (\Leftrightarrow  \frac{1}\alpha \frac{q}{q-1}<1).$$
Note that the previous condition is satisfied when     $q>1$  satisfies~\eqref{eq.qp} since in this case:
$$\frac{1}\alpha \frac{q}{q-1}\le \frac{1+\mathscr A_0(p)}{\alpha}\frac{q}{q-1} <1.$$
Then, for some constant $ C(p, q, \|\mathbf{B}\|_\infty)>0$, we have that 
 $$K_T = C(p, q, \|\mathbf{B}\|_\infty) \, T^{\frac{q-1}{q} - \frac{1}{\alpha}}.$$
Consequently, we deduce that:
 $$\left( \int_0^T \|R_T f(t, \cdot)\|_{L^p}^q dt \right)^{1/q} \le \left( \int_0^T \left( K_T \|f\|_{L^q_T L^p} \right)^q dt \right)^{1/q}\le  K_T \, T^{1/q} \, \|f\|_{L^q_T L^p},$$
 i.e. 
 $$\|R_T\|_{L^q_T L^p \to L^q_T L^p} \le C(p, q, \|\mathbf{B}\|_\infty) \, T^{1 - \frac{1}{\alpha}}.$$
 Hence $R_T:  L^{q}_TL^{p}(\mathbb R^{2d}) \to   L^{q}_TL^{p}(\mathbb R^{2d})$ is a bounded linear operator for every $p\in [1,+\infty]$ and $q\in ( \frac{\alpha}{\alpha - 1},+\infty]$. In addition, there exists $T_0\in (0,1]$ such that for all $T\in [0,T_0]$, $\mathbb{I}- R_T$ is invertible with bounded inverse  $(\mathbb{I}- R_T)^{-1}: L^{q}_TL^{p}(\mathbb R^{2d})\to    L^{q}_TL^{p}(\mathbb R^{2d}) $. Consider any  $g\in L^{q}_TL^{p}(\mathbb R^{2d})$ with    $p>1$ and $q>1$   satisfying resp.~\eqref{eq.p0} and~\eqref{eq.qp}. We can then insert 
 $f = (\mathbb{I}- R_T)^{-1} g\in L^{q}_TL^{p}(\mathbb R^{2d})$ in~\eqref{eq.q=} to deduce that  
$$
 \int_0^T \int_{\mathbb R^{2d}} g(t,\mathsf y ) \, q_t(\mathsf x,\mathsf y)\, d\mathsf y \, dt=0, \  \forall  T\in [0,T_0], \, \forall  g\in L^{q}_TL^{p}(\mathbb R^{2d}),
 $$
i.e.  
$$
\mathbb E_{\mathsf x}\Big [\int_0^T g(s,X_s)ds\Big] =\bar{\mathbb E}_{\mathsf x}\Big [\int_0^T g(s,Y_s)ds\Big],\  \forall  T\in [0,T_0], \, \forall  g\in L^{q}_TL^{p}(\mathbb R^{2d}).
$$
 Let us now pick $g\in C_c^\infty(\mathbb R^{2d})$. Noting that both $t\mapsto \mathbb E_{\mathsf x}[g(X_t)]$ and $t\mapsto \bar{\mathbb E}_{\mathsf x}[g(Y_t)]$ are  continuous over $\mathbb R_+^*$ (by~\eqref{eq.Sfd}), we can differentiate the previous relation to get that: 
 \begin{equation}\label{eq.=L}
 \mathbb E_{\mathsf x}[g(X_t)]= \bar{\mathbb E}_{\mathsf x}[g(Y_t)],\  \forall  t\in [0,T_0], \,  g\in C_c^\infty(\mathbb R^{2d}).
 \end{equation}
 Note that, if we are given  $\mu\in \mathcal P(\mathbb R^{2d})$, our previous analysis  implies that~\eqref{eq.=L} also holds if $X_0$ and $Y_0$ have law $\mu$, i.e. in this case, $X_t=Y_t$ in law for all $t\in [0,T_0]$. 
Recall that the laws of $(X_t)_{t\in [0,T_0]}$  and of $(Y_t)_{t\in [0,T_0]}$ are solution to the associated martingale problem over $[0,T_0]$. 
By~\cite[Theorem 4.2 in Chapter 4]{EK}, 
we conclude that weak uniqueness holds on $[0,T_0]$. It is then standard to propagate this weak uniqueness over $\mathbb R_+$ using regular conditional probability distributions. 
By~\cite[Item (a) in Theorem 4.2 in Chapter 4]{EK}, any weak solution to~\eqref{eq.Lan} is  a  Markov process. 
The proof of Theorem~\ref{th.uni} is complete.



\subsection{Weak existence  to~\eqref{eq.Lan}} 
\label{sec.weakEx}
Assume {\rm\textbf{[{\footnotesize H$^{\alpha\in (1,2)}_{\text{bounded}}$}]}}. 
Recall  \eqref{eq.La}.  
Let $\mu \in \mathcal{P}(\mathbb{R}^{2d})$. 
In view of Proposition~\ref{pr.Pre}, establishing the existence of a weak solution to \eqref{eq.Lan} with initial law $\mu$ is equivalent to proving that the martingale problem $({\mathcal L^\alpha}, \mu)$ is solvable. Furthermore, it is now well-known this general existence result follows directly once a solution $({\mathcal L^\alpha}, \mathsf{x})$ is constructed for every deterministic starting point $\mathsf{x} \in \mathbb{R}^{2d}$.

\subsubsection{Approximating  sequence and Krylov's type  estimates} We first start by approximating $\mathbf B$.
\medskip

\noindent
\underline{Approximating  sequence}. 
\medskip

\noindent
Let $(\mathbf B_n)_{n\ge 0}$ be a sequence of   smooth compactly supported vector fields, which is uniformly bounded, and  such that 
\begin{equation}\label{eq.CvBn}
  \mathbf B_n\to  \mathbf B \text{ a.e. over } \mathbb R^{2d}.
  \end{equation}
Consider the  strong solution 
$(X_t^n(\mathsf x)=(x^n_t,v^n_t), t\ge 0)$ to~\eqref{eq.Lan} with initial law $\delta_{\mathsf x}$ when $\mathbf B$ is replaced by $\mathbf B_n$. 
Denote by   
$$\mathcal L^{\alpha}_n=\mathcal S_v^{\nu_\alpha}+ v\cdot \nabla x + \mathbf B_n \cdot \nabla_v $$
   the infinitesimal generator associated with   $(X_t^n, t\ge 0)$.   
   In addition, it is easy to prove  using  the  tightness criteria~\cite[Theorem 6.10]{billingsley2013}  that $(X_s^n(\mathsf x), s\ge 0)$ is (weakly) relatively compact in $ D([0,+\infty), \mathbb R^{2d})$. 
Hence, there exists a subsequence (still denoted by $n$) and a process $X$ such that as $n\to +\infty$:
$$
X^n(\mathsf x)  \overset{\text{law}}{\to}  X   
\quad \text{in } D([0,+\infty),\mathbb{R}^{2d}).
$$  
On the other hand, by  Skorokhod's representation theorem, we may assume (on a new probability space\footnote{With a slight abuse of notation we keep the same notation for  this new probability space and for the  processes defined on this new probability space.}) that as $n\to +\infty$:
\begin{equation}\label{eq.cvXX}
X^{n}(\mathsf x)\to  X 
\quad \text{a.s. in } D([0,+\infty),\mathbb{R}^{2d}).
\end{equation}
The goal is  to prove that $X$ solves the martingale problem $({\mathcal L^\alpha}, \mathsf x)$. Note that in law, we have $X_0^n \to X_0$ when $n\to +\infty$ (see e.g. \cite[Lemma 7.8  in Chapter 3]{EK}).  Hence $X_0=\mathsf x$.  
\medskip

\noindent
\underline{Krylov's type estimates for $X_t^n(\mathsf x)$}. 
\medskip

\noindent
 Because $\alpha \in (1,2)$, the identity~\eqref{eq.DI} holds for the   the process $(X_t^n(\mathsf x),t\ge 0)$. Then, with  the same computations as those made in  \textbf{Step 2} in Section~\ref{sec.Kl}  and since $(\mathbf B_n)_{n\ge 0}$ is uniformly bounded, one deduces the following Krylov's type  estimate. For any      $p>1$  satisfying~\eqref{eq.p0} and $t>0$, there exists $C_{t,p}>0$,  such that for  all $f\in C_c^\infty(\mathbb{R}^{2d})$, all $\mathsf x\in \mathbb R^{2d}$, and all $n\ge 0$,
 \begin{equation}\label{eq.K}
 \Big | \mathbb{E}_{\mathsf x} \Big  [ \int_0^t f\!\left(X_s^{n}\right)\, ds \Big ] \Big |
 \le 
C_{t,p}\, \Vert  f\Vert _{L^{p}}. 
 \end{equation}  Moreover, since $f$ is continuous and bounded, by \eqref{eq.cvXX}, we can pass to the limit $n\to +\infty$ to get that:
  \begin{equation}\label{eq.K2}
 \Big | \mathbb{E}_{\mathsf x} \Big  [ \int_0^t f\!\left(X_s\right)\, ds \Big ] \Big |
 \le 
C_{t,p}\, \Vert  f\Vert _{L^{p}}. 
 \end{equation}
Note that ~\eqref{eq.K}  and ~\eqref{eq.K2}
extend to every $f\in L^{p}(\mathbb{R}^{2d})$,  for   $p>1$  satisfying~\eqref{eq.p0}. 

\subsubsection{Martingale problem} 
\label{sec.MG}
Let $\varphi \in C_c^2(\mathbb{R}^{2d})$. 
For each $n$ and $\mathsf x\in \mathbb R^{2d}$,
$$
M_t^n
=
\varphi  (X_t^n(\mathsf x))
-
\varphi (\mathsf x)
-
\int_0^t
\mathcal L^{\alpha}_n \varphi (X_s^n(\mathsf x) )\, ds
$$
is a (true) martingale.  Thus, for   all $0\le t_1\le \ldots \le t_k\le  s\le t$ and for all bounded and continuous function    $g: \mathcal (\mathbb{R}^{2d})^k\to \mathbb R$,   we have  for all $n\ge 0$ (see~\cite[Section 3 in Chapter 4]{EK}):
\begin{equation}\label{eq.MP}
\mathbb E_{\mathsf x} \Big [   \Big(\varphi  (X_t^n)- \varphi  (X_s^n) -  \int_s^t
\mathcal L^{\alpha}_n \varphi (X_u^n  )\, du  \Big)  g(X^n_{t_1},\ldots, X^n_{t_k})   \Big ]=0,
\end{equation}
that we rewrite 
 \begin{align} 
 \nonumber
&\mathbb E \Big [   \Big(\varphi  (X_t^n(\mathsf x))- \varphi  (X_s^n(\mathsf x)) -  \int_s^t [\mathcal S_v^{\nu_\alpha}+ v\cdot \nabla x] \varphi (X_u^n(\mathsf x) )\,  du  \Big)  g(X^n_{t_1}(\mathsf x),\ldots, X^n_{t_k}(\mathsf x))   \Big ]\\
\label{eq.MP2}
&\quad  - \mathbb E \Big [ g(X^n_{t_1}(\mathsf x),\ldots, X^n_{t_k}(\mathsf x))     \int_s^t [\mathbf B_n\cdot \nabla_v] \varphi (X_u^n(\mathsf x) )\,  du    \Big ]=0.
\end{align}
We now wants to pass to the limit $n\to +\infty$ in~\eqref{eq.MP2}. Since $\mathbf B_n$ is   only assumed to converge a.e. to $\mathbf B$, the only delicate point 
 to deal with in~\eqref{eq.MP2} is the convergence of the term 
 $$ H(\mathbf B_n, X^n(\mathsf x)):= \mathbb E \Big [ g(X^n_{t_1}(\mathsf x),\ldots, X^n_{t_k}(\mathsf x))    \int_s^t [\mathbf B_n\cdot \nabla_v] \varphi (X_s^n(\mathsf x) )\, ds     \Big ].$$
 We aim to prove that 
 \begin{equation}\label{eq.Hc}
 \lim_n H(\mathbf B_n, X^n(\mathsf x))= H(\mathbf B, X(\mathsf x)).
 \end{equation}
 One has for all $m,n\ge 0$, 
 \begin{align}
  \nonumber
 H(\mathbf B_n, X^n(\mathsf x)) -H(\mathbf B, X(\mathsf x)) &= H(\mathbf B_n, X^n(\mathsf x))-H(\mathbf B_m,X^n(\mathsf x))\\
 \nonumber
 &\quad + H(\mathbf B_m,X^n(\mathsf x))- H(\mathbf B_m,X(\mathsf x))\\
 \label{eq.Hph}
 &\quad +   H(\mathbf B_m,X(\mathsf x))-H(\mathbf B, X(\mathsf x)).
   \end{align}
   Let us treat separately all  the three  term on the right-hand side  of~\eqref{eq.Hph}. For ease of notation, set   for $q\ge 0$, $G_{q}(\mathsf x)=g(X^q_{t_1},\ldots, X^q_{t_k})$ and $h_q= [\mathbf B_q\cdot \nabla_ v] \varphi$. Set also  $G(\mathsf x)=g(X_{t_1},\ldots, X_{t_k})$ and $h= [\mathbf B\cdot \nabla_v] \varphi$.
   By \cite[Lemma 7.7 in Chapter 3]{EK},   the complementary of $\mathscr C_X:=\{t> 0, \mathbb P[X_t=X_{t_-}]=1\}$   in $\mathbb R_+$  is at most countable. Note that  $u\in \mathscr C_X$ iff $X$ is  a.s. continuous at $u$.   From \cite[Proposition 5.2  in Chapter 3]{EK},  for all  $q\ge 1$,  it holds  a.s.   when  $ n \to +\infty$,    
 $$
      \text{$X_{u_j}^n\to X_{u_j}$ for all   $u_1,\ldots, u_q  \in   \mathscr C_X$.}
$$
From now on, $0\le t_1\le \ldots \le t_k\le  s\le t$, with  $t_1,\ldots,t_k\in     \mathscr C_X$, so that in particular, we have by dominated convergence theorem:  
\begin{equation}\label{eq.kapp}
\kappa_q:= \mathbb E\big [| G_q(\mathsf x)-G(\mathsf x)|\big ]\to 0 \text{ as } q\to +\infty.
\end{equation} 
   Note that since $\varphi\in C_c^2(\mathbb R^{2d}) $  and by~\eqref{eq.CvBn}, we have  for all $p\ge 1$, 
\begin{equation}\label{eq.Hconv}
h_q,h\in L^p(\mathbb R^{2d}) \text{ and } \lim_{q\to +\infty}h_q=h  \text{ in } L^p(\mathbb R^{2d}). 
\end{equation}
From now on,  $p>1$  is fixed  so that $2p$ satisfies~\eqref{eq.p0}. Hence, the inequalities~\eqref{eq.K} and~\eqref{eq.K2} hold in $L^{2p}(\mathbb R^{2d})$. 
   On the one hand, one has:
   \begin{align*}
  &H(\mathbf B_n, X^n(\mathsf x))-H(\mathbf B_m,X^n(\mathsf x))\\
  &=\mathbb E \Big [    G_n(\mathsf x) \Big( \,  \int_s^t h_n (X_u^n(\mathsf x) )\,  du -\int_s^t h_m (X_u^n(\mathsf x) )\, du\Big)  \Big ]. 
   \end{align*}
    Using~\eqref{eq.K},  we have that:
$$
 \mathbb{E}  \Big [
\int_s^t
\big | \big (h_n-h_m) 
 (X_u^n(\mathsf x) )\big | ^2\, du
\Big  ] 
\le 
\delta_{m,n}^2:=C_{t,p} \Vert   h_n-h_m  \Vert  _{L^{2p}}^2.
$$ 
In addition, using   Cauchy–Schwarz inequality, we deduce that as $n\to +\infty$, 
 \begin{align*}
&\Big | H(\mathbf B_n, X^n(\mathsf x))-H(\mathbf B_m,X^n(\mathsf x)) \Big | \\
&\le \Vert g\Vert_\infty  \sqrt{  \mathbb E  \Big [
\Big|\int_s^t (h_n-h_m) (X_u^n (\mathsf x)) \, du\Big|^2
\Big  ]   }  \le   (t-s)^{1/2}  \, \Vert g\Vert_\infty  \,     \delta_{m,n}    . 
\end{align*}
Let $C_*>0$ such that $\sup_m \Vert h_m\Vert_{\infty}\le C_*$.  Set (see~\eqref{eq.K2})
\begin{equation}\label{eq.dd}
\delta_{q}^2:=C_{t,p} \Vert    h_q-h \Vert  _{L^{2p}}^2\to 0. 
\end{equation}
Then, we have by~\eqref{eq.K2}: 
 \begin{align*}
  &\Big|H(\mathbf B_m,X(\mathsf x))-H(\mathbf B, X(\mathsf x))\Big|\\
  &=\Big|\mathbb E \Big [    G(\mathsf x) \,  \int_s^t h_m (X_u(\mathsf x) )\,  du   \Big ]-\mathbb E \Big [    G(\mathsf x)  \,  \int_s^t h (X_u(\mathsf x) )\, du  \Big ]\Big|\\
  &\le   (t-s)^{1/2}\Vert g\Vert_\infty \, \delta_m. 
   \end{align*}
 Finally, we have that  
  \begin{align*}
  &\Big|H(\mathbf B_m,X^n)- H(\mathbf B_m,X)\Big|\\
  &= \Big|\mathbb E \Big [    G_n(\mathsf x) \,  \int_s^t h_m (X_u^n(\mathsf x) )\,  du   \Big ]- \mathbb E \Big [    G(\mathsf x)  \,  \int_s^t h_m (X_u (\mathsf x) )\,  du   \Big ]\Big|\\
  &\le C_*  (t-s)   \kappa_n + \underbrace{\Big|\mathbb E \Big [    G(\mathsf x)\Big( \,  \int_s^t h_m (X_u^n(\mathsf x) )\,  du  - \int_s^t h_m (X_u (\mathsf x) )\,  du\Big)   \Big ]\Big|}_{=:A_{m,n}(\mathsf x)}.
   \end{align*}
 Since $h_m$ is continuous and bounded, using \eqref{eq.cvXX} and by dominated convergence theorem, it follows that 
 \begin{align}\label{eq.Anm}
\forall m\ge 0, \  \lim_{n\to +\infty} A_{m,n}(\mathsf x)=0.
    \end{align} 
 In conclusion, coming back to~\eqref{eq.Hph},  we have for all $0\le m\le n$, 
 \begin{align*}
  \nonumber
 \big |H(\mathbf B_n, X^n(\mathsf x)) -H(\mathbf B, X(\mathsf x)) \big |&\le   (t-s)^{1/2}  \, \Vert g\Vert_\infty  \,     \delta_{m,n}\\
 &\quad + C_*  (t-s)   \kappa_n  +  A_{m,n}(\mathsf x) \\
 &\quad +   (t-s)^{1/2}\Vert g\Vert_\infty \, \delta_m.
   \end{align*} 
 Let $\epsilon>0$. Using~\eqref{eq.dd}, there exists   $m_\epsilon\ge 0$ such that
 $$  (t-s)^{1/2}\Vert g\Vert_\infty\, \delta_{m_\epsilon}\le \epsilon.$$
 and  (note that by~\eqref{eq.Hconv}, $(h_q)_{q\ge 0}$ is a Cauchy sequence in $ L^{2p}(\mathbb{R}^{2d})$), 
 $$\forall n\ge m_\epsilon, \ (t-s)^{1/2}  \, \Vert g\Vert_\infty  \,     \delta_{m_\epsilon,n} \le \epsilon.$$ 
Since  $\lim_{n\to +\infty} A_{m_{\epsilon},n}(\mathsf x)=0$ (recall indeed~\eqref{eq.Anm}) and using~\eqref{eq.kapp},  there exists $M_\epsilon\ge 0$, for all $n\ge M_\epsilon$,  
 $$A_{m_{\epsilon},n}(\mathsf x)+ C_*  (t-s)   \kappa_n\le \epsilon.$$
 In conclusion, for $n\ge \max(m_\epsilon,M_\epsilon)$, 
 $$ \big |H(\mathbf B_n, X^n) -H(\mathbf B, X) \big |\le 3 \epsilon,$$
 proving~\eqref{eq.Hc}. 
Hence, letting finally $n\to +\infty$ in~\eqref{eq.MP2}, one deduces that for every $0\le t_1\le \ldots \le t_k\le  s\le t$, with  $t_1,\ldots,t_k,s,t\in     \mathscr C_X$:
$$
\mathbb E \Big [ \Big (\varphi  (X_t )- \varphi  (X_s) -  \int_s^t
 {\mathcal L}^\alpha \varphi (X_u  )\, du  \Big)  g(X_{t_1},\ldots, X_{t_k}) \Big ]=0.
$$
By the right continuity of the trajectories of $X$ and since $[0,  +\infty)\setminus \mathscr C_X$ is at most countable, the latter equality extends to any $0\le t_1\le \ldots \le t_k\le  s\le t$. 
   That is $X$ is a solution to the martingale problem with $( {\mathcal L^\alpha}, \delta_{\mathsf x})$.  
This achieves the proof of the existence   part in  Theorem~\ref{th.ex} (see indeed Proposition~\ref{pr.Pre}).

\begin{remark}\label{re.krylov}We emphasize for later purposes that the proof identifying the weak limit $X$ of the sequence $X^n(\mathsf{x})$ as a solution to the martingale problem $(\mathcal{L}^\alpha, \mathsf{x})$ rests on two fundamental components: the uniform (in $n\ge 0$) Krylov estimate \eqref{eq.K} for the sequence $X^n(\mathsf x)$ and the corresponding estimate \eqref{eq.K2} for the limit process $X$. 
\end{remark}

 \subsection{Proof of the weak continuity w.r.t. the initial condition}
Recall that we work under the assumption {\rm\textbf{[{\footnotesize H$^{\alpha\in (1,2)}_{\text{bounded}}$}]}}. 
Recall also that $ (X_t( \mathsf y ),t\ge 0)$ is the weak solution to  ~\eqref{eq.Lan} with initial condition $\mathsf y\in \mathbb R^{2d}$. 
 Let $\mathsf x_n\to \mathsf x\in \mathbb R^{2d}$.  
 Let us assume that as $n\to +\infty$:
 \begin{equation}\label{eq.cvw}
 X(\mathsf x_n)  \to  X (\mathsf x)   \text{ weakly in } \mathcal P(D([0,+\infty), \mathbb{R}^{2d})).
\end{equation} 
Recall that for a process $(Z_t,t\ge 0)$, $\mathscr C_Z=\{t> 0, \mathbb P[Z_t=Z_{t_-}]=1\}$. 
  Note that since  every weak solution $ (X_t,t\ge 0)$ is driven by a Lévy process, we have\footnote{This is actually a consequence of a more general result for solution to matringale problems~\cite[Theorem 3.12 in Chapter 4]{EK}.}:
 $$\mathscr C_X= \mathbb R_+.$$
 Hence, by~\eqref{eq.cvw}  and \cite[Lemma 7.8  in Chapter 3]{EK},  for every $t\ge 0$, it holds:
$$
 X_t(\mathsf x_n)  \to   X_t(\mathsf x)  \text{ weakly in } \mathcal P( \mathbb{R}^{2d}),
$$
 which is precisely the continuity of $\mathsf x\mapsto \mathbb P[X_t(\mathsf x)\in \cdot]$ for the weak convergence in $\mathcal P(\mathbb R^{2d})$.

 We now prove ~\eqref{eq.cvw}.  
 Since $\mathbf B$ is bounded, using   standard criteria (such as e.g.~\cite[Theorem 16.10]{billingsley2013}), it   is rather straightforward  to prove that $(X_s(\mathsf x_n), s\ge 0)$ is (weakly) relatively compact in $ D([0,+\infty), \mathbb R^{2d})$. Thus, there exists an extraction (still denoted by $n$) and a process $Y$ such that as $n\to +\infty$ and in law, 
$
 X(\mathsf x_n)  \to  Y
 $ in $ D([0,+\infty),\mathbb{R}^{2d}).
$  
Recall   the Krylov's estimate~\eqref{eq.K}    for the process $X(\mathsf x_n)$ which, we emphasize,  is  uniform in the initial condition in $\mathbb R^{2d}$. Therefore, as we have already seen, passing to the limit $n\to +\infty$ in~\eqref{eq.K} ,  the process  $Y$ (more precisely its law)  also satisfies~\eqref{eq.K} (with the same constant $C_{t,p}>0$). Moreover, recall also Remark~\ref{re.krylov}. 
Then, using  the same arguments as those used in Section~\ref{sec.MG}, one proves that $(Y_t,t\ge 0)$ is a    solution to the martingale problem $(\mathcal L^\alpha, \delta_{\mathsf x})$. By weak uniqueness we proved in the previous section  when $\mathbf B$ is measurable and bounded (see  also Proposition~\ref{pr.Pre}), $Y= X (\mathsf x)$ in law, and the whole sequence $( X(\mathsf x_n))_{n\ge 0}$ converges weakly to $X (\mathsf x)$.   This completes the proof of~\eqref{eq.cvw}. The proof of Theorem~\ref{th.ex} is complete. 


 \section{Well-posedness~\eqref{eq.Lan}   and  several  properties of~\eqref{eq.Lan} when  {\rm\textbf{[{\footnotesize H$^{\alpha\in (1,2)}_{\text{LG}}$}]}} holds}

\label{sec.Eu2}
 
In this section, we prove the well-posedness of~\eqref{eq.Lan}  and  several properties of~\eqref{eq.Lan}  when  {\rm\textbf{[{\footnotesize H$^{\alpha\in (1,2)}_{\text{LG}}$}]}} holds, namely when $\mathbf B$ has at most linear growth, and when the driving noise $L$ is a RI$\alpha$S process with $\alpha \in (1,2)$ (recall that we denote such a process by $L^\alpha$).

\begin{theorem}\label{th.ex2}
Assume {\rm\textbf{[{\footnotesize H$^{\alpha\in (1,2)}_{\text{LG}}$}]}}.  Then, for every $\mu \in \mathcal P( \mathbb R^{2d})$, there is a unique weak solution   to~\eqref{eq.Lan}  with initial law $\mu$. In addition:
\begin{enumerate}
\item[-]  For $t\ge 0$,  the mapping $\mathsf x\mapsto \mathbb P[X_t(\mathsf x)\in \cdot]$ (resp. $\mathsf x\mapsto \mathbb P[X(\mathsf x)\in \cdot]$) is continuous for the weak convergence in $\mathcal P(\mathbb R^{2d})$ (resp. in $\mathcal P(D([0,+\infty),\mathbb R^{2d}))$).  \item[-] For all $\mathsf x \in \mathbb R^{2d}$ and $t>0$,  $X_t(\mathsf x)$   has a density  $p_t(\mathsf x,\mathsf x')$ w.r.t. the Lebesgue measure $d\mathsf x'$. Moreover,  for every $p> 1$ satisfying~\eqref{eq.p0},  $\mathsf x'\mapsto p_t(\mathsf x,\mathsf x')$ belongs to   $L^{p'}(\mathbb R^{2d})$ (where $p'= p/(p-1)$).
\item[-] Finally, $(P_t,t\ge 0)$  is a  Feller $C_0(\mathbb R^{2d})$-semigroup.  
\end{enumerate}
 \end{theorem}

Notice that Theorem \ref{th.ex2} implies in particular that \textbf{[{\footnotesize C$_{\text{LG}}$}]} is satisfied  when  {\rm\textbf{[{\footnotesize H$^{\alpha\in (1,2)}_{\text{LG}}$}]}} holds. 
\medskip

\noindent
  \textbf{Note}. 
  The unique weak solutions to~\eqref{eq.Lan}  form a strong Markov process, see e.g. \cite[Theorem 6.17]{le2016brownian}. 
   Theorem~\ref{th.ex2} will be extended to the case of drifts $\mathbf B$ with possibly superlinear growth in the position variable  $x$ under the additional assumption that $\mathbf B$ possesses a perturbed gradient field structure, see Theorem~\ref{th.sta}.  
The proof we give of the Feller $C_0(\mathbb R^{2d})$-semigroup property of $(P_t, t \ge 0)$ is independent of the specific noise structure and shows that this property remains  valid under \textbf{[{\footnotesize C$_{\text{LG}}$}]}.

\subsection{Weak well-posedness under {\rm\textbf{[{\footnotesize H$^{\alpha\in (1,2)}_{\text{LG}}$}]}}}

We recall that weak well-posedness for an arbitrary initial distribution is guaranteed if, for every $\mathsf{x} \in \mathbb{R}^{2d}$, there exists a unique weak solution to \eqref{eq.Lan} starting at $X_0 = \mathsf{x}$. Consequently, we may restrict our focus to such deterministic initial conditions.  
  The proof   relies on    the associated stopped martingale problem  (see Definition~\ref{de.SMP}), Theorem~\ref{th.ex} and Grönwall's inequality.

\subsubsection{Stopped martingale problem} 
\label{sec.SMP}
  Define the  vector field   $\mathbf B_R$ by
\begin{equation} \label{eq.BR}
\mathbf B_R(\mathsf y):=\mathbf B(\mathsf y)\mathbf 1_{|\mathsf y|\le R}, \ \mathsf y  \in  \mathbb R^{2d}.
\end{equation} 
Let us denote by $\mathcal L^\alpha_R$ the infinitesimal generator of \eqref{eq.Lan} with vector field $\mathbf B_R$, i.e.  
$$
\mathcal L^\alpha_R= \mathcal S_v^{\nu_\alpha} + v\cdot \nabla_x + \mathbf B_R\cdot \nabla_ v.
$$
        Equation~\eqref{eq.Lan} can be rewritten as follows:
$$
  dX_t= \mathbf b(X_t)dt +  \Sigma 
 dL_t^\alpha  
   \text{ with }  \mathbf b(x',v')=   \begin{pmatrix}
  v'  \\
  \mathbf B(x',v')  
 \end{pmatrix} \text{ and }   \, \Sigma= \begin{pmatrix}
  0    \\
  \mathbb{I}_d 
 \end{pmatrix}.
$$   
For $\mathsf x  \in  \mathbb R^{2d}$   and  $R>|\mathsf x|$, denote by $(X_t^R =(x_t^R,v_t^R),t\ge 0)$ the weak solution (see Theorem~\ref{th.ex}) to   
 \begin{equation}\label{eq.exXtr}
 dX^R_t= \mathbf b_R(X^R_t)dt +  \Sigma 
dL^\alpha_t, \ X^R_0= \mathsf x, \   
  \text{ with }  \mathbf b_R(x',v')=   \begin{pmatrix}
 v'  \\
 \mathbf B_R(x',v')   
\end{pmatrix},
\end{equation}
We denote this weak solution by $(X_t^R(\mathsf x),t\ge 0)$. 
Denote   for  $\mathsf x\in \mathbb R^{2d}$ and  $R>|\mathsf x|$,
$$
\tau_R(\mathsf x)  = \inf \{t \ge 0 : | {X}^R_t(\mathsf x)| \ge R \text{ or } | {X}^R_{t_-}(\mathsf x)| \ge R\}.
$$
For every  $\mathsf x\in \mathbb R^{2d}$ and  $R>|\mathsf x|$, by Itô's formula~\cite[Theorem 4.4.7]{applebaum2009levy},  $X^R_{\cdot \wedge   \tau_R}(\mathsf x)$    solves the stopped martingale problem  $(\mathcal L^\alpha,\mathsf x,\mathscr B_{\mathbb R^{2d}}(0,R))$  (see Definition~\ref{de.SMP}) which is well-posed by   Theorem~\ref{th.ex} and  Proposition~\ref{pr.Pre}, and ~\cite[Theorem 6.1 in Chapter 4]{EK}. Consequently,  $X^R_{\cdot \wedge   \tau_R}(\mathsf x)$ is the unique solution   to the stopped martingale problem  $(\mathcal L^\alpha_R,\mathsf x, \mathscr B_{\mathbb R^{2d}}(0,R))$.  
  Note that for all $t\ge 0$, it holds:
$$\int_0^{t \wedge   \tau_R(\mathsf x)}\mathcal L_R^\alpha \varphi (X^R_s(\mathsf x))ds= \int_0^{t }  \mathbf 1_{s< \tau_R(\mathsf x)}\,\mathcal L_R^\alpha \varphi (X^R_s(\mathsf x))ds.$$ 
Moreover, we have that  $\mathcal L^\alpha_R\varphi = \mathcal L^\alpha\varphi$ on $ \bar{\mathscr B}_{\mathbb R^{2d}}(0,R)$ and $|Y_s^R(\mathsf x)|\le R$ when $s<\tau_R(\mathsf x)$. It then  follows that   a solution to  the stopped martingale problem  $(\mathcal L_R^\alpha,\mathsf x,\mathscr B_{\mathbb R^{2d}}(0,R))$ is a solution to   the stopped martingale problem  $(\mathcal L^\alpha,\mathsf x,\mathscr B_{\mathbb R^{2d}}(0,R))$ and vice versa.  Consequently, for every  $\mathsf x\in \mathbb R^{2d}$ and  $R>|\mathsf x|$,
$X^R_{\cdot \wedge   \tau_R}(\mathsf x)$ is (also) the unique solution to  the stopped martingale problem  $(\mathcal L^\alpha,\mathsf x,\mathscr B_{\mathbb R^{2d}}(0,R))$. 
As a result, by~\cite[Theorem 6.3 in Chapter 4]{EK} and Proposition~\ref{pr.Pre},  weak well-posedness under {\rm\textbf{[{\footnotesize H$^{\alpha\in (1,2)}_{\text{LG}}$}]}} holds if we prove the following result. 
  
  \begin{proposition}\label{pr.rt} Assume {\rm\textbf{[{\footnotesize H$^{\alpha\in (1,2)}_{\text{LG}}$}]}}. 
  For each $t\ge 0$,  $\lim_{R\to +\infty} \mathbb P_{\mathsf x}[\tau_R\le t]=0$.
  \end{proposition}

   \begin{proof}(Proposition~\ref{pr.rt}) 
 First of all, since $|\mathbf B_R(\mathsf y)|\le |\mathbf B(\mathsf y)|\le C(1+|\mathsf y|)$, by Grönwall's inequality~\cite[Section 5 in Appendices]{EK},   we have for all $t\ge 0$, 
\begin{equation}\label{eq.Gg0}
\sup_{s\in [0,t]}|X^R_s(\mathsf x)| \le \big [ |\mathsf x| + Ct + S^\alpha_t \big] e^{Ct}, \quad \text{where  } S^\alpha_t = \sup_{s \in [0,t]} |L^\alpha_s|.
 \end{equation}
On the other hand, we have that 
\begin{equation}
\label{eq.tauRinclu}
\{\tau_R(\mathsf x)\le t\}\subset \Big \{\sup_{s\in [0,t]}|X^R_s(\mathsf x)|\ge R\Big \}.
\end{equation}
 Indeed on the event $\{\tau_R(\mathsf x) <+\infty\}$, by definition of $\tau_R(\mathsf x)$, it holds:
$$|X^R_{\tau_R}(\mathsf x)| \geq R \qquad \text{or} \qquad |X^R_{\tau_R^-}(\mathsf x)| \geq R.$$ 
In conclusion, using \eqref{eq.Gg0} and \eqref{eq.tauRinclu},  it  follows that for all $0<r<R$ and  every $|\mathsf x|\le r$, 
\begin{equation}\label{eq.tauR}
\mathbb P_{\mathsf x}[{\tau}_R\le t]\le \Psi_r (R)= \mathbb P\big [ S^\alpha_t \ge R\, e^{-Ct}  - r + Ct\big ].
 \end{equation}
 Note that $\Psi_r (R)\to 0$ as $R\to +\infty$.        
 This  proves Proposition~\ref{pr.rt}.  
  \end{proof}
   The proof of existence and uniqueness of a weak solution to~\eqref{eq.Lan} when \textbf{[{\footnotesize H$^{\alpha\in (1,2)}_{\text{LG}}$}]} holds is therefore complete.  In particular, the weak solutions  form a Markov process. 

\subsection{Weak  continuity of the trajectories under {\rm\textbf{[{\footnotesize H$^{\alpha\in (1,2)}_{\text{LG}}$}]}}}
 \label{sec.Cx}
In this section, we show that when \textbf{[{\footnotesize H$^{\alpha\in (1,2)}_{\text{LG}}$}]} holds:
\begin{equation}\label{eq.Xc}
\mathsf x\in \mathbb R^{2d}\mapsto   \mathbb P[X (\mathsf x)\in \cdot] \in \mathcal P(D([0,+\infty), \mathbb{R}^{2d})) \text{ is  weakly continuous.}
\end{equation}
In all this section, we assume  \textbf{[{\footnotesize H$^{\alpha\in (1,2)}_{\text{LG}}$}]}. 
For $\mathsf x\in \mathbb R^{2d}$, let  $X(\mathsf x)$ be the unique weak solution to \eqref{eq.Lan} with $\mathsf X_0=\mathsf x$. Set 
  $${\theta}_R(\mathsf x)  = \inf \{t \ge 0 : | {X}_t(\mathsf x)| \ge R \text{ or } | {X}_{t_-}(\mathsf x)| \ge R\}.$$  
Since $|\mathbf B(\mathsf y)|\le C(1+|\mathsf y|)$, we recall that by Grönwall's inequality and \eqref{eq.Lan},   for all $t\ge 0$, 
\begin{equation}\label{eq.Gg}
\sup_{s\in [0,t]}|X_s(\mathsf x)| \le \left[ |\mathsf x| + Ct + S^\alpha_t \right] e^{Ct}.
 \end{equation}
Since $\{{\theta}_R(\mathsf x)\le t\}\subset \{\sup_{s\in [0,t]}|X_s(\mathsf x)|\ge R\}$, it follows 
that for all $0<r<R$ and  every $|\mathsf x|\le r$, 
\begin{equation}\label{eq.tauR2}
\mathbb P_{\mathsf x}[{\theta}_R\le t]\le \Psi_{r} (R).
 \end{equation} 
 Note that \eqref{eq.tauR2} can also be directly obtained from \eqref{eq.tauR} and the fact that 
 $ {\tau}_R(\mathsf x) \overset{\text{law}}{=} {\theta}_R(\mathsf x)$.   Note  indeed that because, by Itô's formula,    $X_{\cdot \wedge   {\theta}_R}(\mathsf x)$   also solves the stopped martingale problem  $(\mathcal L^\alpha,\mathsf x,\mathscr B_{\mathbb R^{2d}}(0,R))$ which has a unique solution, we have that:
\begin{equation}\label{eq.=law}
(X_{t\wedge   {\theta}_R}(\mathsf x),t\ge 0)\overset{\text{law}}{=}(X^R_{t\wedge   \tau_R}(\mathsf x),t\ge 0), \  \text{for all $0<r<R$, $|\mathsf x|\le r$}.
\end{equation}

 \subsubsection{Relative compactness} 
 The goal of this section is to prove the following result. 
 \begin{lemma}\label{le.tight}
Assume  {\rm \textbf{[{\footnotesize H$^{\alpha\in (1,2)}_{\text{LG}}$}]}}. Let $\mathsf x_n\to \mathsf x\in \mathbb R^{2d}$.   Then, for all $T\ge 0$,  
\begin{equation}\label{eq.gg0}
(\mathbb P_{\mathsf x_n}[X_{[0,T]}   \in \cdot ])_{n\ge 0}\text{  is tight in $\mathcal P(D([0,T], \mathbb{R}^{2d})$)}.
\end{equation} 
 \end{lemma}
 \begin{proof}  
We use~\cite[Theorem 13.2]{billingsley2013}.  
On the one hand, from \eqref{eq.Gg}, we immediately  conclude that:
  $$\forall T\ge 0,\, \forall  a\ge 0, \  \ \ \mathbb P \Big [\sup_{t\in [0,T]}|X_t(\mathsf x_n)|\ge a\Big ]\le \Psi_{R_0} (a) \to 0 \text{ as } a\to +\infty.$$
Let $\delta>0$ and $0\le a<b\le T$  such that $b-a\le \delta$. Pick $s,t\in [a,b)$. Then, by  \eqref{eq.Gg}:
  \begin{align*}
   |x_s(\mathsf x_n)-x_t(\mathsf x_n)|\le \delta   \left[ R_0 + CT  + S^\alpha_{T} \right] e^{C{T}} ,
  \end{align*}
  and 
   \begin{align*}
  |v_t(\mathsf x_n)-v_s(\mathsf x_n)|\le \delta   \left[C+ R_0 + CT + S^\alpha_{T+1} \right] e^{CT}  +  |L^\alpha_t-L^\alpha_{s}|.
  \end{align*}
  Hence $|X_s(\mathsf x_n)-X_t(\mathsf x_n)|\le  2\delta Z_T+ 2|L^\alpha_t-L^\alpha_{s}|$, where   $$Z_T:=\left[C+ R_0 + C(T+1) + S^\alpha_{T} \right] e^{C{T}}.$$ 
  Let us recall that for $\gamma\in D([0,T],\mathbb R^{2d})$ and $\delta\in (0,1)$, the modulus of continuity of $\gamma$ over $[0,T]$ is defined by (see~\cite[Eq. (12.7)]{billingsley2013}): 
  $$w_T'(\gamma,\delta)= \inf_{\mathscr P= \{t_i\}}\  \max_{1\le i\le \ell(\mathscr P)}\  \sup_{s,t\in [t_i,t_{i-1})}|\gamma_s-\gamma_t|,$$
  where the infimum is taken over all $\mathscr P= \{t_0=0,t_1,\ldots, t_{\ell (\mathscr P)}=T\}$ with $t_{i-1}<t_i$ and $t_{i}-t_{i-1}\ge \delta$. 
   Note that  
  $w_T'(X(\mathsf x_n),\delta)\le 2\delta Z_T+ 2w_T'(L^\alpha,\delta)$. Hence,  for $\epsilon>0$, we deduce the following upper bound which is independent of $n\ge 0$:
     \begin{align*}
  \mathbb P[  w_T'(X(\mathsf x_n),\delta)\ge \epsilon]&\le   \mathbb P[ 2\delta Z_T\ge  \epsilon /2]+  \mathbb P[ 2w_T'(L^\alpha,\delta)\ge  \epsilon /2].
  \end{align*}
Let us recall that  a.s. $w_T'(L^\alpha,\delta)\to 0$ a.s. when $\delta\to 0^+$ (see~\cite[Lemma 1]{billingsley2013}). As a result,
$$\lim_{\delta\to 0^+}\sup_{n\ge 0}\mathbb P[  w_T'(X(\mathsf x_n),\delta)\ge \epsilon]=0.$$
 This completes the proof of Lemma \ref{le.tight}. 
  \end{proof}

 \subsubsection{Proof of \eqref{eq.Xc}}   In this section we prove the following result which, together with Lemma \ref{le.tight} implies \eqref{eq.Xc} by~\cite[Theorems 16.7 and 13.1]{billingsley2013}. 

 \begin{proposition}\label{pr.C}
Assume  {\rm \textbf{[{\footnotesize H$^{\alpha\in (1,2)}_{\text{LG}}$}]}}. Let $\mathsf x_n\to \mathsf x\in \mathbb R^{2d}$.   Then, for all  $0\le t_1\le \ldots \le t_k \le T$ and every continuous  function $g: \mathcal (\mathbb{R}^{2d})^k\to \mathbb R$, we have: 
\begin{equation}\label{eq.gg}
\lim_{n\to+\infty }\mathbb E[g(X_{t_1}(\mathsf x_n),\ldots, X_{t_k}(\mathsf x_n)) ]= \mathbb E[g(X_{t_1}(\mathsf x),\ldots, X_{t_k}(\mathsf x)) ].
\end{equation}
 \end{proposition} 

 \begin{proof}  
 Assume  {\rm \textbf{[{\footnotesize H$^{\alpha\in (1,2)}_{\text{LG}}$}]}}. Let $R_0>0$ such that $|\mathsf x_n|+|\mathsf x|\le R_0$. 
Recall  also that $X^R(\mathsf y)$ denotes  the weak solution   to    \eqref{eq.exXtr} with $X^R_0=\mathsf y$ when $\mathbf B_R$ is given by \eqref{eq.BR}.   Using  \eqref{eq.tauR}, \eqref{eq.=law} and~\eqref{eq.tauR2}, we have for all $n\ge 0$ and $R>R_0$:
\begin{align*}
\mathbb E[g(X_{t_1}(\mathsf x_n),\ldots, X_{t_k}(\mathsf x_n)) ]&=\mathbb E[\mathbf 1_{T<{\theta}_R(\mathsf x_n)} g(X_{t_1}(\mathsf x_n),\ldots, X_{t_k}(\mathsf x_n)) ] \\
&\quad + O(\mathbb P_{\mathsf x_n}[{\theta}_R(\mathsf x_n)\le T]) \\
&= \mathbb E[\mathbf 1_{T<\tau_R(\mathsf x_n)} g(X^R_{t_1}(\mathsf x_n),\ldots, X^R_{t_k}(\mathsf x_n)) ]\\
&\quad  + O(\mathbb P_{\mathsf x_n}[{\theta}_R(\mathsf x_n)\le T])\\
&= \mathbb E[ g(X^R_{t_1}(\mathsf x_n),\ldots, X^R_{t_k}(\mathsf x_n)) ]\\
&\quad + O(\mathbb P_{\mathsf x_n}[\tau_R(\mathsf x_n)\le T])  + O(\mathbb P_{\mathsf x_n}[{\theta}_R(\mathsf x_n)\le T])  \\
&=  \mathbb E[ g(X^R_{t_1}(\mathsf x_n),\ldots, X^R_{t_k}(\mathsf x_n)) ] + r^T_n(R),
 \end{align*}
 where $\sup_n r^T_n(R)\to 0$ as $R\to +\infty$. On the other hand, for any fixed $R>R_0$, by Theorem~\ref{th.ex} (recall that $\mathbf B_R$ is bounded and   $\{t\ge 0, \mathbb P_{\mathsf x}[X^R_t=X^R_{t_-}]=1\}=\mathbb R_+$), 
\begin{equation}\label{eq.Mkp}
\lim_n \mathbb E[ g(X^R_{t_1}(\mathsf x_n),\ldots, X^R_{t_k}(\mathsf x_n)) ]=  \mathbb E[ g(X^R_{t_1}(\mathsf x),\ldots, X^R_{t_k}(\mathsf x)) ].
\end{equation}
 Now we write thanks to \eqref{eq.tauR2} and \eqref{eq.=law},
 \begin{align*}
 \mathbb E[ g(X^R_{t_1}(\mathsf x),\ldots, X^R_{t_k}(\mathsf x)) ]&= \mathbb E[ \mathbf 1_{T<\tau_R(\mathsf x)}  g(X^R_{t_1}(\mathsf x),\ldots, X^R_{t_k}(\mathsf x)) ] \\
 &\quad +  O(\mathbb P_{\mathsf x}[\tau_R(\mathsf x )\le T])\\
 &=  \mathbb E[   g(X_{t_1}(\mathsf x),\ldots, X_{t_k}(\mathsf x)) ]\\
 &\quad +  O(\mathbb P_{\mathsf x}[\tau_R(\mathsf x )\le T])+ O(\mathbb P_{\mathsf x}[{\theta}_R(\mathsf x )\le T])\\
 &=  \mathbb E[   g(X_{t_1}(\mathsf x),\ldots, X_{t_k}(\mathsf x)) ] + o_1(R).
 \end{align*} 
We finally  write: 
 \begin{align*}
& \big |\mathbb E[g(X_{t_1}(\mathsf x_n),\ldots, X_{t_k}(\mathsf x_n)) ]- \mathbb E[   g(X_{t_1}(\mathsf x),\ldots, X_{t_k}(\mathsf x)) ] \big |\\
&\le  \big |\mathbb E[g(X_{t_1}(\mathsf x_n),\ldots, X_{t_k}(\mathsf x_n)) ]-\mathbb E[ g(X^R_{t_1}(\mathsf x_n),\ldots, X^R_{t_k}(\mathsf x_n)) ] \big |\\
&\quad +\big | \mathbb E[ g(X^R_{t_1}(\mathsf x_n),\ldots, X^R_{t_k}(\mathsf x_n)) ]-\mathbb E[ g(X^R_{t_1}(\mathsf x),\ldots, X^R_{t_k}(\mathsf x)) ] \big |\\
&\quad +\big  |\mathbb E[ g(X^R_{t_1}(\mathsf x),\ldots, X^R_{t_k}(\mathsf x)) ]- 
 \mathbb E[   g(X_{t_1}(\mathsf x),\ldots, X_{t_k}(\mathsf x)) ]\big |.
  \end{align*}
 We now conclude the proof of \eqref{eq.gg}. Let $\epsilon>0$. Pick  $R_\epsilon\ge R_0+1$ such that for all $n\ge 0$, 
   \begin{align*}
& \big |\mathbb E[g(X_{t_1}(\mathsf x_n),\ldots, X_{t_k}(\mathsf x_n)) ]- \mathbb E[   g(X_{t_1}(\mathsf x),\ldots, X_{t_k}(\mathsf x)) ] \big |\\
 &\le  \epsilon +\big | \mathbb E[ g(X^{R_\epsilon}_{t_1}(\mathsf x_n),\ldots, X^{R_\epsilon}_{t_k}(\mathsf x_n)) ]-\mathbb E[ g(X^{R_\epsilon}_{t_1}(\mathsf x),\ldots, X^{R_\epsilon}_{t_k}(\mathsf x)) ]\big  |.
  \end{align*}
In addition, by \eqref{eq.Mkp}, there exits $N_\epsilon\ge 0$, for all $n\ge N_\epsilon$, 
  $$\big |\mathbb E[g(X_{t_1}(\mathsf x_n),\ldots, X_{t_k}(\mathsf x_n)) ]- \mathbb E[   g(X_{t_1}(\mathsf x),\ldots, X_{t_k}(\mathsf x)) ] \big |\le 2\epsilon.$$
   This completes the proof of Proposition~\ref{pr.C}. 
    \end{proof}
  The proof of    \eqref{eq.Xc} is complete. 
  
   \subsubsection{Existence of a density} 
   \label{sec.density}
In this step we prove the claims on the density of $X_t(\mathsf x)$  stated in Theorem \ref{th.ex2}.  Let $t>0$. All the claims will follow if we prove that   for every $t>0$ and for every $p> 1$ satisfying~\eqref{eq.p0}, there exists $C_{p,t}>0$ such that for all $\mathsf x\in \mathbb R^{2d}$  and  $f\in  C_c^\infty(\mathbb R^{2d})$: 
\begin{equation}\label{eq.Dua-bis}
|  \mathbb E_{\mathsf x} [ f(X_{t})]|  \le  C_{p,t}\big [1+|\mathsf x|\big ] \  \Vert f \Vert_{ L^{  p}}. 
\end{equation}
  To prove \eqref{eq.Dua-bis} we need to explain how to extend \eqref{eq.Sfd} to the case when $\mathbf B$ has at most linear growth. In what follows,  $C>0$ is a constant independent of $0\le s\le t$ and $\mathsf x\in \mathbb R^{2d}$ which can change from one occurence to another. 
  Note  first that, when  $\mathbf B$ has at most linear growth, the equality \eqref{eq.RS} is still valid  and that we just has to justify that one can pass to the limit $\epsilon\to 0^+$  in this equality. To do so, they are two terms to handle 
  $$\mathbb E_{\mathsf x}\big [\hat P_{\epsilon}f(X_{t-\epsilon})\big ] \ 
 \text{ and } \ 
\mathbb E_{\mathsf x}\Big[
\int_0^{t-\epsilon}
\mathbf B(X_s)\cdot
\nabla_v\hat P_{t-s}f(X_s)\,ds
\Big].
$$
We first consider the second term above. 
For $s\in [0,t)$,  we observe that  
   $$ \big|    \mathbb E_{\mathsf x} [\mathbf B(X_s)\cdot    \nabla _v  \hat P_{t-s}f (X_s)]  \big|\le   C\,   \mathbb E_{\mathsf x} \big [1+| (X_s)| \big]  \,  \Vert     \nabla _v  \hat P_{t-s}f \Vert_\infty. 
 $$
By \eqref{eq.Gg} and Doob’s inequality in $L^m$  (for some fixed $m\in (1,\alpha)$),  it holds:
\begin{align*}
\mathbb E_{\mathsf x} \big [ \sup_{s\in [0,t]}| (X_s)| \big] &\le  \Big ( |\mathsf x| + Ct + \mathbb E\big [\sup_{s \in [0,t]} |L^\alpha_s|\big ] \Big) e^{Ct}  \le   \Big ( |\mathsf x| + Ct + \mathbb E\big [ |L^\alpha_t|^m\big ]^{1/m} \Big) e^{Ct} . \end{align*}
Note that $\mathbb{E}[|L^\alpha_t|^m]^{1/m}   = t^{1/\alpha} \mathbb{E}[|L_1^\alpha|^m]^{1/m}$. 
As a consequence, reasoning as just below \eqref{eq.RS} and using~\eqref{eq.LpB}, we thus get that: 
 $$
\lim_{\epsilon\to 0} \mathbb E_{\mathsf x}\Big[
\int_0^{t-\epsilon}
\mathbf B(X_s)\cdot
\nabla_v\hat P_{t-s}f(X_s)\,ds
\Big]= \mathbb E_{\mathsf x}\Big[
\int_0^{t}
\mathbf B(X_s)\cdot
\nabla_v\hat P_{t-s}f(X_s)\,ds
\Big]
$$
We now consider the term $\mathbb E_{\mathsf x}\big [\hat P_{\epsilon}f(X_{t-\epsilon})\big ]$ and claim again that $\mathbb E_{\mathsf x}\big [\hat P_{\epsilon}f(X_{t-\epsilon})\big ]\to  \mathbb E_{\mathsf x} [ f(X_{t})]$. 
  To prove such a result we use the decomposition given in \eqref{eq.deco}-\eqref{eq.deco2}. The term in \eqref{eq.deco2} is treated with the same arguments as we did previously.  To deal with the second term \eqref{eq.deco2} we recall that one uses \eqref{eq.Cab} and shows that for $\beta\in (0,1)$, 
  $$\lim_{\epsilon\to 0}\mathbb E_{\mathsf x}[|x_{t}-x_{t-\epsilon}|^{\frac{\beta}{1+\alpha}}] + \mathbb E_{\mathsf x}[|v_{t}-v_{t-\epsilon}|^{\beta}  ]= 0.$$
We show the previous limit using  the Grönwall's inequality. We then finally  deduce 
that $\mathbb E_{\mathsf x}\big [\hat P_{\epsilon}f(X_{t-\epsilon})\big ]\to  \mathbb E_{\mathsf x} [ f(X_{t})]$.  This completes the proof of  \eqref{eq.Dua-bis}.

\subsubsection{$C_0$-Feller semigroup} 
   \label{sec.Feller}
   The Feller  $C_0(\mathbb R^{2d})$-semigroup property is well known for stochastic differential equations driven by Brownian motion with bounded and/or Lipschitz continuous drift coefficients. In contrast, in the present framework the drift is neither continuous nor bounded, and the driving noise is purely discontinuous, possessing finite moments only up to order strictly less than $\alpha$.  
   
   Let us prove that under {\rm\textbf{[{\footnotesize H$^{\alpha\in (1,2)}_{\text{LG}}$}]}}, $(P_t,t\ge 0)$ is a $C_0(\mathbb R^{2d})$-semigroup. 
   Note that because $X_t(\mathsf y)\to X_t(\mathsf x)$ in law when $\mathsf y\to \mathsf x\in \mathbb R^{2d}$, we have $P_tC_b(\mathbb R^{2d})\subset C_b(\mathbb R^{2d})$ for all $t\ge 0$. 
It is thus enough to prove that for all $f\in C_0(\mathbb R^{2d})$:
\begin{enumerate}
\item[-]   $\lim_{|\mathsf x|\to \infty} \mathbb E_{\mathsf x}[f(X_t)]= 0$.
\item[-]   $\lim_{t\to 0^+}\Vert P_tf- f\Vert_\infty=0$. 
\end{enumerate}
Let us prove that for every $t>0$ sufficiently small,  $\lim_{|\mathsf x|\to +\infty}\mathbb E_{\mathsf x}[f(X_t)]=0$.  Since $f\in C_0(\mathbb R^{2d})$, it suffices to show that:
$$\lim_{|\mathsf x|\to +\infty}\mathbb P_{\mathsf x}[|X_t|\le R]=0, \ \forall R>0.$$
 Note that  $|X_t(\mathsf x)|\le R$ implies $|\mathsf x|\le R + Ct[1+ \sup_{s\in [0,t]}|X_s(\mathsf x)|]+   S_t^\alpha$ and then by \eqref{eq.Gg0} (which also holds for $X(\mathsf x)$), it thus implies that  $ |\mathsf x|\le R + Ct[1+     [ |\mathsf x| + Ct + S^\alpha_t  ] e^{Ct} ]+   S_t^\alpha$. As a result, when $|X_t(\mathsf x)|\le R$, it holds that:
$$|\mathsf x| ( 1- Ct e^{Ct})\le R + Ct\big[1+   Ct \big]+   [1+ Cte^{Ct} ] S_t^\alpha.$$
Hence, for all $t\in [0,t_c]$ (where $t_c>0$ is such that $Ct_c e^{Ct_c}<1$), $\lim_{|\mathsf x|\to +\infty}\mathbb P_{\mathsf x}[|X_t|\le R]=0$. Therefore, $P_tC_0(\mathbb R^{2d})\subset C_0(\mathbb R^{2d})$, for all $t\in [0,t_c]$. We extend it to every $t>t_0$ thanks to the semigroup property. 

Let us now consider the second item. 
For a fixed $\epsilon > 0$, we consider   $\Psi^\epsilon \in C_c^\infty(\mathbb{R}^{2d})$  such that $\Psi^\epsilon(\mathsf{y}) = 1$ when $|\mathsf{y} | \le \epsilon/2$, $\Psi^\epsilon(\mathsf{y}) = 0$ when $|\mathsf{y}| \ge \epsilon$, and $0 \le \Psi^\epsilon \le 1$ everywhere.
We then   consider the function $\mathsf z\mapsto \Phi^\epsilon_{\mathsf x}(\mathsf z):= \Psi^\epsilon(\mathsf z-\mathsf x)$. Note that $\nabla \Phi^\epsilon_{\mathsf x}$ and all its derivatives are bounded over $\mathbb R^{2d}$, uniformly w.r.t. $\mathsf{x}\in \mathbb R^{2d}$. 
Moreover, $ \Phi^\epsilon_{\mathsf x} (\mathsf z) \le \mathbf{1}_{ |\mathsf z - \mathsf{x}| \le \epsilon}$, and consequently we have:
$$\mathbb{P}_{\mathsf{x}}[|X_t - \mathsf{x}| > \epsilon] \le 1 - \mathbb{E}_{\mathsf{x}}[\Phi^\epsilon_{\mathsf x}(X_t)] = |\Phi^\epsilon_{\mathsf x}(\mathsf{x}) - P_t \Phi^\epsilon_{\mathsf x}(\mathsf{x})|.$$
On the other hand, by Itô's formula and Fubini's Theorem, 
$$P_t \Phi^\epsilon_{\mathsf x}(\mathsf{x}) - \Phi^\epsilon_{\mathsf x}(\mathsf{x}) = \int_0^t \mathbb{E}_{\mathsf{x}}[\mathcal L^\alpha \Phi^\epsilon_{\mathsf x}(X_s)] ds.$$
In addition, we recall that  for $\mathsf z=(x',v')\in \mathbb R^{2d}$, 
\begin{align*}
\mathcal{L}^\alpha \Phi^\epsilon_{\mathsf{x}}(\mathsf z) &= v' \cdot \nabla_{x'} \Phi^\epsilon_{\mathsf{x}}(\mathsf z) + \mathbf{B}(\mathsf z) \cdot \nabla_{v'} \Phi^\epsilon_{\mathsf{x}}(\mathsf z)\\
&\quad + \int_{\mathbb{R}^d} \left[ \Phi^\epsilon_{\mathsf{x}}(x', v'+w) - \Phi^\epsilon_{\mathsf{x}}(x', v') - \nabla_{v'} \Phi^\epsilon_{\mathsf{x}}(x', v') \cdot w \mathbf{1}_{|w|\le 1} \right] \nu_\alpha(dw).
\end{align*} 
 By definition, $ \Phi^\epsilon_{\mathsf{x}}$ is supported in  $\mathscr{B}_{\mathbb{R}^{2d}}(\mathsf{x}, \epsilon)$. 
 Hence, if $|\mathsf{x}| \le R$, it holds that 
 $$\sup_{\mathsf{z} \in \mathbb{R}^{2d}}\Big( |v' \cdot \nabla_{x'} \Phi^\epsilon_{\mathsf{x}}(\mathsf{z})|  + |\mathbf{B}(\mathsf{z}) \cdot \nabla_{v'} \Phi^\epsilon_{\mathsf{x}}(\mathsf{z})|\Big)\le \Big[(R + \epsilon) + \sup_{|\mathsf{z}| \le R + \epsilon} |\mathbf{B}(\mathsf{z})| \Big] \|\nabla \Psi^\epsilon\|_\infty.$$
 In addition, for all $\mathsf x\in \mathbb R^{2d}$, the absolute value of the jump part in the expression of $\mathcal{L}^\alpha \Phi^\epsilon_{\mathsf{x}}(x',v') $  is bounded by 
 $$   \frac{1}{2} \|\nabla^2 \Psi^\epsilon\|_\infty \int_{|w| \le 1} |w|^2 \nu_\alpha(dw) + 2 \int_{|w|>1}  \nu_\alpha(dw).$$
Hence,  for all $R>0$ and $\epsilon>0$,  
$$ M_R^\epsilon:= \sup_{|\mathsf x|\le R}\Vert\mathcal L^\alpha \Phi^\epsilon_{\mathsf x}\Vert_\infty <+\infty.$$
This implies that for all $R>0$ and $\epsilon>0$, 
\begin{equation}\label{eq.Rrp}
\sup_{|\mathsf x|\le R} \mathbb{P}_{\mathsf{x}}[|X_t - \mathsf{x}| > \epsilon] \le \sup_{|\mathsf x|\le R}  |\Phi^\epsilon_{\mathsf x}(\mathsf{x}) - P_t \Phi^\epsilon_{\mathsf x}(\mathsf{x})| \le t M_R^\epsilon.
\end{equation}
Note that we can also deduce a similar estimate to \eqref{eq.Rrp} using  Grönwall inequality. 
 It thus follows that for every $f\in C_0(\mathbb R^{2d})$,  $R>0$,  and $\epsilon>0$, 
 \begin{align*}
\sup_{|\mathsf x| \le R} |P_tf(\mathsf x)- f(\mathsf x)|  &\le  \sup_{|\mathsf y-\mathsf x|\le \epsilon} |f(\mathsf y)-f(\mathsf x)|+ 2\Vert f\Vert_\infty\sup_{|\mathsf x| \le R} \mathbb P_{\mathsf x} [ |X_t - \mathsf{x}| > \epsilon]\\
&\le \sup_{|\mathsf y-\mathsf x|\le \epsilon} |f(\mathsf y)-f(\mathsf x)|+ 2\Vert f\Vert_\infty \, t M_R^\epsilon. 
\end{align*}  
We now aim to estimate $\sup_{|\mathsf x|  > R} |P_tf(\mathsf x)- f(\mathsf x)|$ for $f\in C_0(\mathbb R^{2d})$.  Let  $R\ge R_0>0$. Writing  
$$|P_t f(\mathsf{x})| \le \mathbb{E}_{\mathsf{x}} \big [|f(X_t)| \mathbf{1}_{|X_t| > R_0 } \big ] + \mathbb{E}_{\mathsf{x}} \big [|f(X_t)| \mathbf{1}_{ |X_t| \le R_0 } \big ],$$
 we deduce that:
 \begin{align*}
\sup_{|\mathsf x| > R} |P_tf(\mathsf x)- f(\mathsf x)| &\le \Vert f\Vert_\infty \sup_{|\mathsf x| > R} \mathbb P_{\mathsf x}[ |X_t| \le R_0] +\sup_{|\mathsf x| > R_0}|f(\mathsf x)|+     \sup_{|\mathsf x| > R}|f(\mathsf x)|\\
&\le    \Vert f\Vert_\infty \sup_{|\mathsf x| > R} \mathbb P_{\mathsf x}[ |X_t| \le R_0] + 2\sup_{|\mathsf x| > R_0}|f(\mathsf x)|.
\end{align*}  
Set $a_c:=1-Ct_c e^{Ct_c}>0$. 
For $t\in [0,t_c]$, we have   (recall \eqref{eq.Gg0}):
 \begin{align*}
 |X_t(\mathsf x)| \le R_0 &\Rightarrow  R_0\ge |\mathsf x|  -Ct- Ct\sup_{s\in [0,t]}|X_s(\mathsf x)| -S_t^\alpha\\
&\Rightarrow  R_0\ge|\mathsf x|  -Ct_c- Ct_c[ |\mathsf x| + Ct_c + S^\alpha_{t}  ] e^{Ct_c}  -S_{t}^\alpha\\
&\Rightarrow  S^\alpha_{t}[1+Ct_c e^{Ct_c}]\ge a_c|\mathsf x| -R_0-Ct_c[1+Ct_c].
\end{align*}
Hence, for $t\in [0,t_c]$ and $|\mathsf x|>R$, 
 \begin{align*}
 |X_t(\mathsf x)| \le R_0 &\Rightarrow     S^\alpha_{t}[1+Ct_c e^{Ct_c}]\ge a_c R -R_0-Ct_c[1+Ct_c].  
\end{align*}
We then deduce that for $t\in [0,t_c]$ and $R\ge R_0>0$, 
\begin{align*}     
&\sup_{|\mathsf x| > R} |P_tf(\mathsf x)- f(\mathsf x)|  \\
&\le    \Vert f\Vert_\infty   \mathbb P \Big [ S^\alpha_{t}[1+Ct_c e^{Ct_c}]\ge a_c R -R_0-Ct_c[1+Ct_c]\Big ] + 2\sup_{|\mathsf x| > R_0}|f(\mathsf x)|.
\end{align*}  
We are now in position to conclude the proof of the second item. Let $\eta>0$ and consider $\epsilon(\eta)>0$ and $R_0(\eta)>0$ such that 
$$ \sup_{|\mathsf y-\mathsf x|\le \epsilon(\eta)} |f(\mathsf y)-f(\mathsf x)| \le \eta \text{ and } 2\sup_{|\mathsf x| > R_0(\eta)}|f(\mathsf x)|\le \eta.$$
Now take $R(\eta) = a_c^{-1}( R_0(\eta)+ Ct_c[1+Ct_c]) +1$ so that (note that  $R(\eta)\ge R_0(\eta)$ since $0<a_c\le 1$), for all $t\in [0,t_c]$,
\begin{align*}
 &\sup_{\mathsf x\in \mathbb R^{2d}} |P_tf(\mathsf x)- f(\mathsf x)|\\
 &\le  \sup_{|\mathsf x| \le R(\eta)} |P_tf(\mathsf x)- f(\mathsf x)|+  \sup_{|\mathsf x| > R(\eta)} |P_tf(\mathsf x)- f(\mathsf x)|\\
 &\le  \eta + 2\Vert f\Vert_\infty \, t M_{R(\eta)}^{\epsilon(\eta)} + \Vert f\Vert_\infty   \mathbb P \Big [ S^\alpha_{t}[1+Ct_c e^{Ct_c}]\ge a_c R(\eta) -R_0-Ct_c[1+Ct_c]\Big ] + 2 \eta \\
 &\le 3 \eta + 2\Vert f\Vert_\infty \, t M_{R(\eta)}^{\epsilon(\eta)} + \Vert f\Vert_\infty   \mathbb P \Big [ S^\alpha_{t}[1+Ct_c e^{Ct_c}]\ge1  \Big ]. 
\end{align*} 
Note that $S_t^\alpha\to 0$ when $t\to 0^+$. 
Thus,  there is $0<t(\eta)\le t_c$ such that for all $t\in [0,t(\eta)]$, 
$$\sup_{\mathsf x\in \mathbb R^{2d}} |P_tf(\mathsf x)- f(\mathsf x)|\le 5\eta.$$
In conclusion, $\sup_{\mathsf x\in \mathbb R^{2d}} |P_tf(\mathsf x)- f(\mathsf x)|\to 0$ as $t\to 0$. 
This ends the proof of the Feller continuity property of $(P_t,t\ge 0)$. 
  The proof of   Theorem~\ref{th.ex2} is complete. 

\section{Strong Feller property and spectral radius}

\subsection{Strong Feller property of  the non-killed semigroups of \eqref{eq.Lan} driven by  RI$\alpha$S processes} 
The goal of this section is to prove the following result.

\begin{theorem}\label{th.FF} 
 Assume  that {\rm\textbf{[{\footnotesize H$^{\alpha\in (0,2)}_{\text{smooth}}$}]}} or {\rm\textbf{[{\footnotesize H$^{\alpha\in (1,2)}_{\text{LG}}$}]}}. 
Then,  the semigroup $(P_t,t> 0)$ of~\eqref{eq.Lan}  has the strong Feller  property. 
 \end{theorem}
 \medskip
 
 \noindent
 \textbf{Note}. Under {\rm\textbf{[{\footnotesize H$^{\alpha\in (0,2)}_{\text{smooth}}$}]}},  the assertion of Theorem~\ref{th.FF} has already been proved for every $\alpha \in (0,2)$ by means of Malliavin calculus on   subordinated Brownian motions  in~\cite{dong2016strong}.  Indeed, when $\mathbf B$ is smooth with bounded partial derivatives of order $1$,~\cite[Eqs. (1.4) and (1.5)]{dong2016strong} are satisfied with  the Lyapunov function $H(\mathsf x)= 1+|\mathsf x|^2$.  We also refer to the series of works~\cite{zhang2014sima,zhang2014densities,zhang2017fundamental,zhangAIHP2020} and the references therein for a comprehensive review concerning  the smoothness of the distributional densities of SDEs with jumps. 
    The goal of this section  is to extend such a result  when less regularity are assumed on $\mathbf B$ (see also Theorem~\ref{th.sta}) and we recall that the perturbative approach we employed before necessitates restricting the range of $\alpha$ to $(1, 2)$.



\subsubsection{Proof of Theorem~\ref{th.FF} under {\rm\textbf{[{\footnotesize H$^{\alpha\in (1,2)}_{\text{bounded}}$}]}}}

Recall  Theorems \ref{th.ex} and  \ref{th.uni}. 
In this section,  we assume {\rm\textbf{[{\footnotesize H$^{\alpha\in (1,2)}_{\text{bounded}}$}]}}, i.e. we consider  when $\mathbf B$ is bounded and measurable,  the weak solution $(X_t=(x_t,v_t), t\ge 0)$ of ~\eqref{eq.Lan} when  $(L^\alpha_t, t\ge 0)$ is a RI$\alpha$S process, $\alpha\in (1,2)$.  
To prove Theorem~\ref{th.FF} under {\rm\textbf{[{\footnotesize H$^{\alpha\in (1,2)}_{\text{bounded}}$}]}}, we will need the following result, see also~\cite[Lemma 3.2]{wu1999} for a very similar result.

 \begin{lemma}(\cite[Theorem 2.1 (b)]{schilling2012strong}) 
\label{pr.Schilling}
A  semigroup $(Q_t,t> 0)$ has the strong Feller  property  over $\mathbb R^{2d}$     if and only if the following two conditions hold:
\begin{enumerate}
\item[-] For all $t\ge 0$, $Q_t(C_b(\mathbb R^{2d}))\subset C_b(\mathbb R^{2d})$ (i.e. $Q_t$ is    Feller). 
 \vspace{0.1cm}
 
\item[-] For every $t>0$, there exists some positive
Radon measure $\mu_t$ on $\mathbb R^{2d}$ such that $(Q_t(\mathsf x,d\mathsf y), \mathsf x\in \mathbb R^{2d})$ is locally uniformly absolutely continuous with respect to $\mu_t$, i.e. for every $K\Subset \mathbb R^{2d}$:
$$\lim_{\delta\to 0} \ \sup _{\mathscr A\in  B(\mathbb R^{2d}), \, \mu_t(\mathscr A)\le \delta} \ \sup_{\mathsf x\in K} Q_t(\mathsf x, \mathscr A)=0.$$
\end{enumerate}
\end{lemma}

 \begin{proof}(Theorem~\ref{th.FF} under {\rm\textbf{[{\footnotesize H$^{\alpha\in (1,2)}_{\text{bounded}}$}]}}).   
 Let us  prove that 
\begin{equation}\label{eq.LS}
\lim_{\delta\to 0} \ \sup _{\mathscr A\in  B(\mathbb R^{2d}), \, |\mathscr A|\le \delta} \ \sup_{\mathsf x\in K} P_t\mathbf 1_{\mathscr A}(\mathsf x)=0,
\end{equation}
 where  $|\mathscr A|=\int_{\mathscr A}d\mathsf x'$ is the Lebesgue measure of $\mathscr A$. 
 Let $\delta >0$ and $\mathscr A$ be a measurable set such that $|\mathscr A|\le \delta $.  Note that $\mathbf 1_{\mathscr A}\in  L^{p}( \mathbb R^{2d})$, for every $p\ge 1$.  As a result, using~\eqref{eq.Goal}, we get for every $t>0$ and $\mathsf x\in \mathbb R^{2d}$:
 $$P_t\mathbf 1_{\mathscr A}(\mathsf x) \le C \delta^{\frac1{p}},$$
 for some $C>0$ independent of $\mathsf x\in \mathbb R^{2d}$ and $ \mathscr A\in  B(\mathbb R^{2d})$ with $|\mathscr A|\le \delta $. 
 This implies~\eqref{eq.LS}.  In order to use Lemma \ref{pr.Schilling}, it remains to show that  for $t\ge 0$, $P_t$ is Feller. This a direct consequence of the second assertion in Theorem~\ref{th.ex}. The proof of Theorem~\ref{th.FF} when {\rm\textbf{[{\footnotesize H$^{\alpha\in (1,2)}_{\text{bounded}}$}]}} holds is complete. 
%
%
%
 \end{proof}

\subsubsection{Proof of Theorem~\ref{th.FF} under {\rm\textbf{[{\footnotesize H$^{\alpha\in (1,2)}_{\text{LG}}$}]}}}
\label{sec.SFH2}
We consider the process~\eqref{eq.Lan} driven by a RI$\alpha$S process under   {\rm\textbf{[{\footnotesize B$_{\text{LG}}$}]}} and with $\alpha\in (1,2)$. The goal is to prove   that the mapping 
$$\mathsf x\in \mathbb R^{2d}\mapsto \mathbb E_{\mathsf x}[f(X_t)]$$
is continuous for every $t>0$ and $f\in bB(\mathbb R^{2d})$. 
One way to establish this is by adapting the arguments from the previous section together with \eqref{eq.Dua-bis}. We present below an alternative approach that is more flexible.
  Recall the definition of $(X_t^R,t\ge 0)$ in \eqref{eq.exXtr}.   
Recall that  from~\eqref{eq.=law}, it holds for $|\mathsf y|<R$, $ \mathbb E_{\mathsf y}[f(X_t)\mathbf 1_{t<{\theta}_R}]=  \mathbb E_{\mathsf y}[f(X^R_t)\mathbf 1_{t<\tau_R}]$. 
Let $\mathsf x_n\to \mathsf x\in \mathbb R^{2d}$. Consider  $R_0>0$  such that  $|\mathsf x_n|+|\mathsf x|\le R_0$ for all $n\ge 1$. Let $R>R_0$. We have:
\begin{align*}
\mathbb E_{\mathsf x}[f(X_t)]-\mathbb E_{\mathsf x_n}[f(X_t)]&= \mathbb E_{\mathsf x}[f(X^R_t)\mathbf 1_{ t<\tau_R} ]-\mathbb E_{\mathsf x_n}[f(X^R_t)\mathbf 1_{ t<\tau_R} ] + \Vert f\Vert_\infty \, o_1(R)\\
&= \mathbb E_{\mathsf x}[f(X^R_t) ]-\mathbb E_{\mathsf x_n}[f(X^R_t)  ] + 2\Vert f\Vert_\infty \, o_1(R),
\end{align*}
where $o_1(R)$ is uniform in $n\ge 0$ (by~\eqref{eq.tauR} and~\eqref{eq.tauR2}). 
  By the strong Feller property proved just before for bounded and measurable drifts, and since  $\mathbf B_R$ is measurable and bounded, for all $R>0$, $\mathsf x\in \mathbb R^{2d}\mapsto \mathbb E_{\mathsf x}[f(X_t^R)]$ is continuous for every $f\in bB(\mathbb R^{2d})$. Hence, we have that  $|\mathbb E_{\mathsf x}[f(X_t)]-\mathbb E_{\mathsf x_n}[f(X_t)]|\to 0$ as $n\to +\infty$, the desired result.
This completes the proof of Theorem~\ref{th.FF} when \textbf{[{\footnotesize H$^{\alpha\in (1,2)}_{\text{LG}}$}]} holds.

\subsection{Strong Feller property of    killed semigroups in  a rather  general framework}
In this section we investigate the strong Feller property of the killed semigroup $P_t^{\mathscr D}$, where we recall that   $\mathscr D=\mathscr O\times \mathbb R^d$, where $\mathscr O$ is an open subset of $\mathbb R^d$, see~\eqref{eq.DD}.  
 \begin{theorem}\label{th.SFD}
Let $\mathscr D=\mathscr O\times \mathbb R^d$, where $\mathscr O$ is non empty open subset of $\mathbb R^d$. 
 Consider the process~\eqref{eq.Lan} driven by a pure jump  Lévy process $(L_t,t\ge 0)$. Assume in addition that either {\rm\textbf{[{\footnotesize B$_{\text{smooth}}$}]}} holds, or that the process $(L_t,t\ge0)$ satisfies {\rm\textbf{[{\footnotesize W$_{\text{LG}}$}]}}.     Assume also  that $(P_t,t>0)$ is strongly Feller over $\mathbb R^{2d}$.    
  Then,   $(P_t^{\mathscr D},t> 0)$ is strongly Feller. 
 \end{theorem}

  \noindent
   \textbf{Note}.   
   The  assumptions of Theorem~\ref{th.SFD} are satisfied if either   
 {\rm\textbf{[{\footnotesize H$^{\alpha\in (0,2)}_{\text{smooth}}$}]}}  or {\rm\textbf{[{\footnotesize H$^{\alpha\in (1,2)}_{\text{LG}}$}]}}   holds, see  indeed Theorems \ref{th.ex2} and~\ref{th.FF}. In these two settings, according to Theorem~\ref{th.SFD}, $(P_t^{\mathscr D},t> 0)$ is thus strongly Feller. 

 \begin{proof}  Since $(P_t,t>0)$ is strongly Feller over $\mathbb R^{2d}$, the proof follows exactly the same lines as the one made to prove     \cite[Theorem 2.2]{chung2001brownian} (see also~\cite[Example 5.2]{Girsanov})  if the following condition is satisfied: 
 for all  $K \Subset \mathbb R^{2d}$ and $\delta>0$,  
\begin{equation}\label{eq.Lpop}
\lim_{s\to 0^+} \sup_{\mathsf x=(x,v)\in K}\mathbb P_{\mathsf x}[  \sigma_{\mathscr  B_{\mathbb R^d}( x,\delta)}\le s]=0,
\end{equation}
 where $ \sigma_{\mathscr  B_{\mathbb R^d}( x,\delta)}:=\inf\{ t\ge 0: x_t\notin \mathscr  B_{\mathbb R^d}( x,\delta)\}$ is the first exit time from $\mathscr  B_{\mathbb R^d}( x,\delta)\times \mathbb R^d$ for the process~\eqref{eq.Lan}. 
 
 We now give two different proofs of \eqref{eq.Lpop}.
 \medskip

\noindent
 \textbf{Proof 1}. Elementary proof of \eqref{eq.Lpop} using that $\mathbf B$ has at most linear growth.  For $\mathsf x=(x,v)\in \mathbb R^{2d}$, 
 $$ \mathbb P_{\mathsf x}[  \sigma_{\mathscr  B_{\mathbb R^d}( x,\delta)}\le s]\le \mathbb P_{\mathsf x}\Big [   \sup_{u\in [0,s]} |x_u-x|\ge \delta\Big ].$$
   Let $r_K>0$ is such that   $K\subset \mathscr{B}_{\mathbb{R}^{2d}}(0, r_K)$. By \eqref{eq.Gg0}, for any $u\in  [0,s]$ and $\mathsf x\in K$, it holds that:
 $$|x_u(\mathsf x)-x| \le \Big| \int_0^u v_t(\mathsf x) dt\Big|\le s \Big[ r_K + Cs +S_s^\alpha\Big]e^{Cs}.$$
 Hence  
 $$\mathbb P_{\mathsf x}[  \sigma_{\mathscr  B_{\mathbb R^d}( x,\delta)}\le s]\le \mathbb P_{\mathsf x}\Big[\delta \le s [ r_K + Cs +S_s^\alpha]e^{Cs}\Big ]\to 0.$$
 The r.h.s. of the previous inequality does not depend on $\mathsf x\in K$ and goes to $0$ as $s\to 0^+$. 
 This achieves the proof of \eqref{eq.Lpop}.

 \medskip

\noindent
\textbf{Proof 2}. Proof of 
\eqref{eq.Lpop} using  the strong Markov property  of the weak solutions  \eqref{eq.Lan} (this property   follows from  Proposition \ref{pr.Pre} and e.g.~\cite[Item (c) in Theorem 4.2 in Chapter 4]{EK} or~\cite[Théorème II$_{12}$]{Lepeltier}). 
 Recall that the position process $(x_t,t\ge 0)$ is continuous. We have for every $\mathsf{x}=(x,v) \in \mathbb{R}^{2d}$, 
\begin{align}
\nonumber
&\mathbb{P}_{\mathsf{x}}[\sigma_{\mathscr{B}_{\mathbb{R}^d}(\mathsf{x}, \delta)} \le s] \\
\nonumber
&\le \mathbb{P}_{\mathsf{x}}\Big [\sigma_{\mathscr{B}_{\mathbb{R}^d}(\mathsf{x}, \delta)} < s, x_s \in \mathscr{B}_{\mathbb{R}^d}(x, \delta/2) \Big ]  + \mathbb{P}_{\mathsf{x}}\big [x_s \in \mathscr{B}_{\mathbb{R}^d}(x, \delta/2)^c\big ] \\
\nonumber
&\le \mathbb{P}_{\mathsf{x}}\Big [\sigma_{\mathscr{B}_{\mathbb{R}^d}(x, \delta)} < s, |x_s -x_{\sigma_{\mathscr{B}_{\mathbb{R}^d}(x, \delta)}} |\le \delta/2 \Big ]  + \mathbb{P}_{\mathsf{x}}\big [x_s \in \mathscr{B}_{\mathbb{R}^d}(x, \delta/2)^c\big ] \\
\label{eq.E+}
&\le \mathbb{E}_{\mathsf{x}} \left[  \mathbf 1_{\sigma_{\mathscr{B}_{\mathbb{R}^d}(x, \delta)} < s}\,  \mathbb{P}_{X_{\sigma_{\mathscr{B}_{\mathbb{R}^d}(x, \delta)}}} \big [|x_{s - \sigma_{\mathscr{B}_{\mathbb{R}^d}(x, \delta)}} - x_0| \ge \delta/2\big ]   \right] + \mathbb{P}_{\mathsf{x}}\big [|x_s - x| \ge \delta/2\big ].  
\end{align}
Moreover, note that by \eqref{eq.Rrp} (where we emphasize  that  the derivation   \eqref{eq.Rrp} is still valid under {\rm\textbf{[{\footnotesize W$_{\text{LG}}$}]}}), it holds for any $K \Subset \mathbb R^{2d}$, 
\begin{equation}
\label{eq.Gnj}
\lim_{s\to 0}  \ \sup_{\mathsf x\in K} \ \mathbb{P}_{\mathsf{x}}\big [|x_s - x| \ge \delta/2\big ]=0.
\end{equation}
It thus only remains to deal with the term appearing in the l.h.s. of  \eqref{eq.E+}.  
We have 
$$X_{\sigma_{\mathscr{B}_{\mathbb{R}^d}(x, \delta)}}= (x_{\sigma_{\mathscr{B}_{\mathbb{R}^d}(x, \delta)}},v_{\sigma_{\mathscr{B}_{\mathbb{R}^d}(x, \delta)}}) \text{ and } |x_{\sigma_{\mathscr{B}_{\mathbb{R}^d}(x, \delta)}}(\mathsf x)-x|= \delta.$$
Unfortunately, we have no such  direct information on the location of the velocity at time     $ \sigma_{\mathscr{B}_{\mathbb{R}^d}(x, \delta)}$. To get some, we proceed as follows.  Let  $\mathsf x=(x,v)\in K \Subset \mathbb R^{2d}$ and $r_K>0$ such that $K\subset \mathscr{B}_{\mathbb{R}^{2d}}(0, r_K)$. From now on, $s\in [0,1]$. By \eqref{eq.Gg0} (which clearly also holds for $X$), we have when $\sigma_{\mathscr{B}_{\mathbb{R}^d}(x, \delta)} < s$, 
 $$|X_{\sigma_{\mathscr{B}_{\mathbb{R}^d}(x, \delta)}}(\mathsf x)|\le e^{C}[ r_K+ C + S_1^\alpha].$$ 
For $\mathsf x\in K$, we  then write for a parameter $\kappa>0$:
 \begin{align*}
\forall \mathsf x\in K, \ \ \   &\mathbb{E}_{\mathsf{x}} \left[  \mathbf 1_{\sigma_{\mathscr{B}_{\mathbb{R}^d}(x, \delta)} < s}\,  \mathbb{P}_{X_{\sigma_{\mathscr{B}_{\mathbb{R}^d}(x, \delta)}}} \big [|x_{s - \sigma_{\mathscr{B}_{\mathbb{R}^d}(x, \delta)}} - x_0| \ge \delta/2\big ]   \right]   \\
  &\le \mathbb P[ S_1^\alpha\ge \kappa]+   \sup_{u\in [0,s] } \  \sup_{|\mathsf y|  \le  e^C[r_K+ C+ \kappa]} \ \mathbb{P}_{\mathsf y} \big [|x_u - x_0| \ge \delta/2\big ]\\
  &\le \mathbb P[ S_1^\alpha\ge \kappa]+   s  M_{e^C[r_K+ C+ \kappa]}^{\delta/2}, 
\end{align*}
where we used \eqref{eq.Rrp} to deduce the last inequality.  Now let $\eta>0$ and choose $\kappa(\eta)>0$ large enough such that $\mathbb P[ S_1^\alpha\ge \kappa(\eta)]\le \eta$. Then for every 
$$0\le s\le  \frac{\eta}{1+M_{e^C[r_K+ C+ \kappa(\eta)]}^{\delta/2}},$$
we have for all $\mathsf x\in K$:
$$\mathbb{E}_{\mathsf{x}} \left[  \mathbf 1_{\sigma_{\mathscr{B}_{\mathbb{R}^d}(x, \delta)} < s}\,  \mathbb{P}_{X_{\sigma_{\mathscr{B}_{\mathbb{R}^d}(x, \delta)}}} \big [|x_{s - \sigma_{\mathscr{B}_{\mathbb{R}^d}(x, \delta)}} - x_0| \ge \delta/2\big ]   \right]\le 2 \eta.$$
This end the proof  of \eqref{eq.Lpop} when   {\rm\textbf{[{\footnotesize W$_{\text{LG}}$}]}} holds.  
\end{proof}
\begin{remark}
Let us mention that while we still used in the second proof of \eqref{eq.Lpop}   that   $\mathbf{B}$ has at most linear drift, probably one advantage of this method (and more precisely of the decomposition \eqref{eq.E+}) is its extensibility to solutions   to  SDEs with  non-locally bounded drifts, a case we would like to study in future works for \eqref{eq.Lan}.
\end{remark}


\subsection{On the spectral radius of strongly Feller semigroups} 
We now establish criteria for the positivity of the spectral radius within the general class of non-negative sub-Markovian kernels. These results serve as a prerequisite for analyzing the positivity of the spectral radius of the killed semigroups $P_t^{\mathscr D}$ (see the note just below). The following theorem formalizes this general approach. 
 
  \begin{proposition}\label{pr.Rsp}
  Let $Q_t(\mathsf x,d\mathsf y)$ be a non-negative sub-Markovian kernel over a nonempty  open subset  $\mathscr M$ of $\mathbb R^{2d}$ which is strongly Feller and topologically irreducible  over $\mathscr M$.   
Then, for any $t>0$,  $\mathsf r_{sp}(Q_t|_{bB (\mathscr M)})>0$. 
 \end{proposition}

  \noindent
   \textbf{Note}. The killed semigroup $(P_t^{\mathscr D},t> 0)$ satisfies the assumptions of Proposition~\ref{pr.Rsp}  when 
 {\rm\textbf{[{\footnotesize H$^{\alpha\in (0,2)}_{\text{smooth}}$}]}} or  {\rm\textbf{[{\footnotesize H$^{\alpha\in (1,2)}_{\text{LG}}$}]}}, or  {\rm\textbf{[{\footnotesize H$^{\alpha\in (1,2)}_{\text{P-Grad}}$}]}}    holds, see indeed Theorem~\ref{th.SFD}, Theorem   \ref{th.C5-infty}  (and their notes) and the proof of Theorem~\ref{co.main}.   Thus, in these cases, $\mathsf r_{sp}(P_t^{\mathscr D}|_{bB (\mathscr D)})>0$. 

\begin{proof}
    Let   us consider  for some   arbitrary point  $\mathsf x_*\in \mathscr M$, the measure 
    $$\phi(\cdot):=\int_{0}^{+\infty} e^{-s} Q_s(\mathsf x_*,\cdot)\, ds.$$
   If $\phi(\mathscr A)>0$,  $\exists s_{\mathscr A}>0$ such that $Q_{  s_{\mathscr A}}(x_*,\mathscr A)>0$. Since 
  $Q_{ s_{\mathscr A}}$ is strongly Feller, there exists a  neighborhood $\mathscr V_*\neq \varnothing$ of $\mathsf x_*$ in $\mathscr M$ and  $c_*>0$ such that $Q_{ s_{\mathscr A}}(\mathsf y ,\mathscr A)\ge c_*$ for all $\mathsf y\in \mathscr V_*$. 
   We assume that $t=1$ (the case $t\neq 0$ is treated the same way).   
   Let $n>  s_{\mathscr A}$ and write $n =u+{s_{\mathscr A}}$ with $u>0$. Let $\mathsf x\in \mathscr M$. We have by assumption, 
   $$Q_n(\mathsf x,\mathscr A)= \int_{\mathscr M} Q_u(\mathsf x,d\mathsf y)Q_{ s_{\mathscr A}} (\mathsf y,\mathscr A)\ge c_* Q_u (\mathsf x, \mathscr V_*) >0$$
     In particular,   $Q_1$ is $\phi$-irreducible~\cite[Definition 2.2]{Nummelin}. 
 Let $\psi$ be    a maximal irreducibility measure for  $Q_1$. By~\cite[Theorem 2.1]{Nummelin}, there exist $m\ge 1$, $\beta>0$, $\mathbf s: \mathscr M\to \mathscr R_+$ measurable  with $\psi(\mathbf s)>0$ and a   measure $\nu\ge 0$ over $(\mathscr M,  B(\mathscr M))$ with  $\nu(\mathscr M)>0$ such that    
\begin{equation}\label{eq.Lpetite}
Q_m (\mathsf x, \mathscr V)\ge \beta \, \mathbf s(x)\,  \nu(\mathscr V),\ \ \forall \mathsf  x\in \mathscr M,\, \forall  \, \mathscr V \in  B(\mathscr M).
\end{equation}
Note that in particular $\nu(\mathscr M)<+\infty$ and $\mathbf s\in bB(\mathscr M)$. 
Note that we can assume    $\nu(\mathbf s)>0$. Indeed, $Q_n \mathbf s$ is small by  multiplying  both side of~\eqref{eq.Lpetite} by $Q_n\mathscr M$. Let $\mathscr A=\{\mathbf s>0\}$. Since $\psi(\mathscr A)>0$, it holds by the previous analysis, for all $\mathsf x\in \mathscr M$ and $n> s_{\mathscr A}$,  $Q_n(x,\mathscr A)>0$  and thus $Q_n \mathbf s(x)>0$. In particular,  
$$\nu(Q_n \mathbf s)>0.$$
We can thus  assume    $\nu(\mathbf s)>0$ up to considering $Q_n \mathbf s$ instead of $\mathbf s$. We now conclude the proof of the proposition, whose arguments are rather standard. We write them for the sake of completeness.  
Set   $J=Q_m$ where $m$ is as in~\eqref{eq.Lpetite}. 
 We  have by~\eqref{eq.Lpetite}, for every $f\in b B (\mathscr M)$, $f\ge 0$,  
 $$\nu (Jf)\ge \beta \nu(\mathbf s)   \nu(f). $$
 Hence, $\nu (J^qf)\ge \nu(f)\beta^q\nu (\mathbf s)^q $, $\forall q\ge 1$. Let $\mathsf x\in \mathscr O$ such that $\mathbf s(\mathsf x)>0$. Since  $\Vert \mathbf 1\Vert_{b B(\mathscr M)} \le 1$:
\begin{align*}
\Vert J^k\Vert_{\mathcal L(b  B(\mathscr M))}\ge (J^k \mathbf 1)(\mathsf x) \ge  \beta \mathbf s(x)  \nu( J^{k-1} \mathbf 1)  \ge \beta^k \mathbf s(x) \,    \nu(\mathbf 1)\nu (\mathbf s)^{k-1}, \ \forall k\ge 1.
\end{align*} 
Then, by Gelfand's formula, $\mathsf r_{sp}(J|_{bB(\mathscr M)}) \ge \beta \nu(\mathbf s)>0$. Since  $\mathsf r_{sp}(Q_1|_{b (\mathscr M)})^m=\mathsf r_{sp}(J |_{b (\mathscr M)}) $, one deduces that $\mathsf r_{sp}(Q_1|_{b (\mathscr M)})>0$.  The proof of the proposition is complete. 
\end{proof}

\section{Stationary and quasi-stationary distributions}
\label{sec.QSD-Sd}
In this section, we study the existence and uniqueness of   quasi-stationary and stationary distributions under \textbf{[{\footnotesize H$^{\alpha\in (1,2)}_{\text{LG}}$}]} and \textbf{[{\footnotesize{H$^{\alpha\in (1,2)}_{\text{P-Grad}}$}}]} respectively,  establishing exponential convergence to the quasi-stationary distribution for the conditioned process and  exponential ergodicity for the non-killed process.

\subsection{Quasi-stationary distributions under {\rm\textbf{[{\footnotesize H$^{\alpha\in (0,2)}_{\text{smooth}}$}]}} and \textbf{[{\footnotesize H$^{\alpha\in (1,2)}_{\text{LG}}$}]}} 
\label{sec.QSD}
We recall the   definition of quasi-stationary distribution.
\begin{definition}
\label{de.QSD}
 A quasi-stationary distribution     for  the process  $(X_t,t\ge 0)$ (see ~\eqref{eq.Lan}) in  $\mathscr   D$ is a probability measure  $\mu$ on $\mathscr    D$
such that
$$
\mu(A)=\mathbb P_{\mu}[X_t\in A| t<\sigma_{\mathscr D}],\ \forall t>0, A\in B(\mathscr D),
$$
  where $B(\mathscr D)$ is the Borel $\sigma$-algebra of ${\mathscr D}$.
\end{definition}

As an application of the previous  results, we deduce existence and uniqueness of a quasi-stationary distribution for   kinetic Langevin processes killed when exiting $\mathscr D=\mathscr O\times \mathbb R^d$.

\begin{theorem}\label{th.main}
 Assume  that either {\rm\textbf{[{\footnotesize H$^{\alpha\in (0,2)}_{\text{smooth}}$}]}} or {\rm\textbf{[{\footnotesize H$^{\alpha\in (1,2)}_{\text{LG}}$}]}}  holds. Consider a bounded subdomain  $\mathscr O$ of $\mathbb R^d$ such that $\mathbb R^{d}\setminus \bar{\mathscr O}$ is nonempty. Let $\mathscr D= \mathscr O\times \mathbb R^d$, see~\eqref{eq.DD}. Then, the process~\eqref{eq.Lan} admits a  unique quasi-stationary distribution $\mu$ in $\mathscr D$. Moreover, 
 \begin{enumerate}
 \item[$\bullet$]    There exists $ \lambda>0$ such that for all $t\ge0$, the spectral radius of $P_t^{\mathscr D}$ on $bB(\mathscr D)$ is given by 
$$\mathsf r(P_t^{\mathscr D}|_{bB(\mathscr D)})= e^{-\lambda t}.$$ 
 In addition, $\mu P_t^{\mathscr D}=e^{-\lambda t}\mu$, for all $t\ge 0$, and  $\mu(O)>0$ for all nonempty open subsets $O$ of $\mathscr D$. 
In addition, there is a unique bounded and continuous function $\varphi:\mathscr D\to \mathbb R$  such that  $\mu(\varphi)=1$ and
$$
      P_t^{\mathscr D} \varphi= e^{-\lambda t} \varphi \text{ on } \mathscr D, \forall t\ge0.
$$
 Moreover, $\varphi>0$ everywhere on $\mathscr  D$.

 \item[$\bullet$] There exist   $\delta>0$ and $C\ge 1$ such that for all $\nu \in \mathcal P(\mathscr D)$,
$$
 \big |\mathbb P_\nu[X_t\in \mathscr A| t<\sigma_{\mathscr D}]-\mu( \mathscr A)\big |\le   \frac{C e^{-\delta t}}{\nu(\varphi)}, \ \forall  \mathscr  A\in B(\mathscr D), t>0.
$$
  \item[$\bullet$]  $\mathbb P_{\mathsf x}[\sigma_{\mathscr D}<+\infty]=1$ for all ${\mathsf x}\in \mathscr D$.

  \item[$\bullet$]  The operator $P_t^{\mathscr D}: bB(\mathscr D)\to bB(\mathscr D)$ is compact for $t>0$. 
 
\end{enumerate}
  
  \end{theorem}

%

  \begin{proof}
 The proof  of  Theorem~\ref{th.main} relies  on the Krein–Rutman type theorem~\cite[Theorem 4.1]{guillinqsd} and the proof of Theorem 5.3 there. More precisely,  according to this theorem,  the assertion of Theorem~\ref{th.main} holds if   the following    properties  are satisfied:
 \begin{enumerate}
 \item[\textbf{C1}.]   $(P_t^{\mathscr D},t> 0)$  is  Feller for $t>0$. 
 \vspace{0.1cm}
  \item[\textbf{C2}.]   $(P_t^{\mathscr D},t>0)$ is topologically irreducible over $\mathscr D$.
     \vspace{0.1cm}
  \item[\textbf{C3}.] There is $\mathsf x\in \mathscr D$ such that  $\mathbb P_{\mathsf x}[\sigma_{\mathscr D}<+\infty]>0$.  
   \vspace{0.1cm}
   \item[\textbf{C4}.] The operator $P_t^{\mathscr D}$ has a  spectral gap  in $b B(\mathscr D)$, that is to say:
   $$   \mathsf r_{ess}(P_t^{\mathscr D}|_{b B(\mathscr D)})< \mathsf r_{sp}(P_t^{\mathscr D}|_{b B(\mathscr D)}),$$ 
 \end{enumerate}
 where $ \mathsf r_{ess}(P_t^{\mathscr D}|_{b B(\mathscr D)})$ is the  essential spectrum radius of $P_t^{\mathscr D}$   and $\mathsf r_{sp}(P_t^{\mathscr D}|_{b B(\mathscr D)})$ its spectral radius. 
 We now check this four conditions. 
 
 Let us first recall  that under {\rm\textbf{[{\footnotesize H$^{\alpha\in (1,2)}_{\text{LG}}$}]}}, the assumption \textbf{[{\footnotesize C$_{\text{LG}}$}]}  is  satisfied (see Theorem~\ref{th.ex2}). Recall also that when {\rm\textbf{[{\footnotesize H$^{\alpha\in (0,2)}_{\text{smooth}}$}]}} or {\rm\textbf{[{\footnotesize H$^{\alpha\in (1,2)}_{\text{LG}}$}]}} holds, 
   $(P_t,t> 0)$ is strongly Feller for $t>0$ (see Theorem~\ref{th.FF}). 
In addition, when either {\rm\textbf{[{\footnotesize H$^{\alpha\in (0,2)}_{\text{smooth}}$}]}} or {\rm\textbf{[{\footnotesize H$^{\alpha\in (1,2)}_{\text{LG}}$}]}}  holds, the  conditions   \textbf{C1}$\to$\textbf{C4} are satisfied. Indeed:
   \begin{enumerate}
   \item[-] By Theorem~\ref{th.SFD} (and its note), \textbf{C1} is satisfied. 
       \item[-] By  Theorem~\ref{th.C5-infty} (and its note), \textbf{C2} and \textbf{C3} are satisfied.
             \item[-] By Proposition~\ref{pr.Rsp} (and its note), for all $t>0$,  $\mathsf r_{sp}(P_t^{\mathscr D}|_{bB (\mathscr D)})>0$. Moreover, because $\mathscr O$ is bounded,  by Theorem~\ref{th.bw} (and the discussion below), $\mathsf r_{ess}(P_t^{\mathscr D}|_{bB (\mathscr D)})=0$, and   $P_t^{\mathscr D}: bB(\mathscr D)\to bB(\mathscr D)$ is compact, for any $t>0$. 
            \end{enumerate}     
 The proof of the theorem is complete. 
  \end{proof}

The literature on quasi-stationary distributions is very large and we refer to~\cite{collet2011quasi,meleard2012quasi} for  a  comprehensive account of the general theory of quasi-stationary distributions. 
Existence and uniqueness of  quasi-stationary distributions for kinetic processes driven by a Brownian motion have been obtained in~\cite{guillinqsd,ramilarxiv2,champagnat2024quasi} and in~\cite{guillinqsd2,guillinqsd3} for singular potentials including  Coulomb and Lennard-Jones potentials (see also~\cite{guillinFK}). 
We also refer in this direction to the works~\cite{BenaimChampagnatEtAl2022,benaim2025degenerate}. 
The quasi-stationarity of elliptic SDEs driven by rotationally invariant $\alpha$-stable processes has been investigated in \cite{guillin2024large}.  
 Let us also mention~\cite{champagnat2016exponential,champagnat2018criteria,bansaye2022non,del2023stability,villemonais2025quasi} for other general criteria to prove existence and uniqueness of a quasi-stationary distribution, were the strong Feller property is not (or partially) required (see also~\cite{ferrari-kesten-martinez-picco-95,gong1988killed,pinsky1985convergence,takeda1,zhang2014quasi,he2019some,champagnat2021lyapunov} and~\cite{kolb2012quasilimiting,steinsaltz-evans-07} in the one dimensional setting).


\subsection{Perturbed gradient fields: stationary and quasi-stationary distributions}
In this section, we study the existence and uniqueness of stationary and quasi-stationary distributions under \textbf{[{\footnotesize{H$^{\alpha\in (1,2)}_{\text{P-Grad}}$}}]}.

\subsubsection{Stationary   distributions and ergodicity of the non-killed process in the  perturbed gradient field setting}  
We recall that a classical ergodicity result due to Doob and Khasminskii (see~\cite[Theorem 1.1]{peszat1995strong}) states that if the semigroup $(P_t,t>0)$ is strongly Feller   and topologically irreducible, then the process $(X_t,t\ge 0)$ admits at most one invariant probability distribution $\mu$. Moreover, if such a distribution exists, then for any $\mathsf x\in \mathbb R^{2d}$ and every $\mathscr A \in B(\mathbb R^{2d})$, $
P_t(\mathsf x,\mathscr A) \to \mu(\mathscr A)
$ as $ t \to \infty$.   
Assume in addition that  there exists a  function $W: \mathbb{R}^{2d} \to [1, +\infty)$ in the extended domain $D_e(\mathcal L)$ of $\mathcal L$ (see~\cite{davis1993markov} for a definition of $D_e(\mathcal L)$) such that $\lim_{|\mathsf{x}| \to +\infty} W(\mathsf{x}) = +\infty$ and, for some constant $c > 0$, $K\Subset \mathbb R^{2d}$, and some $b\in \mathbb R$:
\begin{equation}\label{eq.Ha2}
\frac{\mathcal L  W}{W} \le -c  \mathbf 1_{K^c} + b \mathbf 1_K \quad \text{ quasi-everywhere on } \mathbb{R}^{2d},
 \end{equation}
 where $\mathcal{L}$ denotes the infinitesimal generator of \eqref{eq.Lan}.  Then by~\cite{down1995exponential} (see also~\cite[Theorem 2.4]{Wu2001}), 
 $(P_t,t\ge 0)$ is \textit{$W$-uniformly exponential ergodic}, i.e. 
 there is a unique invariant probability measure $\mu$ for the process $(X_t,t\ge 0)$ and   moreover:  $\mu(W)<+\infty$, and 
 there exist some $C>0$ and $\lambda >0$ such that for all $t\ge 0$,
$$
\sup_{|f|\le W}
\left|
P_t f(\mathsf x) - \mu(f)
\right|
\le C\, W(\mathsf x) \, e^{-\lambda t},
\qquad \forall \mathsf x \in \mathbb R^{2d}.
$$
In addition,   $P_t$ (for $t>0$) has a spectral gap near its largest eigenvalue $1$ in $bB_W(\mathbb R^{2d})= \{f:\mathbb R^{2d}\to \mathbb R, f \text{ is measurable and } f/W \text{ is bounded}\}$.

\begin{theorem}\label{th.sta}
 Assume     {\rm\textbf{[{\footnotesize{H$^{\alpha\in (1,2)}_{\text{P-Grad}}$}}]}}  holds.    
  Then,   
  \begin{enumerate}
  \item[a.] For any $\mu\in \mathcal P(\mathbb R^{2d})$, there is a unique weak solution to \eqref{eq.Lan} with initial law $\mu$. In addition, for all $\mathsf x \in \mathbb R^{2d}$ and $t>0$, $X_t(\mathsf x)$   has a density   w.r.t. the Lebesgue measure $d\mathsf x'$.
    \item[b.]   The mapping  $\mathsf x\in \mathbb R^{2d}  \to  \mathbb P_{\mathsf x}[X   \in \cdot ] \in \mathcal P(D([0,+\infty), \mathbb{R}^{2d}))$ is  weakly continuous (as is, for all $t\ge 0$,  the mapping $\mathsf x\in \mathbb R^{2d}  \to  \mathbb P_{\mathsf x}[X_t   \in \cdot ]  $ w.r.t. the weak convergence in $ \mathcal P(  \mathbb{R}^{2d})$). 
    \item[c.]  In addition, $(P_t,t>0)$ is strongly Feller and   topologically irreducible over $\mathbb R^{2d}$. 
    \item[d.] Finally, for every $p\in (0, \alpha)$,  $(P_t,t\ge 0)$ is $W_p$-uniformly exponential ergodic where $W_p$ is defined in \eqref{eq.cond-Wp} (see also the condition {\rm\textbf{[ab]}} below on the parameters $a,b>0$)  and where    for some $M>0$  (see \eqref{eq.cond-F}), 
  $$W_p(\mathsf x)\le M \big[ \mathbf U(x) +|v|^2\big]^{p/2}, \ \mathsf x=(x,v)\in \mathbb R^{2d}.$$
  \end{enumerate}
  If   in addition  $\mathbf U$ and $\Theta$  are $C^\infty$, then, for any $\mathsf x\in \mathbb R^{2d}$,  there is a unique pathwise solution to \eqref{eq.Lan} with $ X_0=\mathsf x$,  and  the assertions of Items b, c and d   remain valid  verbatim for every stability index $\alpha\in (0,1]$.
  \end{theorem}

  \noindent
  \textbf{Note}. The invariant probability distribution $\mu_p$  for the process $(X_t,t\ge 0)$  is independent of $p\in(0,\alpha\wedge 1)$.   In addition, the unique weak (or strong when $\mathbf U$ and $\Theta$  are $C^\infty$ are smooth) solutions to~\eqref{eq.Lan}  form a strong Markov process. This can obtained from the theory of  Markov semigroups associated with  càdlàg processes and  preserving $C_b(\mathbb R^{2d})$  (see~\cite[Theorem 6.17]{le2016brownian} and its proof) or via   the theory of martingale problems  in $D([0,+\infty),\mathbb R^{2d})$, see e.g.~\cite[Item (c) in Theorem 4.2 in Chapter 4]{EK} or~\cite[Théorème II$_{12}$]{Lepeltier}.

  \medskip
The study of exponential ergodicity for Lévy-driven SDEs has been developed through several important contributions. Zhang et al.~\cite{zhangAIHP2020} analyzed the exponential ergodicity of elliptic Lévy-driven SDEs with non-smooth drift using the Zvonkin transformation; see also~\cite{liang2021exponential}, where a coupling approach is employed. Kulik~\cite{kulik2009} established general criteria for SDEs with $C^1$ drifts, with particular emphasis on one-dimensional settings; see also~\cite{masuda2007ergodicity}. Additional results can be found in~\cite{sandric2016ergodicity}.
More recently, Bao et al.~\cite{bao2024exponential} studied $\alpha$-stable kinetic Langevin dynamics with singular interaction potentials $\mathbf{V}(x^1, \dots, x^N)$. Their framework includes Lennard–Jones and Coulomb interactions which, although singular, are assumed to be $C^\infty$ on the non-collision set ${x_i \neq x_j}$.

\begin{proof}  
We only consider the low-regularity framework  {\rm\textbf{[{\footnotesize{H$^{\alpha\in (1,2)}_{\text{P-Grad}}$}}]}} where  $\alpha \in (1, 2)$,  the potential $\mathbf{U}$ is assumed to be $C^1$, and the vector field $\Theta$ is merely measurable and locally bounded. The case where $\mathbf{U}$ and $\Theta$ are $C^\infty$ is  treated using the same arguments. In this $C^\infty$ setting, one considers strong solutions and the strong Feller property of the non-killed process is satisfied for the entire range $\alpha \in (0, 2)$. It is important to emphasize that the Lyapunov function $W_p$ constructed in the following step is specifically designed to satisfy the dissipativity condition \eqref{eq.Ha2} for every stability index $\alpha \in (0, 2)$.
\medskip

\noindent
\textbf{Step 1}.  In this step, we construct a Lyapunov function designed to satisfy the dissipativity condition \eqref{eq.Ha2}. For this step only, we allow the stability index $\alpha$ to span its full  range $(0, 2)$.     Recall $\mathbf U$ is $C^1$ and $\mathbf U\ge 1$. 
Let $a>0$. Consider also  $b>0$  to be chosen and  fixed later. Consider the following   function defined by,  for $\mathsf x=(x,v)\in \mathbb R^{2d}$, 
$$F_0(x, v) =   a \big[\mathbf U(x) +  \frac 12 |v|^2\big] + b\,  x \cdot v.$$
The function $\mathbf U(x) +  \frac 12 |v|^2$ is the Hamiltonian of the system. 
The coupling of position and velocity $x \cdot v$ is   a standard  way  in kinetic modeling to incorporate dissipative effects into the spatial evolution of the system. 
We have that 
$$F_0(x, v)\ge  a\mathbf U(x)\Big[ 1- \frac{b}{2a}\frac{|x|^2}{\mathbf U(x)} \Big]  \, +\, \frac 12 (a-b) |v|^2.$$   
By the second condition in \eqref{eq.cond-U}, we have that, for every $|x|$ sufficiently large  (recall $q\ge 2$),  
\begin{align*}
\frac{|x|^2}{\mathbf U(x)}  \le    \frac{2 }{m_2 |x|^{q-2}}, 
\end{align*}
which,  if $q>2$,  goes to $0$ when $|x|\to +\infty$,  and equal $2/m_2$ if $q=2$.
In the case $q=2$, we take $b>0$ such that $b<a$ and $\frac{b}{am_2}\le 1/2$. In conclusion, we take 
\begin{equation}\label{eq.cond-b}
0<b < a\wedge   \frac{m_2 a}2 \ \text{ if } q=2 \ \text{ and } \ 0<b<a \  \text{ if } q>2.
\end{equation} 
Then, with such constants $a,b>0$,    $F_0(\mathsf x)\to +\infty$ as $|\mathsf x|\to +\infty$  and in particular $F_0$ is lower bounded on $\mathbb R^{2d}$.
 In addition, for any  $a,b>0$  satisfying  \eqref{eq.cond-b}, there exists $m,c_m >0$ such that 
$$
F_0(\mathsf x)\ge m[\mathbf U(x) + |v|^2]-C_m, \ \ \forall \mathsf x\in \mathbb R^{2d}. 
$$
Then, set  for $\mathsf x=(x,v)\in \mathbb R^{2d}$:
\begin{equation}\label{eq.F2}
F (\mathsf x) = F_0(\mathsf x)+C_m+1\ge m[\mathbf U(x) + |v|^2]+1.
 \end{equation}
Note that  $F\in C^{1,\infty}(\mathbb R^{2d})$. Here $C^{1,\infty}(\mathbb R^{2d})$ is the space of $C^1$ functions $\psi: \mathbb R^{d}\times \mathbb R^d\to \mathbb R$ such that  for any multi-index $\beta \in \mathbb N^d$,  $(x,v)\mapsto \partial_v ^\beta\psi(x,v)$ exists and is continuous.
Note also that by the second condition in \eqref{eq.cond-U}, there exist   $M>0$ (which depends on $a,b>0$) such that for all $\mathsf x=(x,v)\in \mathbb R^{2d}$:
\begin{equation}\label{eq.cond-F}
  F(\mathsf x)\le M \big[   \mathbf U(x) +|v|^2\big].
\end{equation}   
Then, the Lyapunov candidate $W_p$ is defined by, for   $\mathsf x \in \mathbb R^{2d}$:
\begin{equation}\label{eq.cond-Wp}
W_p(\mathsf x) = F(\mathsf x)^{p/2} \text{ with } 0 < p <   \alpha.
\end{equation}  
The goal is to check that   $W_p$ satisfies \eqref{eq.Ha2}. 
Recall that 
$$\mathcal L^\alpha W_p=  \mathcal S_v^{\nu_\alpha} W_p + v\cdot \nabla_xW_p - \nabla \mathbf U\cdot \nabla_ vW_p+ \boldsymbol \Theta \cdot \nabla_ vW_p.
$$
We estimate both terms above, starting with the drift term. 
We have  for $\mathsf x=(x,v)\in \mathbb R^{2d}$: 
$$\nabla_x F(\mathsf x) = a \nabla \mathbf{U}(x) + b \,   v \text{ and } \nabla_v F(\mathsf x) = av + b \, x.$$
Set $\mathcal{L}_0= v \cdot \nabla_x + (-\nabla \mathbf{U}(x) + \boldsymbol \Theta( \mathsf x)) \cdot \nabla_v$. Then, for $\mathsf x=(x,v)\in \mathbb R^{2d}$, 
$$\mathcal{L}_0 W_p  (\mathsf x) =  \frac{p}{2} F^{\frac{p}{2}-1}(\mathsf x)  \, Q (\mathsf x),$$
where,
\begin{align*}
Q(\mathsf x) &= -b \nabla \mathbf{U}(x) \cdot x + a \boldsymbol \Theta( \mathsf x) \cdot v +  b |v|^2 + b \, \boldsymbol \Theta( \mathsf x) \cdot x.
\end{align*}
We now treat the term $Q$ depending on whether \eqref{eq.cond-gammaB1} or \eqref{eq.cond-gammaB2} holds. 
\medskip

\noindent
\underline{Case I}:  \eqref{eq.cond-gammaA} and \eqref{eq.cond-gammaB1} hold.
\medskip

\noindent 
By \eqref{eq.cond-gammaA} and \eqref{eq.cond-gammaB1}, we have that   $a \boldsymbol \Theta(\mathsf x) \cdot v \le -a \Gamma |v|^2 + a C^1_\Theta (1 + |x|^{\ell_1})$ and $b\, \boldsymbol \Theta(\mathsf x) \cdot x  \le   b C^2_\Theta   + b C^2_\Theta |v|^2 + b C^2_\Theta |x|^{\ell_2}$. 
In addition, by the first condition in \eqref{eq.cond-U}, $-b \nabla \mathbf{U}(x) \cdot x \le -b m_1 \mathbf{U}(x) + b C^1_{\mathbf{U}}$. We then have  for $\mathsf x=(x,v)\in \mathbb R^{2d}$,
\begin{align*}
Q(\mathsf x)  
&\le -b m_1 \mathbf{U}(x) + a C^1_\Theta|x|^{\ell_1}   + b C^2_\Theta |x|^{\ell_2}    \\
&\quad - (a \Gamma - b[1+ C^2_\Theta]) |v|^2  + b C^1_{\mathbf U} + a C^1_\Theta + b C^2_\Theta.
\end{align*} 
Recall that  $\ell_1,\ell_2 \in [1,q)$ and $q\ge 2$. 
 In particular $|x|^{\ell_i}/ \mathbf U(x)\to 0$ as $|x|\to +\infty$ (by the second condition in \eqref{eq.cond-U}). 
 We then only need to  pick $b>0$   such that: 
\begin{equation}\label{eq.cond-b1}
b<  \frac{a \Gamma}{1+ C^2_\Theta}. 
\end{equation}
\medskip

\noindent
\underline{Case II}: \eqref{eq.cond-gammaA} and \eqref{eq.cond-gammaB2} hold. 
\medskip

\noindent
By \eqref{eq.cond-gammaA} and \eqref{eq.cond-gammaB2}, we have that 
 $a \boldsymbol \Theta(\mathsf x) \cdot v \le -a \Gamma |v|^2 + a C^1_\Theta (1 + |x|^{\ell_1})$ 
and 
 $b |\boldsymbol \Theta(\mathsf x) \cdot x| \le   b C^2_\Theta |x| + b C^2_\Theta |v||x| + b C^2_\Theta |x|^{\ell_2}$. 
Recall that by  \eqref{eq.cond-U}, $-b \nabla \mathbf{U}(x) \cdot x \le -b m_1 \mathbf{U}(x) + b C^1_{\mathbf{U}}$. We then have  for $\mathsf x=(x,v)\in \mathbb R^{2d}$, 
\begin{align*}
Q(\mathsf x) &\le -b m_1 \mathbf{U}(x) + b C^2_\Theta |x|  + a C^1_\Theta |x|^{\ell_1} + b C^2_\Theta |x|^{\ell_2}  \\
&\quad - (a \Gamma - b) |v|^2 + b C^2_\Theta |v||x|   + b C^1_{\mathbf U} + a C^1_\Theta.
\end{align*} 
For the cross term $b C^2_\Theta |v||x|$ we have for every $\epsilon>0$:
$$b C^2_\Theta |v||x| \le \epsilon |v|^2 + \frac{(b C^2_\Theta)^2}{4\epsilon} |x|^2.$$
Hence, we obtain that:
\begin{align*}
Q(\mathsf x) &\le -b m_1 \mathbf{U}(x) + \frac{(b C^2_\Theta)^2}{4\epsilon} |x|^2 + b C^2_\Theta |x|  + a C^1_\Theta |x|^{\ell_1} + b C^2_\Theta |x|^{\ell_2} \\
&\quad - (a \Gamma - b -\epsilon) |v|^2   + b C^1_{\mathbf U} + a C^1_\Theta.
\end{align*} 
Recall that    $|x|^{\ell_i}/ \mathbf U(x)\to 0$ as $|x|\to +\infty$. 
 If $q>2$ (in this case $|x|^2/ \mathbf U(x)\to 0$ as $|x|\to +\infty$), we take $\epsilon=b$ and  we only need to pick $b>0$  small enough such that: 
\begin{equation}\label{eq.cond-b2}
2b<  {a \Gamma}. 
\end{equation}
If $q=2$, we    take $\epsilon=\sqrt b$, so that 
\begin{align*}
Q(\mathsf x) &\le -b \Big[m_1 \mathbf{U}(x) - \sqrt b\,  \frac{ |C^2_\Theta|^2}{4} |x|^2\Big] + b C^2_\Theta |x|  + (a C^1_\Theta + b C^2_\Theta) |x|^\ell  \\
&\quad - (a \Gamma - b -\sqrt b) |v|^2    + b C^1_{\mathbf U} + a C^1_\Theta. 
\end{align*}
Thus, if $q=2$,  by the second condition  in \eqref{eq.cond-U}, for $|x|$ large enough, $\mathbf U(x)\ge m_2|x|^2/2$, and we have 
\begin{align*} 
  \mathbf{U}(x) - \sqrt b\,  \frac{ |C^2_\Theta|^2}{4} |x|^2 &=  m_1\mathbf U(x) \Big[ 1-   \sqrt b\,  \frac{ |C^2_\Theta|^2}{4 m_1} \frac{|x|^2}{ \mathbf U(x)}\Big] \ge m_1\mathbf U(x) \Big[ 1-   \sqrt b\,  \frac{ |C^2_\Theta|^2}{2 m_1m_2}  \Big].
\end{align*} 
We then pick $b>0$ small enough such that: 
\begin{equation}\label{eq.cond-b3}
\sqrt b < \frac{m_1m_2}{ |C^2_\Theta|^2} \ \,  \text{ and }  \ \, \sqrt b + b<  a \Gamma.
\end{equation} 
To sum up all these conditions on $a,b>0$, we  say that condition \textbf{[ab]} is satisfied  if  \eqref{eq.cond-b} holds and:
\begin{enumerate}
\item[-] When  \eqref{eq.cond-gammaA} and \eqref{eq.cond-gammaB1} are satisfied, the condition  \eqref{eq.cond-b1} holds. 
\vspace{0.1cm}
\item[-] When  \eqref{eq.cond-gammaA} and \eqref{eq.cond-gammaB2} are satisfied, the condition   \eqref{eq.cond-b2} holds if $q>2$, whereas,  if $q=2$, the condition \eqref{eq.cond-b3} holds. 
\end{enumerate} 
From now on, we assume \textbf{[ab]}. 
With such a choice of constants $a,b>0$,  by \eqref{eq.cond-F}, there exist $C,c>0$ such that 
\begin{align*}
 Q(\mathsf x) &\le -   c\big[\mathbf U(x)   + |v|^2\big]  +C \le -\frac cM F (\mathsf x)+ C. 
\end{align*} 
Thus, there exist  $c_p,M_p>0$   such that  
$$\frac{ \mathcal{L}_0 W_p(\mathsf x)}{W_p(\mathsf x)}\le -c_p  +  \frac {M_p}{F(\mathsf x)} \to -c_p \  \text{ as }\ |\mathsf x|\to +\infty.$$
Let us now estimate  
$$\mathcal S_v^{\nu_\alpha} W_p(x, v) = \int_{\mathbb{R}^d} \left[ W_p(x, v+z) - W_p(x, v) - \nabla_v W_p(x, v) \cdot z \mathbf{1}_{|z|\le 1} \right] \nu_\alpha(dz),$$
where we recall that the Lévy measure is $\nu_\alpha(dz) = c_{\alpha,d} |z|^{-(d+\alpha)} dz$.
We claim that $\nabla_v^2 W_p$  is uniformly bounded over $\mathbb R^{2d}$. Indeed, we have for all $\mathsf x=(x,v)\in \mathbb R^{2d}$,
\begin{equation}\label{eq.nabla_v}
\nabla_v W_p(\mathsf{x}) = \frac{p}{2} F(\mathsf{x})^{\frac{p}{2}-1} [av + bx].
\end{equation}
As a consequence, for every $\mathsf x=(x,v)\in \mathbb R^{2d}$:
$$\nabla^2_v W_p(\mathsf x) = \frac{p}{2} \left[ \nabla_v F^{\frac{p}{2}-1}(\mathsf x)\,  \otimes \,  [av + bx] +a \,   F^{\frac{p}{2}-1}(\mathsf x)  \,  \mathbb{I}_d \right],$$
where 
 $$\nabla_v F^{\frac{p}{2}-1} (\mathsf x) =\Big ( \frac{p}{2}-1\Big) F^{\frac{p}{2}-2} (\mathsf x)[av + bx].$$
The function $F^{p/2-1}$ is bounded (because $0<p/2<1$ and $F\ge 1$ is coercive). Using in addition   \eqref{eq.F2}, there is a constant $L_p>0$ such that for all $\mathsf x=(x,v)\in \mathbb R^{2d}$,
 $$|\nabla^2_v W_p (\mathsf x)|\le L_p \Big[1+  \frac{ |x|^2+ |v|^2}{  \Big( m[\mathbf U(x) + |v|^2]+1 \Big)^{2-\frac{p}{2}} }\, \Big ]. $$
 Recall the second condition \eqref{eq.cond-U}. Then, since $4-p>2$ and $q(2 -\frac p2)>2$ (recall $q\ge 2$ and $0< p< 2$), by   the second condition in \eqref{eq.cond-U}, we deduce that: 
\begin{equation}\label{eq.cond-borne}
  \Vert\nabla_v^2 W_p\Vert_\infty<+\infty,
\end{equation}  
proving the claim. With this regularity established, we return to the estimation of the non-local term  $\mathcal S_v^{\nu_\alpha} W_p$. We split the integral over $z\in \mathbb R^d$ appearing  in $\mathcal S_v^{\nu_\alpha} W_p$   into the two regions $|z|\le 1$ and $|z|>1$. 
In the region $|z|\le 1$, we use the second-order Taylor expansion of $W_p$ in the $v$-variable to get that  (recall  \eqref{eq.cond-borne}):
$$\big | W_p(x,v+z) - W_p(x,v) - \nabla_v W_p(x,v) \cdot z \big | \le  \frac{1}{2}  \Vert\nabla_v^2 W_p\Vert_\infty  \, |z|^2.$$
Hence, we have for all $\mathsf x=(x,v)\in \mathbb R^{2d}$:
\begin{align*}
&\int_{|z|\le 1} \left[ W_p(x, v+z) - W_p(x, v) - \nabla_v W_p(x, v) \cdot z \mathbf{1}_{|z|\le 1} \right] \nu_\alpha(dz)\\
&\le  \Vert\nabla_v^2 W_p\Vert_\infty  \int_{|z|< 1} |z|^2\,  \nu_\alpha(dz)\\ 
&\le K_p,
\end{align*}
for some $K_p>0$. 
Let us now consider the region $|z|>1$.
Set  $$\mathcal H(x,v):=\int_{|z|> 1} \left[ W_p(x, v+z) - W_p(x,v) \right] \nu_\alpha(dz),$$ 
The key observation  is that the term 
$$F(x, v+z) - F(x, v) = \frac{a}{2}\left(|v+z|^2 - |v|^2\right) + bx \cdot (v+z - v) = (av + bx) \cdot z + \frac{a}{2}|z|^2$$
is linear in $\mathsf x=(x,v)\in \mathbb R^{2d}$. 
Since $0\le p/2\le 1$, we have for every $A,B\ge 0$,  $|A^{p/2}-B^{p/2}|\le |A-B|^{p/2}$ and $|A+B|^{p/2} \le |A|^{p/2}+  |B|^{p/2}$. Hence,  
we obtain that for every $\mathsf x=(x,v)\in \mathbb R^{2d}$ and $|z|>1$:
\begin{align*}
 \big|W_p(x, v+z) - W_p(x, v)\big| &\le \big|F(x, v+z) - F(x, v)\big|^{p/2}\\
&\le  \Big | (av + bx) \cdot z + \frac{a}{2}|z|^2\, \Big | ^{p/2}\\
&\le   |av + bx|^{p/2} |z|^{p/2} + \left(\frac{a}{2}\right)^{p/2} |z|^p\\
&= \Big[|av + bx|^{p/2} +  \left(\frac{a}{2}\right)^{p/2}\Big] |z|^p. 
\end {align*}
As a result,  we have that for every $\mathsf x=(x,v)\in \mathbb R^{2d}$:
\begin{align*}
\int_{|z|>1} \, \frac{\big|W_p(x, v+z) - W_p(x, v)\big|}{W_p(x,v)}\,   \nu_\alpha(dz)    &\le  \frac{a|v|^{p/2} + b|x|^{p/2} +  \left(\frac{a}{2}\right)^{p/2}}{\Big( m[\mathbf U(x) + |v|^2]+1 \Big)^{ \frac{p}{2}} }\  \int_{|z|>1} \,  |z|^p\, \nu_\alpha(dz). 
\end{align*} 
Note that $\int_{|z|>1} \,  |z|^p\, \nu_\alpha(dz)<+\infty$ because $0<p<\alpha$. 
Hence, we have shown that for $\mathsf x=(x,v)\in \mathbb R^{2d}$:
\begin{align*}
\frac{|\mathcal H(x,v)|}{W_p(x,v)} &\le \int_{|z|>1} \,  \, \frac{\big|W_p(x, v+z) - W_p(x, v)\big|}{W_p(x,v)}\,   \,   \nu_\alpha(dz)      \to 0 \text{ as } |\mathsf x|\to +\infty. 
\end{align*}  
In conclusion,     for all $\mathsf x\in \mathbb R^{2d}$, it holds:
$$\frac{ \mathcal{L}^\alpha W_p(\mathsf x)}{W_p(\mathsf x)}\le -c_p  +  \frac {M_p}{F(\mathsf x)} +  \frac {K_p}{W_p(\mathsf x)}  + \frac{\mathcal H(\mathsf x)}{W_p(\mathsf x)}   \to -c_p \  \text{ as }\ |\mathsf x|\to +\infty.$$
Note also that the function $ x\mapsto {\mathcal H(\mathsf x)}/{W_p(\mathsf x)}$ is continuous, therefore locally bounded. 
This shows that $W_p$ satisfies \eqref{eq.Ha2}.
In particular, there exists $D_p>0$ such that 
\begin{equation}\label{eq.Dp}
 \mathcal{L}^\alpha W_p \le  D_pW_p \text{ over } \mathbb R^{2d}.
 \end{equation}

\noindent
\textbf{Step 2} (Item a). In this step, we prove the weak well-posedness of \eqref{eq.Lan} when \textbf{[{\footnotesize{H$^{\alpha\in (1,2)}_{\text{P-Grad}}$}}]} holds.  Recall the discussion before Proposition~\ref{pr.rt}. 
Let $\mathsf x\in \mathbb R^{2d}$ and  $R>|\mathsf x|$. 
For $R>0$, consider the open and bounded set
\begin{equation}\label{eq.WpR}
\mathscr W_{p,R}:=\{ \mathsf y\in \mathbb R^{2d}, W_p(\mathsf x)<R\}.
\end{equation} 
Define:
\begin{equation}\label{eq.BR2}
\mathbf B_R(\mathsf y):= \mathbf 1_{ \mathsf y \in \bar{\mathscr W}_{p,R} }[- \nabla \mathbf U(\mathsf y) -\Theta(\mathsf y)], \ \mathsf y \in \mathbb R^{2d},
\end{equation} 
 which is measurable and bounded.   
  Introduce   $X^R(\mathsf y)$    the weak solution   to    \eqref{eq.exXtr} with $X^R_0=\mathsf y$ when $\mathbf B_R$ is given by \eqref{eq.BR2} and 
  \begin{equation}\label{eq.zeta_R}
 \zeta_R(\mathsf y)  = \inf \big \{t \ge 0 : W_p( {X}^R_t(\mathsf y)) \ge R \text{ or }W_p({X}^R_{t_-}(\mathsf y)) \ge R\big \}.
\end{equation}  
 Following the reasoning established just before Proposition~\ref{pr.rt},   $X^R_{\cdot \wedge   \zeta_R}(\mathsf x)$  is the unique solution to  $(\mathcal L^\alpha,\mathsf x,\mathscr W_{p,R})$. Moreover, the weak well-posedness of \eqref{eq.Lan} under \textbf{[{\footnotesize{H$^{\alpha\in (1,2)}_{\text{P-Grad}}$}}]}
holds if   we prove that 
 \begin{equation}\label{eq.Ptau}
 \lim_{R\to +\infty} \mathbb P_{\mathsf x}[\zeta_R\le t]=0.
 \end{equation}
To prove \eqref{eq.Ptau}, we first show that
$$
 (Z^R_t (W_p,\mathsf x)= W_p(X^R_{t \wedge \zeta_R}(\mathsf x))e^{-D_pt},t\ge 0)
$$
    is a non-negative $(\mathcal F_t^{X^R(\mathsf x)})_{t\ge 0}$   supermartingale, where  $D_p>0$ is the constant appearing in \eqref{eq.Dp}.
      Applying Itô's formula~\cite[Theorem 4.4.7]{applebaum2009levy}, we have that  
$$ 
W_p(X^R_{t\wedge   {\zeta}_R}(\mathsf x)) - W_p(\mathsf x) -  \int_0^{t\wedge   {\zeta}_R} \mathcal L^\alpha W_p(X^R_s) \,  ds  =  M_{t\wedge   {\zeta}_R}(W_p,\mathsf x),
$$
where   
\begin{align*}
M_{t\wedge   {\zeta}_R}(W_p,\mathsf x) &= \int_0^{t\wedge   {\zeta}_R} \int_{\mathbb R^d\setminus \{0\}} \left[ W_p(x^R_{s-}, v^R_{s-} + z) - W_p(x^R_{s-}, v^R_{s-}) \right] \tilde{N}_\alpha(ds, dz)  \\
&= \int_0^t  \int_{\mathbb R^d\setminus \{0\}} \mathbf 1_{s<    {\zeta}_R} \left[ W_p(x^R_{s-}, v^R_{s-} + z) - W_p(x^R_{s-}, v^R_{s-}) \right] \tilde{N}_\alpha(ds, dz)
\end{align*}
The process $(M_{t\wedge   {\zeta}_R}(W_p,\mathsf x),t\ge 0)$ is a $L^2$-martingale   
if for all $t\ge 0$, 
$$\mathbb{E} \left[ \int_0^{t} \int_{\mathbb{R}^d \setminus \{0\}} \mathbf 1_{s<    {\zeta}_R}  \, \big | W_p(x^R_{s-}, v^R_{s-} + z) - W_p(x^R_{s-}, v^R_{s-})\big   |^2 \, \nu_\alpha(dz) ds \right] < \infty.$$
Recall that when $s<\zeta_R(\mathsf x)$, it holds $W_p(X^R_{s})\le R$. 
Let $\mathscr W^1_{p,R}=\{\mathsf y\in \mathbb R^{d}, \mathbf d_{\mathbb R^{2d}}(\mathsf y,\mathscr W_{p,R})\le 1\}$ be the $1$-neighborhood of $\mathscr W_{p,R}$. 
Thus, the previous  condition is fulfilled  because for all   $(x,v)\in \mathbb R^{2d}$ such that $W_p(X^R_{s})\le R$, one has: 
 $$ \int_{|z|\ge1}  \, \big | W_p(x,v + z) - W_p(x,v)\big   |^2  \, \nu_\alpha(dz)  \le 4 \sup_{\mathsf y \in \mathscr W^1_{p,R}} |W_p(\mathsf y)|^2\  \int_{|z|\ge1} \nu_\alpha(dz),$$
 and by the mean value theorem, 
\begin{align*}
 \int_{|z|<1} \, \big | W_p(x,v + z) - W_p(x,v)\big   |^2 \, \nu_\alpha(dz)  \le  \sup_{\mathsf y \in \mathscr W^1_{p,R}} |\nabla_v W_p(\mathsf y) |^2  \int_{|z|<1} |z|^2 \nu_\alpha(dz). 
 \end{align*} 
 Hence, $(M_{t\wedge   {\zeta}_R}(W_p,\mathsf x),t\ge 0)$ is a $L^2$-martingale   
if for all $t\ge 0$. Recalling that
$$\int_0^{t\wedge   {\zeta}_R} \mathcal L^\alpha W_p(X^R_s) \,  ds=\int_0^{t } \mathbf 1_{s<    {\zeta}_R} \mathcal L^\alpha W_p(X^R_s) \,  ds,$$
we deduce that:
\begin{align*}
d \big(W_p(X^R_{t \wedge {\zeta}_R}(\mathsf x))e^{-D_pt}\big)&= e^{-D_p t} dM_{t\wedge {\zeta}_R}(W_p,\mathsf x) \\
&\quad + e^{-D_p t} \Big[\mathbf 1_{t < {\zeta}_R} \mathcal L^\alpha W_p(X^R_t) - C_p W_p(X^R_{t\wedge \zeta_R})\Big]dt,
\end{align*}
where $(\int_0^t e^{-D_p s} dM_{s\wedge {\zeta}_R}(W_p,\mathsf x),t\ge 0)$ is a true martingale. Hence, using \eqref{eq.Dp},  it follows that  $( W_p(X^R_{t \wedge \zeta_R}(\mathsf x))e^{-D_pt},t\ge 0)$   is a non-negative   true  supermartingale, proving the claim.

We now  conclude the proof of  \eqref{eq.Ptau}.
Note that $\mathbb E_{\mathsf x} [V_p(X^R_{t\wedge \zeta_R})]\le e^{-C_pt} V_p(\mathsf x)$,  but we cannot directly conclude the proof of \eqref{eq.Ptau} with such a bound  because $W_p(X^R_{\zeta_R}(\mathsf x))$ can be strictly less that $R$. 
We  claim that 
$$
\{\zeta_R(\mathsf x) \leq t\}\subset \Big \{\sup_{s\in[0,t]} W_p(X^R_{s\wedge\vartheta_R }(\mathsf x)) \geq R \Big \}. 
$$  
Indeed, on the event  $\{\zeta_R(\mathsf x) <+\infty\}$, by definition of $\zeta_R(\mathsf x)$, it holds:
$$W_p(X^R_{\zeta_R}(\mathsf x)) \geq R \, \text{ (case $1$)} \qquad \text{or} \qquad W_p(X^R_{\zeta_R^-}(\mathsf x)) \geq R \, \text{ (case $2$)}.$$ 
Case $1$.    Since $\zeta_R(\mathsf x) \leq t$, we have:
$$\sup_{s\in[0,t]} W_p(X^R_{s\wedge\zeta_R }(\mathsf x)) 
\geq W_p(X^R_{\zeta_R }(\mathsf x))\ge R.$$
Case $2$. Let  $s_n \nearrow \zeta_R(\mathsf x)$ with $s_n < \zeta_R(\mathsf x)$ for $n\ge 0$  (note that $\zeta_R(\mathsf x)>0$ because $R>|\mathsf x|$ and the process has càdlàg sample paths). Then,    
$$W_p(X^R_{s_n}(\mathsf x) ) \to W_p(X^R_{\zeta_R^-}(\mathsf x)) \geq R.$$
 For $n\ge 0$, 
$s_n \in [0,t]$ (because $\zeta_R(\mathsf x) \leq t$) and $s_n \wedge \zeta_R(\mathsf x) = s_n$. Therefore:
\begin{align*}
\sup_{s\in[0,t]} W_p(X^R_{s\wedge\zeta_R}(\mathsf x)) 
 \geq W_p(X^R_{s_n}(\mathsf x))   \to W_p(X^R_{\zeta_R^-}(\mathsf x))\ge R.  
\end{align*}
This proves the claim.  
Then, by  Doob's maximal inequality with the non-negative supermartingale $Z^R (W_p,\mathsf x)$,   
 we deduce  that 
\begin{equation}\label{eq.lmz}
\mathbb{P}_{\mathsf x}[\zeta_R(\mathsf x) \leq t] \le \mathbb{P}_{\mathsf x}\!\left[\sup_{0\leq s\leq t} W_p(X^R_{t \wedge \zeta_R}) \geq R \right] 
\leq   \frac{W_p(\mathsf x)\, e^{D_pt}}{R}.
\end{equation}
This completes the proof of \eqref{eq.Ptau} and the proof of the weak well-posedness under the assumptions \textbf{[{\footnotesize{H$^{\alpha\in (1,2)}_{\text{P-Grad}}$}}]}.

Let us now prove  the existence of a density for the process   $X_t(\mathsf x)$  ($t>0$) under the assumption \textbf{[{\footnotesize{H$^{\alpha\in (1,2)}_{\text{P-Grad}}$}}]}. Set  for $\mathsf y\in \mathbb R^{2d}$,
\begin{equation}\label{eq.vars_R}
\varsigma_R(\mathsf y)  =   \inf \big \{t \ge 0 : W_p( {X}_t(\mathsf y)) \ge R \text{ or }W_p({X}_{t_-}(\mathsf y)) \ge R\big \},
\end{equation}
where $X(\mathsf x)$ is the unique weak solution to \eqref{eq.Lan}. For $\mathsf x \in \mathscr W_{p,R}$ and $f\in bB(\mathbb R^{2d})$,  it holds:
$
 \mathbb E_{\mathsf x}[f(X_t)\mathbf 1_{t<\varsigma_R}  ]  =  \mathbb E_{\mathsf x}[f(X^R_t)\mathbf 1_{t<{\zeta}_R}  ]$. 
 By Theorem~\ref{th.uni}, $X^R_t(\mathsf x)$ admits a density   and thus the killed (outside of $\mathscr W_{p,R}$) process  also possesses a density i.e.  for every $\mathscr A\in B(\mathscr W_{p,R})$,   $\mathbb P_{\mathsf x}[X_t^R\in \mathscr A, t<{\tau}_R]= \int_{\mathscr A} q^R_t(\mathsf x,\mathsf x') d\mathsf x'$, for some  function $q^R_t(\mathsf x,\mathsf x')$. 
  Let $\mathscr U\in B( \mathbb R^{2d})$.  We have that 
 \begin{align*}
 \mathbb P_{\mathsf x}[X_t\in \mathscr U  ] = \mathbb P_{\mathsf x}[X_t\in \mathscr U, \varsigma_R\le t    ]  +  \int_{\mathscr U\cap \mathscr W_{p,R}} q^R_t(\mathsf x,\mathsf x') d\mathsf x'.
 \end{align*}
  Hence, if  $\mathscr U$ has zero Lebesgue measure, $\mathbb P_{\mathsf x}[X_t\in \mathscr U  ]  \le \mathbb P_{\mathsf x}[ \varsigma_R\le t    ] \to 0$ as $R\to +\infty$ (indeed this is a consequence of the fact that $\varsigma_R =\zeta_R$ in law  together with \eqref{eq.lmz}). Consequently $\mathbb P_{\mathsf x}[X_t\in \mathscr U  ]=0$, proving the desired result.
 
  \medskip

\noindent
\textbf{Step 3} (Item b). Let us now prove  the weak continuity of the trajectories w.r.t. initial conditions   of \eqref{eq.Lan} under the assumption  \textbf{[{\footnotesize{H$^{\alpha\in (1,2)}_{\text{P-Grad}}$}}]}. 
Let $\mathsf x_n\to \mathsf x\in \mathbb R^{2d}$. We want to show that 
 $\mathbb P_{\mathsf x_n}[X\in \cdot ]\to \mathbb P_{\mathsf x}[X\in \cdot ]$
weakly in $\mathcal P(D([0,+\infty),\mathbb R^{2d}))$. 
Let $R_0>0$ such that $|\mathsf x_n|+|\mathsf x|\le R_0$. In what follows $R>R_0$. 
In view of the proof     made when $\mathbf B$ was measurable with at most  linear growth, see more precisely Section~\ref{sec.Cx}, this continuity property    holds under \textbf{[{\footnotesize{H$^{\alpha\in (1,2)}_{\text{P-Grad}}$}}]}  if one proves that  
\begin{enumerate}
\item[\textbf 1.]  For all $K\Subset \mathbb R^{2d}$ and $T\ge 0$, $\lim_{R\to +\infty} \sup_{\mathsf x\in K} \mathbb P_{\mathsf x}[\varsigma_R\le T]=0$. 
\item[\textbf 2.]      For all $T\ge 0$, the sequence $(\mathbb P_{\mathsf x_n}[X_{[0,T]}   \in \cdot ])_{n\ge 0}$   is tight in $\mathcal P(D([0,T], \mathbb{R}^{2d})$. 

\item[\textbf 3.] For all  $0\le t_1\le \ldots \le t_k \le T$ and $g: (\mathbb{R}^{2d})^k\to \mathbb R$,  
$$
\lim_{n\to+\infty }\mathbb E[g(X_{t_1}(\mathsf x_n),\ldots, X_{t_k}(\mathsf x_n)) ]= \mathbb E[g(X_{t_1}(\mathsf x),\ldots, X_{t_k}(\mathsf x)) ].
$$

\end{enumerate}
 Recall that Item $\textbf 1$ is a consequence of  the equality $\varsigma_R\overset{\text{law}}{=} \zeta_R$ and \eqref{eq.lmz}.  The proof of Item $\textbf 3$ is proved exactly as we proved Proposition~\ref{pr.C}  with the help of Item $\textbf 1$. We now establish  Item $\textbf 2$. 
Applying Theorem 13.2 of Billingsley~\cite{billingsley2013} in a manner analogous to the proof of Lemma~\ref{le.tight}, one can show that for any fixed $R > R_0$, the sequence of laws $(\mathbb{P}_{\mathsf{x}_n}[X_{\cdot \wedge \varsigma_R} \in \cdot])_{n \ge 0}$ is tight in   $\mathcal P(D([0, T], \mathbb{R}^{2d}))$. This follows from the fact that the sequence  $(\mathsf x_n)_{n\ge 0}$ is bounded and  $\mathbf{B}$ is bounded by a constant $C_R$ on the set $\bar{\mathscr W}_{p,R}$. 
  Let $K$ be a compact subset of $D(\mathbb R_+,\mathbb R^{2d})$.
 Then, we have:
 \begin{align*}
\mathbb P_{\mathsf x_n}[( X_{s\wedge  \varsigma_R},s\in [0,T]) \in  K]&= \mathbb P_{\mathsf x_n}[( X_{s\wedge  \varsigma_R},s\in [0,T])\in  K, T<  \varsigma_R ]\\
&\quad +\mathbb P_{\mathsf x_n}[( X_{s\wedge  \varsigma_R},s\in [0,T])\in  K, T\ge \varsigma_R ]\\
&=  \mathbb P_{\mathsf x_n}[( X_s,s\in [0,T])\in  K, T<  \varsigma_R ]\\ 
&\quad +\mathbb P_{\mathsf x_n}[( X_{s\wedge  \varsigma_R},s\in [0,T])\in  K, T\ge \varsigma_R ]\\
&=  \mathbb P_{\mathsf x_n}[( X_s,s\in [0,T])\in  K  ]\\ 
& \quad -  \mathbb P_{\mathsf x_n}[( X_s,s\in [0,T])\in  K,   \varsigma_R\le T ]\\
&\quad +\mathbb P_{\mathsf x_n}[( X_{s\wedge  \varsigma_R},s\in [0,T])\in  K, T\ge \varsigma_R ].
\end{align*}
Hence, it holds:
\begin{equation}\label{eq.uj}
\mathbb P_{\mathsf x_n}[( X_s,s\in [0,T])\in  K  ]\ge \mathbb P_{\mathsf x_n}[( X_{s\wedge  \varsigma_R},s\in [0,T]) \in  K] -   \mathbb P_{\mathsf x_n}[T\ge \varsigma_R].
\end{equation} 
 Let $\epsilon>0$. By Item $\textbf 1$ above, we can take $R_\epsilon\ge R_0+1$ such that for all $n\ge 0$, $\mathbb P_{\mathsf x_n}[T\ge \varsigma_{R_\epsilon}]\le\epsilon/2$. For this fixed $R_\epsilon>0$,  let  $K_\epsilon\subset D([0,T],\mathbb R^{2d})$ be a compact set such that for all $n\ge 0$, $\inf_{n\ge 0}\mathbb P_{\mathsf x_n}[( X_{s\wedge  \varsigma_{R_\epsilon}},s\in [0,T]) \in  K_\epsilon]     >1-\epsilon/2$ (this is possible because $(\mathbb P_{\mathsf x_n}[X_{\cdot \wedge  \varsigma_R}  \in \cdot ])_{n\ge 0}$ is tight for all $R>0$).  By~\eqref{eq.uj}, we have that 
$$
\mathbb P_{\mathsf x_n}[( X_s,s\in [0,T])\in  K_\epsilon  ]\ge  1-\epsilon.
$$
Hence  the sequence of probability measures  
$$(\mathbb P_{\mathsf x_n}[( X_s,s\in [0,T])\in \cdot ])_{n\ge 0} \text{ is tight in $\mathcal P(D([0,T], \mathbb{R}^{2d}))$}.$$ 
This proves  Item $\textbf 2$ and ends the proof of the continuity property   under \textbf{[{\footnotesize{H$^{\alpha\in (1,2)}_{\text{P-Grad}}$}}]}.
 \medskip

\noindent
\textbf{Step 4} (Item c). Strong Feller property and topological irreducibility of $P_t$, $t>0$, when  \textbf{[{\footnotesize{H$^{\alpha\in (1,2)}_{\text{P-Grad}}$}}]} holds.  
 \medskip

\noindent
\textbf{Step 4a}. Topological irreducibility.
In this step, we actually prove that, when \textbf{[{\footnotesize{H$^{\alpha\in (1,2)}_{\text{P-Grad}}$}}]} holds,  for all $\mathscr D=\mathscr O\times \mathbb R^d$, where $\mathscr O$ is a subdomain of $\mathbb R^d$, $P_t^{\mathscr D}$ is 
topologically irreducible for every $t>0$. To do so, the arguments are the same as those used to prove Theorem~\ref{th.C5-infty} except one little point that should be checked. Let us recall that the  proof of Theorem~\ref{th.C5-infty} is based on two intermediate results, Proposition \ref{pr.Co}  and Proposition \ref{pr.C5-infty}.
Recall also that we have proved in the previous steps  the weak well-posedness of \eqref{eq.Lan} and   the weak continuity of their trajectories w.r.t. initial conditions $\mathsf x\in \mathbb R^{2d}$ when \textbf{[{\footnotesize{H$^{\alpha\in (1,2)}_{\text{P-Grad}}$}}]} holds. Therefore,   the assertion of Proposition \ref{pr.Co} holds verbatim in this setting.   For that reason,   we thus only  have to show that the assertion of Proposition~\ref{pr.C5-infty} remains valid when \textbf{[{\footnotesize{H$^{\alpha\in (1,2)}_{\text{P-Grad}}$}}]} holds.  
In light of the proof of Proposition~\ref{pr.C5-infty}, the only remaining point to verify is that the localization argument used at the beginning of \textbf{Step 2}  in this proof can be applied again. Let us recall that this  argument justifies that, without loss of generality, one may restrict to the case  when \eqref{eq.Hgp} holds. Let us check this argument. First of all, because $\mathbf B$ satisfies \textbf{[{\footnotesize{B$_{\text{P-Grad}}$}}]}, it holds:
\begin{equation}\label{eq.Cpm}
\forall r>0, \exists C_r>0,  \forall \mathsf x=(x,v)\in \mathbb R^{2d}, \  \mathbf 1_{|x|\le r} | \mathbf B(\mathsf x)|\le C_r(1+|v|). 
\end{equation}
The property  \eqref{eq.Cpm} is  straightforward to derive because $\nabla \mathbf U$ is continuous  and $\boldsymbol \Theta$ satisfies   \eqref{eq.cond-gammaB2} or the second condition in  \eqref{eq.cond-gammaB1}. We now detail the localization argument and conclude the proof of the topological irreducibility property of $P_t^{\mathscr D}$ in the current setting, namely \textbf{[{\footnotesize{H$^{\alpha\in (1,2)}_{\text{P-Grad}}$}}]}.    

Let $ \mathsf x_0=(x_0,v_0)$ and $\mathsf x_F=(x_F,v_F)$  in $\mathscr D=\mathscr O\times \mathbb R^{d}$. Consider   a bounded subdomain $\mathscr O^\star$  of    $\mathscr O$ containing  $ x_0$ and $ x_F$, with $\bar{\mathscr O}^\star \subset \mathscr O$. Set for $\mathsf x=(x,v)\in  \mathbb R^{2d}$, 
$$\mathbf B^*(\mathsf x)=   \mathbf 1_{x\in \bar{\mathscr O}^\star} \mathbf B(\mathsf x).$$ 
By \eqref{eq.Cpm}, the drift  $\mathbf B^\star $ satisfies \eqref{eq.Hgp}  and in particular, it  satisfies   {\rm\textbf{[{\footnotesize B$_{\text{LG}}$}]}}. 
 Then, by   Theorem~\ref{th.ex2}, one can consider the unique  weak solution $(X^\star_t=(x^\star_t,v^\star_t),t\ge 0)\in \mathbb R^d \times \mathbb R^d $ of: 
 $$
  dx^\star_t=v^\star_tdt, \ dv^\star_t=\mathbf B^\star(x^\star_t,v^\star_t)dt   +   dL^\alpha_t.
 $$
 Set $\mathscr D^\star=\mathscr O^\star \times \mathbb R^d$ and $\sigma^\star_{\mathscr D^\star}:=\inf\{t\ge 0, x^\star_t\notin \mathscr O^\star\}$. For $\mathsf x\in \mathbb R^{2d}$ with $x\in \mathscr O^\star$ both $X_{\cdot \wedge \sigma_{\mathscr D}}(\mathsf x)$ and $X^\star_{\cdot\wedge \sigma_{\mathscr D}^\star}(\mathsf x)$ solve the stopped martingale problem  $(\mathcal L_\star^\alpha, \mathsf x, \mathscr D^\star)$ where $\mathcal L_\star^\alpha$ is the generator of $X^\star$. This stopped martingale problem  is well-posed by~\cite[Theorem 6.1 in Chapter 4]{EK},  Theorem~\ref{th.ex2}, and Proposition~\ref{pr.Pre}. Thus, the two processes  $
X_{\cdot \wedge \sigma_{\mathscr D}}(\mathsf x)$ and $X^\star_{t\wedge \sigma_{\mathscr D}^\star}(\mathsf x)$ are equal in law.  
 Hence,    for all $\boldsymbol\epsilon>0$ and every $\mathsf x=(x,v)$ with  $x\in \mathscr O^\star$,  we have that:
\begin{align*}
P_{t}^{\mathscr D} \big (\mathsf x, \mathscr B_{\mathbb R^{2d}}(\mathsf x_F,\boldsymbol\epsilon)  \big )&\ge P_{t}^{\mathscr D^\star} \big (\mathsf x, \mathscr B_{\mathbb R^{2d}}(\mathsf x_F,\boldsymbol\epsilon)  \big ) =\mathbb P_{\mathsf x}\big [t<\sigma^\star_{\mathscr D^\star}, X_t^\star\in \mathscr B_{\mathbb R^{2d}}(\mathsf x_F,\boldsymbol\epsilon)\big ]. 
\end{align*}
As a consequence,  by substituting the process $(X^\star_t,t\ge 0)$ for the original one (namely for $(X_t,t\ge 0)$) and working over  $\mathscr D^\star$ instead of $\mathscr D$, we may henceforth assume without loss of generality that     \eqref{eq.Hgp}   is satisfied.  Consequently, and as explained previously,   we can apply verbatim all the arguments used in  the proof of
 Proposition~\ref{pr.C5-infty} to deduce that for every $\mathsf x_0,\mathsf x_F\in \mathscr D$, and   $\boldsymbol\epsilon>0$, there exists $t_{\boldsymbol\epsilon}>0$ such that for all $t\in (0,t_{\boldsymbol\epsilon}]$ and every  $\mathsf x \in \{\mathsf x_0\}\cup \bar{\mathscr B}_{\mathbb R^{2d}}(\mathsf x_F,\boldsymbol\epsilon) $,
 $$ P_{t}^{\mathscr D} \big (\mathsf x, \mathscr B_{\mathbb R^{2d}}(\mathsf x_F,\boldsymbol\epsilon)  \big )>0,$$
 which is precisely the statement of Proposition~\ref{pr.C5-infty}. 
   In conclusion,   when \textbf{[{\footnotesize{H$^{\alpha\in (1,2)}_{\text{P-Grad}}$}}]} holds, $P_t^{\mathscr D}$ is topologically irreducible for all $t>0$.   
 \medskip

\noindent
\textbf{Step 4b}. The strong Feller property of $P_t$. 
We have for $\mathsf x \in \mathscr W_{p,R}$ and $f\in bB(\mathbb R^{2d})$ (see \eqref{eq.BR2} and the lines below):
$$
\mathbb E_{\mathsf x}[f(X_t)]=  \mathbb E_{\mathsf x}[f(X_t)\mathbf 1_{\varsigma_R} \le t ]  +  \mathbb E_{\mathsf x}[f(X^R_t)\mathbf 1_{t<{\zeta}_R}  ].$$ 
By Theorem~\ref{th.FF},  for every $R>0$, $\mathsf x\mapsto \mathbb E_{\mathsf x}[f(X^R_t)]$ is continuous. 
Therefore, the strong Feller property  of $P_t$ for $t>0$ is proved exactly as in Section~\ref{sec.SFH2} using Item $\textbf 1$ above, the details are left to the reade.  This completes the proof of Theorem~\ref{th.sta}.  
\end{proof}


 We  extend  Theorem~\ref{th.main} to the   perturbed gradient field setting.   

\begin{theorem}\label{co.main}
Assume     {\rm\textbf{[{\footnotesize{H$^{\alpha\in (1,2)}_{\text{P-Grad}}$}}]}}. 
Consider a bounded subdomain $\mathscr O$ of $\mathbb R^d$ such that $\mathbb R^{d}\setminus \overline{\mathscr O}$ is nonempty, and set $\mathscr D = \mathscr O \times \mathbb R^d$. Then all the assertions  of Theorem~\ref{th.main} remain valid, verbatim, for the process \eqref{eq.Lan} killed upon exiting $\mathscr D$. Moreover, if $\mathbf{U}$ and $\Theta$ are $C^\infty$, these results extend to all $\alpha \in (0,1]$.
\end{theorem}

\begin{proof} 
We restrict ourselves to the low-regularity framework when  {\rm\textbf{[{\footnotesize{H$^{\alpha\in (1,2)}_{\text{P-Grad}}$}}]}} holds, the other case is treated with the same arguments.  
In view of the proof of  Theorem~\ref{th.main}, it is sufficient to check that all the conditions \textbf{C1}$\to$\textbf{C4} are fulfilled. 

We start with \textbf{C1}. 
Let us  prove that  $P_t^{\mathscr D}$ is   strongly Feller for $t>0$. 
Since $P_t$ is   strongly Feller for $t>0$ (see Theorem \ref{th.sta}), in view of the proof of Theorem~\ref{th.SFD}, it suffices to show \eqref{eq.Lpop}. Recall \eqref{eq.vars_R}. Let  $K\Subset \mathbb R^{2d}$ and $R>0$ such that $K\Subset \mathscr W_{p,R}$ (see \eqref{eq.WpR}). 
We have for every $\mathsf{x}=(x,v) \in K$, 
\begin{align*} 
&\mathbb{P}_{\mathsf{x}}[\sigma_{\mathscr{B}_{\mathbb{R}^d}(\mathsf{x}, \delta)} \le s]  \le \mathbb{P}_{\mathsf{x}}[\sigma_{\mathscr{B}_{\mathbb{R}^d}(\mathsf{x}, \delta)} \le s, s<\varsigma_R] + \mathbb{P}_{\mathsf{x}}[   \varsigma_R\le s]. 
\end{align*}
Recall  $X_{\cdot \wedge \varsigma_R}(\mathsf x)\overset{\text{law}}{=}X^R_{\cdot \wedge \zeta_R}(\mathsf x)$ and $\varsigma_R\overset{\text{law}}{=} \zeta_R$ (see \eqref{eq.zeta_R}). 
By \eqref{eq.lmz}, we thus get that 
$$ \lim_{R\to +\infty}\  \sup_{s\in [0,1]} \ \sup_{\mathsf x\in K} \  \mathbb{P}_{\mathsf{x}}[   \varsigma_R\le s] =0.$$
 Moreover, for $\mathsf x=(x,v)\in K$, 
 \begin{align*} 
 \mathbb{P}_{\mathsf{x}}[\sigma_{\mathscr{B}_{\mathbb{R}^d}(\mathsf{x}, \delta)} \le s, s<\varsigma_R]   &\le  \mathbb{P}_{\mathsf{x}}[ \sup_{u\in [0,s]} |x_u-x| \ge \delta, s<\varsigma_R]  \\
 &=    \mathbb{P}_{\mathsf{x}}[ \sup_{u\in [0,s]} |x^R_u-x| \ge \delta, s<\zeta_R]\\
 &\le  \mathbb{P}_{\mathsf{x}}[ \sup_{u\in [0,s]} |x^R_u-x| \ge \delta].
\end{align*}
Because $\mathbf B_R$ is bounded, we can use all the estimates we derived  in the proof of Theorem~\ref{th.SFD} and in particular by \eqref{eq.Gnj}, we have for all $R>0$,  
$$
\lim_{s\to 0}  \ \sup_{\mathsf x\in K} \ \mathbb{P}_{\mathsf{x}}\big [|x^R_s - x| \ge \delta/2\big ]=0.
 $$
 This achieves the proof of \eqref{eq.Lpop}. The operator $P_t^{\mathscr D}$ is  thus  strongly Feller. 
 This implies in particular \textbf{C1}. Let us now consider the other conditions.

 We have already shown in \textbf{Step 4a} in the proof of Theorem~\ref{th.sta} that  \textbf{C2} is satisfied  under  {\rm\textbf{[{\footnotesize{H$^{\alpha\in (1,2)}_{\text{P-Grad}}$}}]}}. It is also straightforward to check that  \textbf{C3} is satisfied because $\mathbb R^{d}\setminus \overline{\mathscr O}$ is nonempty (see the proof of Theorem~\ref{th.C5-infty}). 
By Proposition~\ref{pr.Rsp}, $\mathsf r_{sp}(P_t^{\mathscr D}|_{bB (\mathscr D)})>0$. Finally, with the same arguments as those employed in the proof of Theorem~\ref{th.bw} and since we can assume without loss of generality that \eqref{eq.Hgp} holds (see indeed \textbf{Step 4a} in the proof of Theorem~\ref{th.sta}),   for $t>0$,  $\mathsf r_{ess}(P_t^{\mathscr D}|_{b B(\mathscr D)})=0$. 
The proof of the theorem is complete. 
\end{proof}

 \bigskip

\noindent
{\small \textbf{Acknowledgement.} This work has been (partially) supported by the Project CONVIVIALITY ANR-23-CE40-0003 of the French National Research Agency. AG has benefited from a government grant managed by the Agence Nationale de la
Recherche under the France 2030 investment plan ANR-23-EXMA-0001.
 This research was also partially funded by the French National Research Agency (ANR) via Project ANR-25-CE40-6875-01 (ANR DySLoS).}

{\small
 \bibliography{alpha-stable-langevin} 

\begin{thebibliography}{10}

\bibitem{applebaum2009levy}
D.~Applebaum.
\newblock {\em L{\'e}vy {P}rocesses and {S}tochastic {C}alculus}.
\newblock Cambridge University Press, 2009.

\bibitem{bansaye2022non}
V.~Bansaye, B.~Cloez, P.~Gabriel, and A.~Marguet.
\newblock A non-conservative {H}arris ergodic theorem.
\newblock {\em Journal of the London Mathematical Society}, 106(3):2459--2510,
  2022.

\bibitem{bao2024exponential}
J.~Bao, R.~Fang, and J.~Wang.
\newblock Exponential ergodicity of {L}{\'e}vy driven {L}angevin dynamics with
  singular potentials.
\newblock {\em Stochastic Processes and their Applications}, 172:104341, 2024.

\bibitem{barczy2015yamada}
M.~Barczy, Z.~Li, and G.~Pap.
\newblock Yamada-{W}atanabe results for stochastic differential equations with
  jumps.
\newblock {\em International Journal of Stochastic Analysis}, 2015(1):460472,
  2015.

\bibitem{BenaimChampagnatEtAl2022}
M.~Bena{\"\i}m, N.~Champagnat, W.~O{\c{c}}afrain, and D.~Villemonais.
\newblock Quasi-compactness criterion for strong {F}eller kernels with an
  application to quasi-stationary distributions.
\newblock {\em Preprint hal-03640205v1}, 2022.

\bibitem{benaim2025degenerate}
M.~Bena{\"\i}m, N.~Champagnat, W.~O{\c{c}}afrain, and D.~Villemonais.
\newblock Degenerate processes killed at the boundary of a domain.
\newblock {\em The Annals of Probability}, 53(2):720--752, 2025.

\bibitem{billingsley2013}
P.~Billingsley.
\newblock {\em Convergence of {P}robability {M}easures}.
\newblock John Wiley \& Sons, 2013.

\bibitem{bouchaud}
J-P. Bouchaud and A.~Georges.
\newblock Anomalous {D}iffusion in {D}isordered {M}edia: {S}tatistical
  {M}echanisms, {M}odels and {P}hysical {A}pplications.
\newblock {\em Physics Reports}, 195(4-5):127--293, 1990.

\bibitem{brockmann2002levy}
D.~Brockmann and I.M. Sokolov.
\newblock L{\'e}vy flights in external force fields: from models to equations.
\newblock {\em Chemical Physics}, 284(1-2):409--421, 2002.

\bibitem{champagnat2018criteria}
N.~Champagnat, K.~A. Coulibaly-Pasquier, and D.~Villemonais.
\newblock Criteria for exponential convergence to quasi-stationary
  distributions and applications to multi-dimensional diffusions.
\newblock In {\em S{\'e}minaire de Probabilit{\'e}s XLIX}, pages 165--182.
  Springer, 2018.

\bibitem{champagnat2024quasi}
N.~Champagnat, T.~Leli{\`e}vre, M.~Ramil, J.~Reygner, and D.~Villemonais.
\newblock Quasi-stationary distribution for kinetic {SDE}s with low regularity
  coefficients.
\newblock {\em Preprint arXiv:2410.01042}, October 2024.

\bibitem{champagnat2016exponential}
N.~Champagnat and D.~Villemonais.
\newblock Exponential convergence to quasi-stationary distribution and
  {Q}-process.
\newblock {\em Probability Theory and Related Fields}, 164(1-2):243--283, 2016.

\bibitem{champagnat2021lyapunov}
N.~Champagnat and D.~Villemonais.
\newblock Lyapunov criteria for uniform convergence of conditional
  distributions of absorbed {M}arkov processes.
\newblock {\em Stochastic Processes and their Applications}, 135:51--74, 2021.

\bibitem{de2023multidimensional}
P-E. Chaudru De~Raynal, J-F. Jabir, and S.~Menozzi.
\newblock Multidimensional stable driven {M}ckean-{V}lasov {SDE}s with
  distributional interaction kernel: critical thresholds and related models.
\newblock {\em Preprint arXiv:2302.09900}, 2023.

\bibitem{chaudru}
P-E. Chaudru~de Raynal and S.~Menozzi.
\newblock Regularization effects of a noise propagating through a chain of
  differential equations: an almost sharp result.
\newblock {\em Transactions of the American Mathematical Society},
  375(1):1--45, 2022.

\bibitem{de2023heat}
P-E. Chaudru~de Raynal, S.~Menozzi, A.~Pesce, and X.~Zhang.
\newblock Heat kernel and gradient estimates for kinetic {SDE}s with low
  regularity coefficients.
\newblock {\em Bulletin des Sciences Math{\'e}matiques}, 183:103229, 2023.

\bibitem{chechkin2002stationary}
A.~Chechkin, V.~Gonchar, J.~Klafter, R.~Metzler, and L.~Tanatarov.
\newblock Stationary states of non-linear oscillators driven by {L}{\'e}vy
  noise.
\newblock {\em Chemical Physics}, 284(1-2):233--251, 2002.

\bibitem{chechkin2006fundamentals}
A.~V. Chechkin, V.~Y. Gonchar, J.~Klafter, and R.~Metzler.
\newblock Fundamentals of {L}{\'e}vy flight processes.
\newblock {\em Fractals, Diffusion, and Relaxation in Disordered Complex
  Systems: Advances in Chemical Physics, Part B}, pages 439--496, 2006.

\bibitem{chen2016uniqueness}
Z-Q. Chen and L.~Wang.
\newblock Uniqueness of stable processes with drift.
\newblock {\em Proceedings of the American Mathematical Society},
  144(6):2661--2675, 2016.

\bibitem{chung2001brownian}
K.L. Chung and Z.~Zhao.
\newblock {\em From Brownian {M}otion to Schr{\"o}dinger’s {E}quation},
  volume 312.
\newblock Springer Science \& Business Media, 2001.

\bibitem{coleman1992fractal}
P.~H. Coleman and L.~Pietronero.
\newblock The fractal structure of the universe.
\newblock {\em Physics Reports}, 213(6):311--389, 1992.

\bibitem{collet2011quasi}
P.~Collet, S.~Mart{\'\i}nez, S.~M{\'e}l{\'e}ard, and J.~San~Mart{\'\i}n.
\newblock Quasi-stationary distributions for structured birth and death
  processes with mutations.
\newblock {\em Probability Theory and Related Fields}, 151(1-2):191--231, 2011.

\bibitem{davis1993markov}
M.H.A Davis.
\newblock {\em Markov {M}odels \& {O}ptimization}, volume~49.
\newblock CRC Press, 1993.

\bibitem{del2023stability}
P.~Del~Moral, E.~Horton, and A.~Jasra.
\newblock On the stability of positive semigroups.
\newblock {\em The Annals of Applied Probability}, 33(6A):4424--4490, 2023.

\bibitem{di-gesu-lelievre-le-peutrec-nectoux-17}
G.~Di~Ges{\`u}, T.~Leli{\`e}vre, D.~Le~Peutrec, and B.~Nectoux.
\newblock Jump {M}arkov models and transition state theory: the
  quasi-stationary distribution approach.
\newblock {\em Faraday Discussions}, 195:469--495, 2017.

\bibitem{DLLN}
G.~Di~Ges\`u, T.~Leli\`evre, D.~Le~Peutrec, and B.~Nectoux.
\newblock Sharp asymptotics of the first exit point density.
\newblock {\em Annals of PDE}, 5(2), 2019.

\bibitem{dong2016strong}
Z.~Dong, X.~Peng, Y.~Song, and X.~Zhang.
\newblock Strong {F}eller properties for degenerate {SDE}s with jumps.
\newblock {\em Annales de l'IHP Probabilit{\'e}s et statistiques},
  52(2):888–--897, 2016.

\bibitem{down1995exponential}
D.~Down, S.~P. Meyn, and R.~L. Tweedie.
\newblock Exponential and uniform ergodicity of {M}arkov processes.
\newblock {\em The Annals of Probability}, 23(4):1671--1691, 1995.

\bibitem{eliazar}
I.~Eliazar and J.~Klafter.
\newblock L{\'e}vy-driven {L}angevin systems: {T}argeted stochasticity.
\newblock {\em Journal of statistical physics}, 111(3):739--768, 2003.

\bibitem{EK}
S.~N. Ethier and T.G. Kurtz.
\newblock {\em Markov {P}rocesses: {C}haracterization and {C}onvergence}.
\newblock John Wiley \& Sons, 1986.

\bibitem{ferrari-kesten-martinez-picco-95}
P.A. Ferrari, H.~Kesten, S.~Martinez, and P.~Picco.
\newblock Existence of quasi-stationary distributions. {A} renewal dynamical
  approach.
\newblock {\em The Annals of Probability}, 23(2):511--521, 1995.

\bibitem{Girsanov}
I.~V. Girsanov.
\newblock Strongly-{F}eller processes {I}. {G}eneral properties.
\newblock {\em Theory of Probability \& Its Applications}, 5(1):5--24, 1960.

\bibitem{gong1988killed}
G.~Gong, M.~Qian, and Z.~Zhao.
\newblock Killed diffusions and their conditioning.
\newblock {\em Probability Theory and Related Fields}, 80(1):151--167, 1988.

\bibitem{guillinFK}
A.~Guillin, D.~Lu, B.~Nectoux, and L.~Wu.
\newblock Long time behavior of killed {F}eynman-{K}ac semigroups with singular
  {S}chr\"{o}dinger potentials.
\newblock {\em Preprint Hal-04790621}, 2024.

\bibitem{guillinqsd3}
A.~Guillin, D.~Lu, B.~Nectoux, and L.~Wu.
\newblock Generalized {L}angevin and {N}os{\'e}-{H}oover processes absorbed at
  the boundary of a metastable domain.
\newblock {\em Preprint arXiv:2403.17471}, March 2024.

\bibitem{guillinqsd}
A.~Guillin, B.~Nectoux, and L.~Wu.
\newblock Quasi-stationary distribution for strongly {F}eller {M}arkov
  processes by {L}yapunov functions and applications to hypoelliptic
  {H}amiltonian systems.
\newblock {\em Journal of the European Mathematical Society}, 26(8):3047--3090,
  2022.

\bibitem{guillinqsd2}
A.~Guillin, B.~Nectoux, and L.~Wu.
\newblock Quasi-stationary distribution for {H}amiltonian dynamics with
  singular potentials.
\newblock {\em Probability Theory and Related Fields}, 185(3-4):921--959, 2023.

\bibitem{guillin2024large}
A.~Guillin, B.~Nectoux, and L.~Wu.
\newblock Large deviations of the empirical measures of a strong-{F}eller
  {M}arkov process inside a subset and quasi-ergodic distribution.
\newblock {\em Preprint arXiv:2411.17216}, 2024.

\bibitem{RocknerZhang}
Z.~Hao, M.~R{\"o}ckner, and X.~Zhang.
\newblock Second order fractional mean-field {SDE}s with singular kernels and
  measure initial data.
\newblock {\em ArXiv:2302.04392}, 2023.

\bibitem{hao2020schauder}
Z.~Hao, M.~Wu, and X.~Zhang.
\newblock Schauder estimates for nonlocal kinetic equations and applications.
\newblock {\em Journal de Math{\'e}matiques Pures et Appliqu{\'e}es},
  140:139--184, 2020.

\bibitem{he2019some}
G.~He, G.~Yang, and Y.~Zhu.
\newblock Some conditional limiting theorems for symmetric {M}arkov processes
  with tightness property.
\newblock {\em Electronic {C}ommunications in {P}robability}, 24(60):1--11,
  2019.

\bibitem{jacod2013limit}
J.~Jacod and A.~Shiryaev.
\newblock {\em Limit theorems for stochastic processes}, volume 288.
\newblock Springer Science \& Business Media, 2013.

\bibitem{kallenberg}
O.~Kallenberg.
\newblock {\em Foundations of {M}odern {P}robability}, volume~2.
\newblock Springer, 1997.

\bibitem{kolb2012quasilimiting}
M.~Kolb and D.~Steinsaltz.
\newblock Quasilimiting behavior for one-dimensional diffusions with killing.
\newblock {\em The Annals of Probability}, 40(1):162--212, 2012.

\bibitem{krylov2005strong}
N.~V. Krylov and M.~R{\"o}ckner.
\newblock Strong solutions of stochastic equations with singular time dependent
  drift.
\newblock {\em Probability theory and related fields}, 131(2):154--196, 2005.

\bibitem{krylov1969ito}
N.V. Krylov.
\newblock On {I}to's stochastic integral equations.
\newblock {\em Theory of Probability \& Its Applications}, 14(2):330--336,
  1969.

\bibitem{kulik2009}
A.M. Kulik.
\newblock Exponential ergodicity of the solutions to {SDE} with a jump noise.
\newblock {\em Stochastic Processes and their Applications}, 119(2):602--632,
  2009.

\bibitem{kulinich2014strong}
G~Kulinich and S.~Kushnirenko.
\newblock Strong uniqueness of solutions of stochastic differential equations
  with jumps and non-{L}ipschitz random coefficients.
\newblock {\em Modern Stochastics: Theory and Applications}, 1(1):65--72, 2014.

\bibitem{kulyk22}
O.~Kulyk.
\newblock Support theorem for {L}{\'e}vy-driven stochastic differential
  equations.
\newblock {\em Journal of Theoretical Probability}, pages 1--23, 2022.

\bibitem{kurtz2010equivalence}
T.~G. Kurtz.
\newblock Equivalence of stochastic equations and martingale problems.
\newblock In {\em Stochastic analysis 2010}, pages 113--130. Springer, 2010.

\bibitem{le2016brownian}
J-F. Le~Gall.
\newblock {\em Brownian {M}otion, {M}artingales, and {S}tochastic {C}alculus}.
\newblock Springer, 2016.

\bibitem{IHPLLN}
T.~Leli\`evre, D.~Le~Peutrec, and B.~Nectoux.
\newblock Exit event from a metastable state and {E}yring-{K}ramers law for the
  overdamped {L}angevin dynamics.
\newblock {\em In: Stochastic Dynamics out of Equilibrium, G. Giacomin, S.
  Olla, E. Saada, H. Spohn and G. Stoltz (Eds), {S}pringer {P}roceedings in
  {M}athematics \& {S}tatistics}, 2018.

\bibitem{ramilarxiv2}
T.~Leli{\`e}vre, M.~Ramil, and J.~Reygner.
\newblock Quasi-stationary distribution for the {L}angevin process in
  cylindrical domains, part {I}: existence, uniqueness and long-time
  convergence.
\newblock {\em Stochastic Processes and their Applications}, 144:173--201,
  2022.

\bibitem{lelievre2016partial}
T.~Leli{\`e}vre and G.~Stoltz.
\newblock Partial differential equations and stochastic methods in molecular
  dynamics.
\newblock {\em Acta Numerica}, 25:681--880, 2016.

\bibitem{Lepeltier}
J-P. Lepeltier and B.~Marchal.
\newblock Probl{\`e}me des martingales et {\'e}quations diff{\'e}rentielles
  stochastiques associ{\'e}es {\`a} un op{\'e}rateur
  int{\'e}gro-diff{\'e}rentiel.
\newblock In {\em Annales de l'institut Henri Poincar{\'e}. Section B.
  Probabilit{\'e}s et statistiques}, volume~12, pages 43--103, 1976.

\bibitem{liang2021exponential}
M.~Liang, M.~B. Majka, and J.~Wang.
\newblock Exponential ergodicity for {SDE}s and mckean--vlasov processes with
  {L}{\'e}vy noise.
\newblock In {\em Annales de l'Institut Henri Poincare (B) Probabilites et
  statistiques}, volume~57, pages 1665--1701. Institut Henri Poincar{\'e},
  2021.

\bibitem{marino}
L.~Marino and S.~Menozzi.
\newblock Weak well-posedness for a class of degenerate {L}{\'e}vy-driven
  {SDE}s with {H}{\"o}lder continuous coefficients.
\newblock {\em Stochastic Processes and their Applications}, 162:106--170,
  2023.

\bibitem{masuda2007ergodicity}
H.~Masuda.
\newblock Ergodicity and exponential $\beta$-mixing bounds for multidimensional
  diffusions with jumps.
\newblock {\em Stochastic processes and their applications}, 117(1):35--56,
  2007.

\bibitem{meleard2012quasi}
S.~M{\'e}l{\'e}ard and D.~Villemonais.
\newblock Quasi-stationary distributions and population processes.
\newblock {\em Probability Surveys}, 9:340--410, 2012.

\bibitem{mel1983stochastic}
A.~V. Mel'nikov.
\newblock Stochastic equations and {K}rylov's estimates for semimartingales.
\newblock {\em Stochastics: An International Journal of Probability and
  Stochastic Processes}, 10(2):81--102, 1983.

\bibitem{metzler2000random}
R.~Metzler and J.~Klafter.
\newblock The random walk's guide to anomalous diffusion: a fractional dynamics
  approach.
\newblock {\em Physics reports}, 339(1):1--77, 2000.

\bibitem{Nummelin}
E.~Nummelin.
\newblock {\em General {I}rreducible {M}arkov {C}hains and {N}on-{N}egative
  {O}perators}, volume~83.
\newblock Cambridge University Press, 2004.

\bibitem{peszat1995strong}
S.~Peszat and J.~Zabczyk.
\newblock Strong feller property and irreducibility for diffusions on {H}ilbert
  spaces.
\newblock {\em The Annals of Probability}, pages 157--172, 1995.

\bibitem{pinsky1985convergence}
R.G. Pinsky.
\newblock On {T}he {C}onvergence of {D}iffusion {P}rocesses {C}onditioned to
  {R}emain in a {B}ounded {R}egion for {L}arge {T}ime to {L}imiting {P}ositive
  {R}ecurrent {D}iffusion {P}rocesses.
\newblock {\em The Annals of Probability}, 13(2):363--378, 1985.

\bibitem{priola2009densities}
E.~Priola and J.~Zabczyk.
\newblock Densities for {O}rnstein--{U}hlenbeck processes with jumps.
\newblock {\em Bulletin of the London Mathematical Society}, 41(1):41--50,
  2009.

\bibitem{priola2012pathwise}
Enrico Priola.
\newblock Pathwise uniqueness for singular sdes driven by stable processes.
\newblock {\em Osaka Journal of Mathematics}, 49(2):421--447, 2012.

\bibitem{saichev1997fractional}
A.~I. Saichev and G.~. Zaslavsky.
\newblock Fractional kinetic equations: solutions and applications.
\newblock {\em Chaos: an Interdisciplinary Journal of Nonlinear Science},
  7(4):753--764, 1997.

\bibitem{sandric2016ergodicity}
N.~Sandri{\'c}.
\newblock Ergodicity of {L}{\'e}vy-type processes.
\newblock {\em ESAIM: Probability and Statistics}, 20:154--177, 2016.

\bibitem{schilling2012strong}
R.~L. Schilling and J.~Wang.
\newblock Strong {F}eller continuity of {F}eller processes and semigroups.
\newblock {\em Infinite Dimensional Analysis, Quantum Probability and Related
  Topics}, 15(02):1250010, 2012.

\bibitem{situ2005theory}
R.~Situ.
\newblock {\em Theory of {S}tochastic {D}ifferential {E}quations with {J}umps
  and {A}pplications: {M}athematical and {A}nalytical {T}echniques with
  {A}pplications to {E}ngineering}.
\newblock Springer, 2005.

\bibitem{steinsaltz-evans-07}
D.~Steinsaltz and S.N. Evans.
\newblock Quasi-stationary distributions for one-dimensional diffusions with
  killing.
\newblock {\em Transactions of the {A}merican {M}athematical {S}ociety},
  359(3):1285--1324, 2007.

\bibitem{stroock1975diffusion}
D.~W. Stroock.
\newblock Diffusion processes associated with {L}{\'e}vy generators.
\newblock {\em journal={Probability theory and related fields}},
  32(3):209--244, 1975.

\bibitem{stroock2007multidimensional}
D.~W. Stroock and S.R.~S. Varadhan.
\newblock {\em Multidimensional {D}iffusion {P}rocesses}.
\newblock Springer, 2007.

\bibitem{takeda1}
M.~Takeda.
\newblock Existence and uniqueness of quasi-stationary distributions for
  symmetric {M}arkov processes with tightness property.
\newblock {\em Journal of Theoretical Probability}, 32(4):2006--2019, 2019.

\bibitem{tanaka1974perturbation}
H.~Tanaka, M.~Tsuchiya, and S.~Watanabe.
\newblock Perturbation of drift-type for {L}{\'e}vy processes.
\newblock {\em Journal of Mathematics of Kyoto University}, 14(1):73--92, 1974.

\bibitem{veretennikov1981strong}
A.~J. Veretennikov.
\newblock On strong solutions and explicit formulas forsolutions of stochastic
  integral equations.
\newblock {\em Mathematics of the USSR-Sbornik}, 39(3):387, 1981.

\bibitem{villemonais2025quasi}
Denis Villemonais.
\newblock Quasi-compactness for dominated kernels with application to
  quasi-stationary distribution theory.
\newblock {\em Preprint arXiv:2510.19573}, 2025.

\bibitem{wu1999}
L.~Wu.
\newblock Uniqueness of {N}elsons diffusions.
\newblock {\em Probability {T}heory and {R}elated {F}ields}, 114(4):549--585,
  1999.

\bibitem{Wu2001}
L.~Wu.
\newblock Large and moderate deviations and exponential convergence for
  stochastic damping {H}amiltonian systems.
\newblock {\em Stochastic Processes and their Applications}, 91(2):205--238,
  2001.

\bibitem{Wu2004}
L.~Wu.
\newblock Essential spectral radius for {M}arkov semigroups. {I}. {D}iscrete
  time case.
\newblock {\em Probability Theory and Related Fields}, 128(2):255--321, 2004.

\bibitem{zhangAIHP2020}
L.~Xie and X.~Zhang.
\newblock Ergodicity of stochastic differential equations with jumps and
  singular coefficients.
\newblock {\em Annales de l'IHP Probabilit{\'e}s et statistiques},
  56(1):175--229, 2020.

\bibitem{yanovsky2000levy}
V.V. Yanovsky, D.~Chechkin, A.V.and~Schertzer, and A.V. Tur.
\newblock L{\'e}vy anomalous diffusion and fractional {F}okker--{P}lanck
  equation.
\newblock {\em Physica A: Statistical Mechanics and its Applications},
  282(1-2):13--34, 2000.

\bibitem{zhang2014quasi}
J.~Zhang, S.~Li, and R.~Song.
\newblock Quasi-stationarity and quasi-ergodicity of general {M}arkov
  processes.
\newblock {\em Science China Mathematics}, 57:2013--2024, 2014.

\bibitem{zhang2014densities}
X.~Zhang.
\newblock Densities for {SDE}s driven by degenerate $\alpha$-stable processes.
\newblock {\em The Annals of Probability}, 42(5):1885--1910, 2014.

\bibitem{zhang2014sima}
X.~Zhang.
\newblock Fundamental solution of kinetic {F}okker--{P}lanck operator with
  anisotropic nonlocal dissipativity.
\newblock {\em SIAM Journal on Mathematical Analysis}, 46(3):2254--2280, 2014.

\bibitem{zhang2017fundamental}
X.~Zhang.
\newblock Fundamental solutions of nonlocal {H}ormander's operators {II}.
\newblock {\em The Annals of Probability}, 45(3):1799--1841, 2017.

\bibitem{zvonkin1974transformation}
A.K. Zvonkin.
\newblock A transformation of the phase space of adiffusion process that
  removes the drift.
\newblock {\em Mathematics of the USSR-Sbornik}, 22(1):129, 1974.

\end{thebibliography}

\bibliographystyle{plain}

}

\end{document}